\documentclass[11pt]{article}
\usepackage[margin=1in]{geometry} 
\usepackage{amsmath,amsthm,amssymb,amsfonts,titling,bbm, titlesec, bm, physics, enumitem, accents, xcolor, setspace, fancyhdr, natbib, geometry,
	pdflscape}
\usepackage{graphicx}
\usepackage{float}
\usepackage[font=footnotesize,labelfont=bf]{caption}
\usepackage{array}
\usepackage[title]{appendix}
\usepackage{afterpage}
\usepackage{listings}
\usepackage{prodint}
\usepackage{algorithm}
\usepackage{algpseudocode}
\usepackage{bm}
\usepackage{lineno}
\usepackage{booktabs}
\usepackage[hidelinks]{hyperref}
\usepackage[smallproofs=false, % so that proofs aren't small, as is default
qedsymbolblacksquare=false, % so that QED symbol is empty
prooftitleitbf=true, % so that proof titles are bold and italicized
sepcounters=true]{phfthm} % so that counters are separate for diff envs

\newcommand{\R}{\mathbb{R}}

\newcommand{\I}{\mathbbm{1}}

\newcommand{\argmax}{\text{argmax}}

\newcommand{\pr}{\textrm{pr}}
\newcommand{\Var}{\textrm{var}}

\newcommand{\midd}{\,|\,}

\newlength{\continueindent}
\makeatletter
\newcommand*{\ALG@customparshape}{\parshape 2 \leftmargin \linewidth \dimexpr\ALG@tlm+\continueindent\relax \dimexpr\linewidth+\leftmargin-\ALG@tlm-\continueindent\relax}
\apptocmd{\ALG@beginblock}{\ALG@customparshape}{}{\errmessage{failed to patch}}
\makeatother

\titleformat{\section}{\large\scshape\bfseries}{\thesection.}{1em}{}
\titleformat{\subsection}{\normalfont\bfseries}{\thesubsection.}{1em}{}

\DeclareMathAlphabet{\mathbbmsl}{U}{bbm}{m}{sl}

\allowdisplaybreaks

\newtheorem{example}{Example}
\newtheorem{step}{Step}
\newtheorem{condition}{Condition}

\begin{document}
	
	\begin{center}
		\large{\textbf{Assessing variable importance in survival analysis using machine learning}}\vspace{0.3cm}
		
		\normalsize{Charles J. Wolock\textsuperscript{1}, Peter B. Gilbert\textsuperscript{3,2}, Noah Simon\textsuperscript{2} \& Marco Carone\textsuperscript{2,3}}
		\\ \vspace{0.3cm}
		\singlespacing
		\footnotesize{\textsuperscript{1}Department of Biostatistics, Epidemiology, and Informatics, University of Pennsylvania\\
			\textsuperscript{2}Department of Biostatistics, University of Washington\\
			\textsuperscript{3}Vaccine and Infectious Disease Division, Fred Hutchinson Cancer Center}
	\end{center}
	
	\vspace{0.3cm}
	\singlespacing
	\begin{center}
		\textbf{Abstract}
	\end{center}
	\begin{quote}
		Given a collection of features available for inclusion in a predictive model, it may be of interest to quantify the relative importance of a subset of features for the prediction task at hand. For example, in HIV vaccine trials, participant baseline characteristics are used to predict the probability of HIV acquisition over the intended follow-up period, and investigators may wish to understand how much certain types of predictors, such as behavioral factors, contribute toward overall predictiveness. Time-to-event outcomes such as time to HIV acquisition are often subject to right censoring, and existing methods for assessing variable importance are typically not intended to be used in this setting. We describe a broad class of algorithm-agnostic variable importance measures for prediction in the context of survival data. We propose a nonparametric efficient estimation procedure that incorporates flexible learning of nuisance parameters, yields asymptotically valid inference, and enjoys double-robustness. We assess the performance of our proposed procedure via numerical simulations and analyze data from the HVTN 702 vaccine trial to inform enrollment strategies for future HIV vaccine trials. \vspace{0.1in}
		
	\end{quote}
	\doublespacing
	
	\section{Introduction}
	
	Variable importance is often defined as a quantification of the contribution that a feature or group of features makes to predicting an outcome. As black-box machine learning prediction algorithms have become popular in scientific applications, such as biomedical research, there has been corresponding growth of a literature focused on defining and making statistical inference on variable importance. Existing methods for assessing variable importance can be broadly categorized as either describing how a given algorithm uses features in making predictions, or as measuring the population-level predictiveness potential of a set of features. We refer to these types of variable importance as \textit{extrinsic} and \textit{intrinsic}, respectively \citep{williamson2021rejoinder}.
	
	The question of variable importance arises naturally in the design of HIV vaccine trials. To achieve desired statistical power in efficacy trials, investigators often aim to recruit individuals thought to have a high probability of HIV acquisition during the intended follow-up period. This process relies in part on risk prediction models for HIV acquisition, which 
	have been developed for various populations; examples include \citet{Menza2009}, \citet{Smith2012},  \citet{Balkus2016}, and \citet{Wand2018}. A variety of features, including demographic information (e.g., age), laboratory assay readouts (e.g., diagnosed prevalent sexually transmitted infections), and behavioral questionnaire responses (e.g., number of sex partners over a specified time frame), have been found to be associated with time to HIV acquisition. Understanding the relative predictive importance of these features may inform data collection for future risk model development, which in turn may influence participant recruitment practices. Highly important variables may be prioritized for data collection, while less important variables may be deemed too onerous or expensive to merit collection. As such, reliable statistical estimation and quantification of uncertainty are paramount.
	
	Variable importance typically relies on the concept of a predictiveness measure, which is meant to quantify the performance of a prediction model. Measuring predictiveness in the setting of HIV vaccine trials, where acquisition risk is often studied using data from previous trials, is complicated by right censoring, a ubiquitous characteristic of these trials and other prospective biomedical studies of time-to-event outcomes. In the presence of censoring, extra care must be taken in evaluating the predictiveness of a given model, since outcomes and predictions cannot be directly compared. Accounting for right-censored data is particularly difficult when the censoring is informative, i.e., when the event and censoring times are not independent. In the realm of survival analysis, many widely used methods for evaluating predictiveness do not account for censoring and may converge to population parameters that (undesirably) depend on the censoring mechanism. Those that do account for censoring often require strong assumptions on the censoring mechanism or event time distribution. For example, many predictiveness estimation methods are based on the semiparametric Cox proportional hazards model \citep{Cox1972}, which, while convenient, raises the concern of model misspecification in practice. 
	
	The growing availability of flexible prediction models in survival analysis motivates the development of correspondingly flexible methods for evaluating variable importance. Extrinsic variable importance methods are commonly used for time-to-event data. Local importance methods, such as Shapley additive explanations \citep{Krzyzinski2023survshap} and local interpretable model-agnostic explanations \citep{Kovalev2020survlime}, are designed to measure the extent to which the prediction for a particular observation in the dataset depends on a feature of interest. Permutation-based methods quantify the change in model performance when predictor values are randomly shuffled \citep{Breiman2001, fisher2019} and have been applied in time-to-event settings \citep{spytek2023survex}. Notably, these extrinsic methods are generally construed as evaluating a fixed prediction algorithm rather than the population-level predictiveness potential of the available features. Methods for statistical inference on extrinsic importance are available for some algorithms, such as random forests \citep{Ishwaran2019RFinference}, but general inferential tools that handle informative right censoring are lacking.
	
	Meanwhile, by and large, the literature on model-free, algorithm-agnostic intrinsic variable importance, including the plug-in framework of \citet{Williamson2021b} and related approaches (e.g., \citealp{Dai2022, Zhang2022, Verdinelli2021}) has not treated survival data, and these methods cannot be directly applied to right-censored data. When right censoring is present, nonparametric identification of a functional of the time-to-event distribution generally involves complex nuisance parameters \citep{vanderLaan2003}; using flexible estimators of these nuisance parameters then necessitates debiasing \citep{Bickel1998, van2011targeted}. Except for the recent work of \citet{Boileau2023}, which is specifically focused on variable importance measures (VIMs) for evaluating treatment effect modification, general variable importance approaches do not address this complexity.
	
	In this paper, we present a general framework for estimating and making inference on predictiveness and variable importance using right-censored time-to-event data. We focus on intrinsic variable importance, which is a model-free and algorithm-agnostic quantity. Our approach accommodates censoring that is possibly informed by measured covariates. We consider a broad class of predictiveness measures for time-to-event outcomes, encompassing many measures used in practice. Our specific contributions include: 
	\begin{enumerate}
		\item[(i)] providing conditions under which measures in this class are identified in terms of the observed data distribution;
		\item[(ii)] deriving the nonparametric efficient influence function for these measures;
		\item[(iii)] devising a debiased machine learning procedure for estimation and inference.
	\end{enumerate}
	These developments have allowed us to analyze variable importance for predicting time to HIV acquisition diagnosis using data from HVTN 702, a phase 2b-3 trial conducted in South Africa between 2016 and 2020 to investigate the safety and efficacy of a canarypox-protein HIV vaccine regimen \citep{Gray2021}; we provide details of the analysis in this paper. The proposed methods are implemented in the \texttt{survML} package.
	
	\section{Problem setup}\label{sec:VIM review}
	
	\subsection{Data structure and notation}
	We observe a vector $X$ of baseline covariates taking values in $\mathcal{X} \subset \R^p$. The outcome of interest is the time elapsed between an initiating event (e.g., randomization into the HVTN 702 study) and a terminating event (e.g., diagnosis of HIV acquisition); we refer to this as the event time and denote it by $T \in (0, \infty)$. We observe the event time subject to right censoring, by which follow-up on a participant may conclude before the participant has experienced the event, potentially due to loss to follow-up or termination of the study. Let $C \in (0, \infty)$ denote the time between the initiating and censoring events. The ideal data unit can be considered to be $\mathbbmsl{Z} := (X,T, C)$, sampled from distribution $\mathbbmsl{P}_0$. For each participant, we observe the follow-up time $Y := \min\{T, C\}$ and the event indicator $\Delta := \I(T \leq C)$, resulting in the observed data unit $Z := (X, Y, \Delta)$. Our sample consists of $n$ independent and identically distributed observations $Z_1, \dots, Z_n$ drawn from $P_0$, the observed data distribution implied by $\mathbbmsl{P}_0$. We denote by $\mathcal{M}_{\text{ideal}}$ the ideal data model in which $\mathbbmsl{P}_0$ lies and by $\mathcal{M}_{\text{obs}}$ the observed data model in which $P_0$ lies. Throughout the article, we use the italicized blackboard bold font (e.g., $\mathbbmsl{Z}$) to emphasize whenever objects correspond to the ideal data setting.  
	
	We denote by $\mathbb{P}_n$ the empirical measure of $Z_1, \dots, Z_n$. We use $a \wedge b$ to denote $\min\{a,b\}$. For vectors $u = (u_1,\dots,u_p)$ and $v = (v_1,\dots,v_p)$, we take inequalities to be component-wise, i.e., $\{u \leq v\} = \{u_1 \leq v_1, \dots, u_p \leq v_p\}.$ We use $s \subset \{1,\dots, p\}$ to denote the index set of a covariate subgroup. For any vector $v$, we use $v_s$ to denote the elements of $v$ with index in $s$ and $v_{-s}$ the elements of $v$ with index not in $s$. 
	
	\subsection{Predictiveness and variable importance}
	
	In this section we give a brief overview of predictiveness and variable importance. We consider a class $\mathcal{F}$ of potential prediction functions, mapping from $\mathcal{X}$ to a context-specific codomain $\mathcal{Y}$ depending on the predictiveness measure being studied. The subset $\mathcal{F}_s := \{f \in \mathcal{F}: f(u) = f(u^*) \text{ for all } u,u^* \in \mathcal{X} \text{ satisfying } u_{-s} = u^*_{-s}\}$ characterizes prediction functions in $\mathcal{F}$ that ignore features with index in $s$. We allow $\mathcal{F}$ to be largely unrestricted up to regularity conditions. We define $\mathbbmsl{V}(f, \mathbbmsl{P}_0)$ to be the ideal data predictiveness measure, quantifying how predictive $f \in \mathcal{F}$ is under $\mathbbmsl{P}_0$, with larger values indicating higher predictiveness. The oracle prediction function $\mathbbmsl{f}_0$ is then defined by $\mathbbmsl{f}_0 \in \argmax_{f \in \mathcal{F}} \ \mathbbmsl{V}(f, \mathbbmsl{P}_0)$.
	The oracle prediction function represents the optimal prediction function possible under $\mathbbmsl{P}_0$, as measured by $\mathbbmsl{V}$. If $\mathcal{F}$ is sufficiently rich, the oracle should not depend on the choice of function class. The oracle predictiveness, a measure of the combined predictive potential of $X$ under $\mathbbmsl{P}_0$, is then defined as $\mathbbmsl{V}(\mathbbmsl{f}_0, \mathbbmsl{P}_0)$. We analogously define the residual oracle predictiveness of $X_s$ as $\mathbbmsl{V}(\mathbbmsl{f}_{0,s}, \mathbbmsl{P}_0)$, where $\mathbbmsl{f}_{0,s} \in \argmax_{f \in \mathcal{F}_{s}}\mathbbmsl{V}(f, \mathbbmsl{P}_0)$ is the residual oracle prediction function. This quantifies the combined predictive potential of $X_{-s}$. 
	
	For $r \subset s$ a strict subset of $s$, we define the intrinsic \textit{exclusion} importance of $X_{s\setminus r}$ relative to $X_{-r}$ as $\mathbbmsl{V}(\mathbbmsl{f}_{0,r}, \mathbbmsl{P}_0) - \mathbbmsl{V}(\mathbbmsl{f}_{0,s}, \mathbbmsl{P}_0)$, that is, the decrease in oracle predictiveness when $s$ is excluded compared to when only $r$ is excluded. For example, if $X = (X_1, X_2, X_3)$, then setting $r = \emptyset$ and $s = \{1,2\}$ measures the importance of $(X_1, X_2)$ relative to $(X_1, X_2, X_3)$, while setting $r = \{1\}$ and $s = \{1,2\}$ measures the importance of $X_2$ relative to $(X_2, X_3)$. Without loss of generality, we focus on $\mathbbmsl{V}(\mathbbmsl{f}_{0}, \mathbbmsl{P}_0) - \mathbbmsl{V}(\mathbbmsl{f}_{0,s}, \mathbbmsl{P}_0)$, the intrinsic importance of $X_s$ relative to the full covariate vector $X$. Because the censoring mechanism may depend on covariates, it may be of interest to include features in the full vector $X$ without considering their predictive potential. Such variables should be included in the index set $r$.
	
	Although permutation importance is typically considered in the context of a fixed prediction model, it may also be defined as an intrinsic characteristic of $\mathbbmsl{P}_0$; see, for example, \citet{hooker2021permutation}. For an index set $s$, we use $\mathbbmsl{P}_{0,s}$ to denote the distribution of $(X_{s}^{(2)}, X_{-s}^{(1)}, T^{(1)}, C^{(1)})$, where $(X^{(1)}, T^{(1)}, C^{(1)})$ and $(X^{(2)}, T^{(2)}, C^{(2)})$ are drawn independently from $\mathbbmsl{P}_0$. The intrinsic \textit{permutation importance} of $X_{s \setminus r}$ relative to $X_{-r}$ is $\mathbbmsl{V}(\mathbbmsl{f}_0, \mathbbmsl{P}_{0, r}) - \mathbbmsl{V}(\mathbbmsl{f}_0, \mathbbmsl{P}_{0, s})$; this quantifies the decrease in the predictiveness of the oracle $\mathbbmsl{f}_0$ when $s$ is permuted compared to when only $r$ is permuted. In the Supplementary Material, we discuss the relationship between permutation and exclusion importance and examine their differences in a simulation study. However, exclusion importance is the primary focus of this work. 
	
	\subsection{Common predictiveness measures in survival analysis}\label{subsec:common vims}
	
	The choice of predictiveness measure should depend on the purpose of the prediction function $f$. For example, as in our scientific application, when $f$ is a risk score intended to stratify participants into risk categories, an appropriate measure may quantify how well $f$ discriminates between high-risk and low-risk participants. Our framework is broadly applicable, but the practitioner must choose a relevant predictiveness measure. We give several examples of a predictiveness measure below. For a discussion of predictiveness measures for survival data, see \citet{Korn1990}. 
	
	\begin{example}[\sc{AUC}]
		\upshape
		\citet{Heagerty2005} describe several extensions of the receiver operating characteristic curve for a time-to-event outcome. For example, they propose the use of the area under the \textit{cumulative/dynamic} receiver operating characteristic curve (AUC) at landmark time $\tau$, defined as    
		\begin{align*}
			\mathbbmsl{V}(f,\mathbbmsl{P}_0) := \pr_{\mathbbmsl{P}_0}\{f(X_1) > f(X_2) \midd T_1 \leq \tau, T_2 > \tau\}\ ,
		\end{align*}
		where $(X_1, T_1)$ and $(X_2, T_2)$ are independent draws from $\mathbbmsl{P}_0$. This measure corresponds to the probability that a participant who has experienced an event at or before time $\tau$ has a higher value of risk score $f$ compared to a participant who has not. Proposed methods for estimating AUC\ using right-censored data have been based, for example, on inverse-probability-of-censoring weights \citep{Uno2007, Hung2010}, a modified Kaplan-Meier approach \citep{Chambless2006}, and  an assumption of proportional hazards  \citep{Song2008}. 
	\end{example}
	
	\begin{example}[\sc{Brier score}]
		\upshape
		It may be of interest to predict whether a participant will remain event-free by landmark time $\tau$. Here, the predictiveness of $f$ can be quantified using a loss function, say $L: \mathcal{F} \times \{0,1\} \to [0, \infty)$, for predicting the binary outcome $\I(T > \tau)$. The predictiveness measure can then be defined as the negative risk $-E_{\mathbbmsl{P}_0}\left[L\{f(X), \I(T > \tau)\}\right]$. Common loss functions include the log, binary classification, and squared error losses. The mean squared error (MSE) for predicting a binary outcome is referred to as the Brier score \citep{Brier1950} and is often estimated using an inverse-probability-of-censoring weighting approach \citep{Gerds2006}. Here, we focus on the negative Brier score
		\begin{align*}
			\mathbbmsl{V}(f, \mathbbmsl{P}_0) := -E_{\mathbbmsl{P}_0}\left[\left\{f(X)- \I(T > \tau)\right\}^2\right].
		\end{align*}
		
	\end{example}
	
	\begin{example}[\sc{Survival time MSE}]
		\upshape
		When $f$ is intended to predict a participant's $\tau$--restricted survival time $T \wedge \tau$, its predictiveness can be evaluated using a loss $L: \mathcal{F}\times(0,\tau] \to [0, \infty)$ and corresponding negative risk $-E_{\mathbbmsl{P}_0}\left[L\{f(X), T \wedge \tau\}\right]$. We focus on the negative MSE\  
		\begin{align*}
			\mathbbmsl{V}(f, \mathbbmsl{P}_0) := -E_{\mathbbmsl{P}_0}\left[\left\{f(X)- (T \wedge \tau)\right\}^2\right].
		\end{align*}
	\end{example}
	
	\begin{example}[\sc{C-index}]
		\upshape
		The concordance index (C-index) is often considered to be a global measure of discriminative performance. The population C-index $\pr_{\mathbbmsl{P}_0}(f(X_1) > f(X_2) \midd T_1 < T_2)$ corresponds to the probability that, for a randomly selected pair of participants, the participant who experiences the event earlier has the higher value of the risk score $f$. For identifiability under right censoring, we restrict this definition to participant pairs whose earlier event time falls before a user-specified time $\tau$, defining
		\begin{align*}
			\mathbbmsl{V}(f, \mathbbmsl{P}_0) := \pr_{\mathbbmsl{P}_0}\{f(X_1) > f(X_2) \midd T_1 < T_2, T_1 \leq \tau\}\ .
		\end{align*} 
		Existing approaches for estimating the C-index include Harrell's C-index statistic \citep{Harrell1982}, as well as methods based on the Cox model \citep{Gonen2005}, the Pareto distribution \citep{Brentnall2018}, the use of inverse-probability-of-censoring weights \citep{Uno2011}, and the stratified Kaplan-Meier estimator \citep{Efron1967}, among others. 
	\end{example}
	
	\section{Adapting VIMs for survival analysis}\label{sec:survival VIMs}
	
	\subsection{Identification and support}
	The predictiveness measures and corresponding VIMs in Section \ref{sec:VIM review} are expressed in terms of the ideal data distribution $\mathbbmsl{P}_0$. Therefore, before proceeding we must identify them in terms of the observed data distribution $P_0$. To do so, we make use of the following condition: 
	\begin{condition}\label{condition:independent censoring}
		\upshape
		$T$ and $C$ are conditionally independent given $X$.
	\end{condition}
	Condition \ref{condition:independent censoring} allows informative censoring, as long as the censoring mechanism is uninformative within strata defined by the vector $X$ of recorded covariates. In light of Condition \ref{condition:independent censoring}, in addition to features of interest for prediction, it is necessary to include any covariate thought to inform the censoring mechanism in the full feature vector $X$. Conditionally independent censoring is a form of coarsening-at-random, and so, the model $\mathcal{M}_{\text{obs}}$ is nonparametric  if $\mathcal{M}_{\text{ideal}}$ is otherwise unrestricted \citep{vanderLaan2003}. 
	
	As is typical in handling right censoring, our identification approach is based on the hazard function. We define as $\mathbbmsl{L}_0(t \midd x) := \int_0^t \,\{1 - \mathbbmsl{F}_0(u^- \midd x)\}^{-1}\mathbbmsl{F}_0(du \midd x)$ the conditional cumulative hazard of $T$ given $X=x$ at $t$, where we write $\mathbbmsl{F}_0(t \midd x) := \pr_{\mathbbmsl{P}_0}(T \leq t \midd X = x)$. The product integral mapping $\alpha  \mapsto  \prodi_{(0, t]}\left\{1 - \alpha(du\midd x)\right\}$, which links the hazard and distribution functions \citep{Gill1990}, then allows us to write
	\begin{align*}
		\mathbbmsl{F}_0(t \midd x) = 1-\Prodi_{(0, t]}\left\{1 - \mathbbmsl{L}_0(du \midd x)\right\}.
	\end{align*}
	Next, we define the observed data hazard $L_0(t \midd x) := \int_0^t \,\{1 - M_0(u^- \midd x)\}^{-1}M_{0,1}(du \midd x)$, where $M_{0,1}(t \midd x) := \pr_{P_0}(Y \leq t, \Delta = 1 \midd X = x)$  and $M_0(t \midd x) := \pr_{P_0}(Y \leq t \midd X = x)$. We also define the product integral mapping $F_0(t \midd x) := 1-\prodi_{(0,t]}\{1 - L_0(du \midd x)\}$. 
	
	In this article, we focus on a class of predictiveness measures that can be written as an expectation taken with respect to the joint distribution of $(X$, $T)$, as defined below. From here and on, $\mathbbmsl{H}_0(x_0,t_0)$ denotes the joint distribution function of $(X, T)$ evaluated at $(x_0,t_0)$, and with a slight abuse of notation, we write $\mathbbmsl{V}(f, \mathbbmsl{P}_0)$ as $\mathbbmsl{V}(f, \mathbbmsl{H}_0)$.
	\begin{definition}
		An ideal data predictiveness measure $\mathbbmsl{V}(f, \mathbbmsl{H}_0)$ is called a standardized survival V--measure if it can be written in the form $
		\mathbbmsl{V}(f, \mathbbmsl{H}_0) = \mathbbmsl{V}_1(f, \mathbbmsl{H}_0)/\mathbbmsl{V}_2(\mathbbmsl{H}_0)$
		with
		\begin{align*}
			\mathbbmsl{V}_1(f, \mathbbmsl{H}_0) &= \int \dots \int \omega[\{f(x_1), t_1\}, \dots, \{f(x_m), t_m\}]\prod_{j=1}^m \mathbbmsl{H}_0(dx_j, dt_j)\ , \\
			\mathbbmsl{V}_2(\mathbbmsl{H}_0) &= \int \dots \int \theta(t_1, \dots, t_m)\prod_{j=1}^m \mathbbmsl{H}_0(dx_j, dt_j)
		\end{align*}
		for symmetric kernel functions $\omega: \left\{\mathcal{Y} \times \R\right\}^m \to \R$ and $\theta: \R^m \to \R$ and $m \geq 1$ an integer.
	\end{definition}
	A more general definition could allow $\mathbbmsl{V}$ to have an additive dependence on some constant $\kappa \in \R$, which would encompass predictiveness measures such as the proportion of explained variance \citep{Schemper2000}. For variable importance, interest lies in the difference between full and residual oracle predictiveness, and the constant $\kappa$ simply cancels.
	
	The joint distribution function of $(X,T)$ evaluated at $(x_0,t_0)$ can be written in the form $\mathbbmsl{H}_0(x_0, t_0) = \int \I(u \leq x_0)\mathbbmsl{F}_0(t_0 \midd u)\mathbbmsl{Q}_0(du)$, where $\mathbbmsl{Q}_0$ is the distribution function of $X$ under $\mathbbmsl{P}_0$. Letting $Q_0$ denote the distribution function of $X$ under $P_0$, we note that $\mathbbmsl{Q}_0 = Q_0$ since $X$ is fully observed. Identification of $\mathbbmsl{H}_0(x_0,t_0)$ then depends on identification of $\mathbbmsl{F}_0(t \midd x)$ for relevant $(x,t)$ values. Due to right censoring, this quantity is unlikely to be identified in the right tail of the time-to-event support, that is, for large $t$ values. Therefore, whether $\mathbbmsl{V}(\mathbbmsl{f}_0, \mathbbmsl{H}_0)$ is identified depends on the kernel functions $\omega$ and $\theta$. We use the following condition, along with Condition \ref{condition:independent censoring}, to ensure identification. 
	
	\begin{condition}\label{condition:region of integration}
		\upshape
		There exists some $\tau_0 \in (0,\infty)$ such that: (a) $\pr_{\mathbbmsl{P}_0}(C > \tau_0 \midd X) > 0$ $P_0$--almost surely; and (b) for any fixed $(x_1,\dots,x_m)$ and $(t_2,\dots,t_m)$, $t_1 \mapsto \omega[\{f(x_1),t_1\},\dots,\{f(x_m), t_m\}]$ and $t_1 \mapsto \theta(t_1,\dots,t_m)$ are constant over $(\tau_0,\infty)$.
	\end{condition}
	
	Before stating the identification result, we define $H_0(x_0,t_0) := \int \I(u \leq x_0)F_0(t_0 \midd u) Q_0 (du)$ and $V(f, H_0) := V_1(f, H_0)/V_2(H_0)$ with
	\begin{align*}
		V_1(f, H_0) &:= \int \dots \int \omega[\{f(x_1), t_1\}, \dots, \{f(x_m), t_m\}]\prod_{j=1}^m H_0(dx_j, dt_j)\ ; \\
		V_2(H_0) &:= \int \dots \int \theta(t_1, \dots, t_m)\prod_{j=1}^m H_0(dx_j, dt_j)\ .
	\end{align*}
	\begin{theorem}[\sc{identification}]\label{thm:ID}
		\noproofref
		Conditions \ref{condition:independent censoring} and \ref{condition:region of integration} imply that $\mathbbmsl{V}_1(f, \mathbbmsl{H}_0) = V_1(f,H_0)$, $\mathbbmsl{V}_2(\mathbbmsl{H}_0) = V_2(H_0)$, and so, $\mathbbmsl{V}(f, \mathbbmsl{H}_0) = V(f, H_0)$. 
	\end{theorem}
	This result is essentially a consequence of the identification of the hazard under conditionally independent censoring \citep{Beran1981}. Table \ref{tab:kernels} gives the form of $\omega$ and $\theta$ for each of the example measures introduced in Section \ref{sec:VIM review}, which can be shown to satisfy Condition \ref{condition:region of integration}(b) with $\tau_0 = \tau$. 
	
	\begin{table}
		\centering
		\begin{tabular}{lccc}  
			\toprule
			VIM & $m$ & $\omega[\{f(x_1), t_1\}, \dots, \{f(x_m), t_m\}]$ & $\theta(t_1,\dots,t_m)$\\\midrule
			AUC &2&$\I\{f(x_1) > f(x_2), t_1 \leq \tau, t_2 > \tau\} $&$\I(t_1 \leq \tau, t_2 > \tau) $\\
			Brier score &1&$ \left\{f(x) - \I(t > \tau)\right\}^2$&1\\ 
			Survival time MSE & 1 & $\left\{f(x) - (t \wedge \tau)\right\}^2$&1\\
			C-index &2&$\I\{f(x_1) > f(x_2),t_1 \leq \tau,t_2 > t_1\}$& $\I(t_1 \leq \tau,t_2 > t_1)$ \\ \bottomrule
		\end{tabular}
		\caption{Degree and kernel functions for example VIMs. The AUC and C-index kernels can be symmetrized by adding a second evaluation of the kernel with arguments exchanged, and dividing by two.}
		\label{tab:kernels}
	\end{table} 
	
	\subsection{Characterizing the oracle prediction function}
	
	Because $\mathbbmsl{V}(f, \mathbbmsl{H}_0) = V(f, H_0)$ as a result of Theorem \ref{thm:ID}, it follows that $\argmax_{f \in \mathcal{F}}\mathbbmsl{V}(f, \mathbbmsl{H}_0) = \argmax_{f\in \mathcal{F}}V(f, H_0)$, that is, any maximizer of the ideal data predictiveness measure is also a maximizer of the corresponding observed data predictiveness measure and vice-versa.  Often, as is the case for AUC, Brier score, and survival time MSE, $\mathbbmsl{f}_0(x)$ can be shown to depend on the evaluation of $t \mapsto \mathbb{F}_0(t \midd x)$ in $(0, \tau]$, and so it is identified under Conditions \ref{condition:independent censoring} and \ref{condition:region of integration}(a). The residual oracle prediction function $\mathbbmsl{f}_{0,s}$ can in many cases be written pointwise as $x \mapsto E_{\mathbbmsl{P}_0}\left\{\mathbbmsl{f}_0(X) \midd X_{-s} = x_{-s}\right\}$ and is therefore identified provided $\mathbbmsl{f}_0$ is. We define $f_0 \in \argmax_{f \in \mathcal{F}}V(f, P_0)$ and $f_{0,s} \in \argmax_{f \in \mathcal{F}_{s}}V(f, P_0)$ as the observed data oracle and residual oracle prediction functions, respectively. From here and on, we denote the identified oracle predictiveness as $v_0 := v_{0,1}/v_{0,2}$ with $v_{0,1} := V_1(f_0, H_0)$ and $v_{0,2} := V_2(H_0).$ Similarly, we denote the identified residual oracle predictiveness as $v_{0,s}:=v_{0,s,1}/v_{0,2}$ with $v_{0,s,1} := V_1(f_{0,s}, H_0)$. The identified importance of $X_s$ relative to the full feature vector is then denoted $\psi_{0,s} := v_{0} - v_{0,s}$.
	
	Oracle prediction functions for the stated example measures are given in Appendix 2. Notably, for AUC, Brier score, and survival time MSE, $f_0(x)$ is available in closed form and depends on the evaluation of $t \mapsto F_0(t \midd x)$. For the C-index, it is not straightforward to characterize the oracle prediction function or even to verify its existence. In the Supplementary Material, we outline a strategy for numerical optimization of the C-index, which is based on using gradient boosting to optimize an estimate of a smoothed, differentiable approximation of the C-index. While this strategy works well in all our numerical experiments, there is a need for additional theoretical work to characterize settings in which a maximizer exists. To our knowledge, this has not been fully addressed in the literature.
	
	\section{Estimation and inference}\label{sec:estimation and inference}
	
	\subsection{Overview}\label{sec:estimation overview}
	The observed data predictiveness measure depends on the unknown nuisances $f_0$ and $F_0$. Using flexible learning methods to estimate these nuisances decreases the risk of inconsistent estimation without the need for restrictive assumptions on the distribution of $T$ given $X$. Given estimators  $f_n$ and $F_n$, and denoting by $Q_n$ the empirical distribution of $X$ based on the data, we might consider the plug-in estimator $v_n := V_1(f_n, H_n)/V_2(H_n)$, where
	\begin{align*}
		H_n:(x_0,t_0)\mapsto \int \I(u \leq x_0) F_n(t_0\midd u) Q_n(du) = \frac{1}{n}\sum_{i=1}^{n}\I(X_i \leq x_0)F_n(t_0 \midd X_i)\ .
	\end{align*}
	In general, we cannot expect $f_n$ and $F_n$ to converge at rate $n^{-1/2}$. The fact that $f_0$ is a maximizer of $f \mapsto V(f, H_0)$ implies that its estimation makes no first-order contribution to the asymptotic behavior of $v_n$ \citep{Williamson2021b}. Nonetheless, we must pursue a debiasing strategy to account for the excess bias possibly induced by $F_n$, since failing to do so often precludes convergence in distribution at rate $n^{-1/2}$ and makes it difficult to achieve valid inference, that is, to build confidence intervals with nominal coverage and tests with proper type I error control. This is in contrast  to the simpler data structure considered in \cite{Williamson2021b}, where the ability to use an empirical plug-in estimate negates the need for debiasing and thereby dramatically simplifies the inferential procedure.
	
	There are several possible approaches to debiasing $v_n$. Here, we choose to debias $H_n$ and then plug the resulting estimator into the functional $H\mapsto V_1(f_n, H)/V_2(H)$. This involves first obtaining an approximation $b_n$ of the bias of $H_n$, which can be done using techniques from efficiency theory. The one-step debiased estimator $H_n^* := H_n - b_n$ can then be used to construct the predictiveness estimator $v_n^*:=V_1(f_n, H_n^*)/V_2(H_n^*)$. As we will show, this approach yields an estimator that is asymptotically linear and nonparametric efficient under certain conditions and enjoys a double-robustness property. A similar result holds for estimation of $v_{0,s}$. We outline the one-step VIM estimation procedure here and provide details in the remainder of this section. 
	\begin{step}
		\upshape
		Compute estimator $F_n$ of conditional distribution function of $T$ given  $X$ and  construct plug-in joint distribution function estimator $H_{n}$. 
	\end{step}
	\begin{step}
		\upshape
		Construct one-step debiased estimator $H_n^* = H_n - b_{n}$. 
	\end{step}
	\begin{step}
		\upshape
		Compute estimators $f_n$ and $f_{n,s}$ of full and residual oracle prediction functions. 
	\end{step}
	\begin{step}
		\upshape
		Compute VIM estimator $V_1(f_n,H_n^*) /V_2(H_n^*) - V_1(f_{n,s},H_n^*) /V_2(H_n^*)$. 
	\end{step}
	
	We note that we may have alternatively considered debiasing $V_1(f_n,H_n)/V_2(H_n)$ rather than $H_n$. However, doing so results in  predictiveness measure estimators with weaker robustness properties; this is discussed and illustrated in Section \ref{sec:sims}.
	
	\subsection{Efficiency}
	
	To perform valid inference, we must first study the predictiveness measure as a real-valued mapping defined on $\mathcal{M}_{\text{obs}}$. For a generic $P \in \mathcal{M}_{\text{obs}}$, we define $f_P$, $L_P$, $M_P$, $F_P$, $Q_P$ and $H_P$ similarly as $f_0$, $L_0$, $M_0$, $F_0$, $Q_0$ and $H_0$ with $P$ substituted in place of $P_0$. 
	
	We first present the efficient influence function of $P \mapsto V(f_P, H_P)$, which plays a key role in our proposed procedure. This efficient influence function involves the conditional survival function of $C$ given $X$, which is identified under Condition \ref{condition:independent censoring} as $\pr_{\mathbbmsl{P}_0}(C \geq t \midd X = x) = G_0(t \midd x)$ with $G_P(t \midd x) := \prodi_{(0,t)}[1 - \{1 - M_P(u^- \midd x)\}^{-1}M_{P,0}(du \midd x)]$ and $M_{P,0}(t \midd x) := \pr_{P}(Y \leq t, \Delta = 0 \midd X = x)$.
	This nuisance parameter appears in the key quantity
	\begin{align*}
		z:=(x,y,\delta) \mapsto \chi_{0}(z,t) := -S_0(t \midd x)\left\{\frac{\delta \I_{[0,t]}(y)}{S_0(y \midd x)G_0(y \midd x)} - \int_0^{t \wedge y}\frac{L_0(du \midd x)}{S_0(u \midd x)G_0(u \midd x)}\right\}
	\end{align*}
	with $S_0 := 1- F_0$. Up to an inverse weighting term, the above mapping is the influence function of the stratified Kaplan-Meier estimator \citep{Reid1981}. We also define the mappings
	\begin{align*}
		\omega_{0,1}: (x,t) &\mapsto \int \cdots \int \omega[\{f_0(x),t\},\{f_0(x_{2}),t_{2}\}, \dots, \{f_0(x_m), t_m\}]\prod_{j=2}^{m}H_0(dx_j, dt_j)\ ;\\
		\theta_{0,1}: t &\mapsto \int \cdots \int \theta(t,t_2, \dots, t_m)\prod_{j=2}^{m}H_0(dx_j, dt_j)\ .
	\end{align*} 
	Condition \ref{condition:efficient influence function}, used in the derivation of the efficient influence function, is given in Appendix 1. This condition essentially requires a certain degree of smoothness of $(f,H)\mapsto V(f,H)$ around $(f_0,H_0)$.
	
	\begin{theorem}[\sc{efficient influence function}]\label{thm:efficient influence function}\noproofref
		Suppose that there exists $\eta \in (0, \infty)$ such that $G_0(\tau \midd X) \geq \eta$ $P_0$--almost surely and that Condition \ref{condition:efficient influence function} holds. Then, $P \mapsto V(f_P, H_P)$ is pathwise differentiable at $P_0$ relative to the nonparametric model $\mathcal{M}_{\text{obs}}$, with efficient influence function given by $
		\phi_0:=(\phi_{\omega,0}-v_0\phi_{\theta,0})/v_{0,2}$,
		where
		\begin{align*}
			&\phi_{\omega,0}: z \mapsto m\left[\int \omega_{0,1}\left(x, t\right)\left\{F_0(dt\midd x) - \chi_{0}(z,dt)\right\}-v_{0,1}\right]; \\
			&\phi_{\theta,0}: z \mapsto m\left[\int \theta_{0,1}\left( t\right)\left\{F_0(dt\midd x) - \chi_{0}(z,dt)\right\}-v_{0,2}\right].
		\end{align*}
	\end{theorem}
	If $F_0$ were known, in view of the theory of $V$--statistics \citep{Serfling1980}, $V_1$ and $V_2$ would represent generalized moment functionals with uncentered nonparametric efficient influence functions $x \mapsto m\int \omega_{0,1}(x,t)F_0(dt \midd x)$ and $x \mapsto m\int \theta_{0,1}(t)F_0(dt \midd x)$. The functions $z \mapsto m\int \omega_{0,1}(x,t)\chi_{0}(z,dt)$ and $z \mapsto m\int \theta_{0,1}(t)\chi_{0}(z,dt)$ therefore represent contributions from estimating $F_0$. As indicated before, the fact that $f_0$ is unknown has no impact on the form of the efficient influence function. 
	
	\subsection{Estimation of the joint distribution function}
	
	As indicated above, to produce an estimator of $v_0$, we first derive a debiased machine learning strategy for estimation of $H_0$. This strategy relies on the efficient influence function of $P \mapsto H_P(x_0, t_0)$ at $P_0$ relative to $\mathcal{M}_{\text{obs}}$, given by $z \mapsto \bar{\varphi}_{0,x_0,t_0}(z):= \varphi_{0,x_0,t_0}(z) - H_0(x_0,t_0)$
	with $\varphi_{0,x_0,t_0}(z) := \I(x \leq x_0)\{F_0(t_0 \midd x) - \chi_{0}(z,t_0)\}$; details are provided in Lemma 1 in the Supplementary Material. Here, $(x_0,t_0) \in \mathcal{X}\times (0, \tau_0]$ denotes a generic point in the identified support of $(X,T)$. Our approach proceeds by constructing an estimator $P_n$ of $P_0$. Since the representation of $\varphi_0$ includes the variation-independent nuisances $F_0$ and $G_0$, a natural parametrization of $P_0$ is $(F_0, G_0, Q_0)$. Of course, any estimate of $F_0$ implies estimates of $S_0$ and $L_0$ as well. We define $P_n$ as the estimator of $P_0$ constructed from estimators $F_n$ and $G_n$ as well as the empirical covariate distribution $Q_n$. We denote by $\varphi_{n}$ and $\bar{\varphi}_{n}$ the estimated uncentered and centered influence functions obtained by replacing $F_0$ and $G_0$ by $F_n$ and $G_n$, respectively, in $\varphi_0$ and $\bar{\varphi}_0$. We define $\varphi_{\infty}$ and $\bar{\varphi}_{\infty}$ similarly but instead replacing $F_0$ and $G_0$ by $F_\infty$ and $G_\infty$, the in-probability limits of $F_n$ and $G_n$, respectively. The notation $F_\infty$ and $G_\infty$ is used to emphasize that $F_n$ and $G_n$ may not necessarily be consistent for the intended nuisance functions $F_0$ and $G_0$. We study the behavior of the plug-in estimator $H_n(x_0,t_0)$ by decomposing
	\begin{equation}
		\begin{aligned}
			H_n(x_0,t_0) - H_0(x_0,t_0) \ =&\ \ \frac{1}{n}\sum_{i=1}^{n}\left\{\bar{\varphi}_{\infty, x_0,t_0}(Z_i) - \int \bar{\varphi}_{\infty, x_0,t_0}(z)P_0(dz)\right\} \label{eq:full linearization F_0}\\
			&\hspace{0.75cm}+ R_{x_0,t_0}(P_n, P_0) + C_{n,x_0,t_0}(P_n, P_\infty)  - \frac{1}{n}\sum_{i=1}^{n}\bar{\varphi}_{n,x_0,t_0}(Z_i) \ ,
		\end{aligned}
	\end{equation}
	where we define $R_{x_0,t_0}(P_n,P_0) := H_n(x_0,t_0) - H_0(x_0,t_0) +\int\bar{\varphi}_{n,x_0,t_0}(z)P_0(dz)$ and $C_{n,x_0,t_0}(P_n, P_\infty) := \int \{\bar{\varphi}_{n,x_0,t_0}(z) - \bar{\varphi}_{\infty,x_0,t_0}(z)\}(\mathbb{P}_n - P_0)(dz)$. The terms $R_{x_0,t_0}(P_n,P_0)$ and $C_{n,x_0,t_0}(P_n, P_\infty)$ are second-order and hence asymptotically negligible under some conditions. 
	The leading linear term of (\ref{eq:full linearization F_0}) is the empirical average of a mean-zero transformation of $Z_i$, $i=1,\ldots,n$. The final term represents the excess bias due to flexibly estimating $F_0$ and $G_0$, and its presence indicates that $H_n(x_0,t_0)$ may fail to achieve $n^{1/2}$--consistency. The one-step estimator of $H_0(x_0,t_0)$ is then
	\begin{align*}
		H_n^*(x_0,t_0):=H_n(x_0,t_0) + \frac{1}{n}\sum_{i=1}^{n}\bar{\varphi}_{n,x_0,t_0}(Z_i) = \frac{1}{n}\sum_{i=1}^{n}\varphi_{n,x_0,t_0}(Z_i)\ .
	\end{align*}
	In light of (\ref{eq:full linearization F_0}), under certain conditions,  $H_n^*(x_0,t_0)$ is an asymptotically linear estimator of $H_0(x_0,t_0)$ with influence function $\bar{\varphi}_{\infty,x_0,t_0}$. 
	
	The expected second-order behavior of $R_{x_0,t_0}(P_n, P_0)$ is a consequence of the pathwise differentiability of $P \mapsto H_P(x_0,t_0)$. For $C_{n,x_0,t_0}(P_n, P_\infty)$ to be negligible, it is often assumed that $F_n$ and $G_n$ fall in Donsker classes in large samples, essentially requiring that the algorithm used to produce $F_n$ and $G_n$ not be too complex. Cross-fitting can circumvent the need for such a condition \citep{Zheng2011, Chernozhukov2018}. For $K$--fold cross-fitting, the data indices $1, \dots, n$ are partitioned into $K$ subsets, say $\mathcal{I}_1, \dots, \mathcal{I}_K$ of roughly equal sizes $n_1,\ldots,n_K$. Then, for each $k \in \{1,\dots, K\}$, the observations with index in $\mathcal{I}_k$ are set aside as test data; estimators $F_{n,k}$ and $G_{n,k}$ are constructed using the rest of the data, yielding estimated influence function $\varphi_{n,k,x_0,t_0}$; and the fold-specific estimator $H^*_{n,k}(x_0,t_0):=n_k^{-1}\sum_{i\in \mathcal{I}_k}\varphi_{n,k,x_0,t_0}(Z_i)$ is computed. Finally, the cross-fitted estimator is taken to be the average of all $K$ fold-specific estimators. Below, we use a similar procedure to produce a cross-fitted VIM estimator. 
	
	\subsection{VIM estimation}\label{subsec:vim estimation}
	
	To estimate the oracle predictiveness $v_0$, we consider the cross-fitted estimator $v_n^* := v_{n,1}^*/v_{n,2}^*$ with $v_{n,1}^*:= K^{-1}\sum_{k=1}^{K}V_1(f_{n,k}, H^*_{n,k})$ and $v_{n,2}^* := K^{-1}\sum_{k=1}^{K}V_2(H^*_{n,k})$, 
	where $f_{n,k}$ denotes an estimator of $f_0$ based on the same data used to construct $F_{n,k}$ and $G_{n,k}$. 
	We study the behavior of $v_n^*$ by separately considering $v_{n,1}^*$ and $v_{n,2}^*$. For $v_{n,1}^*$, we can decompose
	\begin{align*}
		v_{n,1}^* &-  v_{0,1} = \frac{1}{K}\sum_{k=1}^{K}\left\{V_1(f_0, H^*_{n,k}) - V_1(f_0, H_0)\right\} + \frac{1}{K}\sum_{k=1}^{K}\left\{V_1(f_{n,k}, H_0) - V_1(f_0, H_0)\right\} + r_n\ ,   
	\end{align*}
	where $r_{n} := K^{-1}\sum_{k=1}^{K}[\{V_1(f_{n,k}, H^*_{n,k}) - V_1(f_{n,k}, H_0)\} - \{V_1(f_0, H^*_{n,k}) - V_1(f_0, H_0)\}]$. In the first term on the right-hand side above, the prediction function argument is fixed at $f_0$, so the behavior of this term is determined by that of $H^*_{n,k}$ via the mapping $H \mapsto V_1(f_0, H)$. The second term is the contribution from the estimation of $f_0$, which, due to the optimality property of $f_0$, is expected to be second-order. The final term $r_n$ is a difference-of-differences term that will also be second-order under some conditions. For $v_{n,2}^*$, which does not involve $f_0$, we simply have that $v_{n,2}^* - v_{0,2} = K^{-1}\sum_{k=1}^{K}\{V_2(H^*_{n,k}) - V_2(H_0)\}$.
	
	Conditions under which the following large-sample results hold are stated in Appendix 1. Our first result establishes the consistency of $v_n^*$. While we require that $f_n$ converge to $f_0$ for this result, we only require that one of $F_n$ and $G_n$ converge to $F_0$ and $G_0$, respectively. In this sense, $v_n^*$ is a doubly-robust estimator. Notably, in many cases it is possible to estimate $f_0$ itself in a doubly-robust manner; we discuss this in Section \ref{subsec:nuisance estimation}.
	
	\begin{theorem}[\sc{double-robust consistency}]\label{thm:consistency}\noproofref
		If Conditions \ref{condition:bounded away from zero}--\ref{condition:continuity of V} hold, then $v_n^*$ converges in probability to $v_0$. 
	\end{theorem}
	
	Our next result establishes the asymptotic linearity of $v_n^*$ under slightly stronger conditions. 
	Often, under regularity conditions on $P_0$ and $\mathcal{F}$, the requirement that estimation of $f_0$ make no first-order contribution to the behavior of $v_n^*$ is satisfied because $f_0$ is an optimizer of $f \mapsto V(f, P_0)$. We provide details for the example measures in the Supplementary Material. We also require that the nuisance functions be estimated at sufficiently fast rates: roughly speaking, we need that $f_{n,k} - f_0 = o_P(n^{-1/4})$, $(G_{n,k} - G_0)(F_{n,k} - F_0) = o_P(n^{-1/2})$ and $(f_{n,k} - f_0)(\bar{\varphi}_{n,k} - \bar{\varphi}_0) = o_P(n^{-1/2})$ in an appropriate sense.
	
	\begin{theorem}[\sc{asymptotic linearity}]\label{thm:asymptotic linearity}\noproofref
		If Conditions \ref{condition:bounded away from zero}, \ref{condition:nuisance limits}, and \ref{condition:optimality}--\ref{condition:second order remainder} hold, and in addition, $F_\infty = F_0$ and $G_\infty = G_0$, then $v_n^* - v_0 = n^{-1}\sum_{i=1}^{n}\phi_0(Z_i) + o_P(n^{-1/2})$. If Condition \ref{condition:efficient influence function} also holds, then $v_n^*$ is nonparametric efficient.
	\end{theorem}
	
	By substituting $f_{n,k}$, $f_{0}$ and $\mathcal{F}$ for $f_{n,k,s}$, $f_{0,s}$ and $\mathcal{F}_s$, respectively, a similar result to Theorem \ref{thm:asymptotic linearity} holds for the cross-fitted one-step estimator $v_{n,s}^*$ of the residual oracle predictiveness $v_{0,s}$, with resulting influence function $\phi_{0,s}$ obtained by replacing $f_0$ with $f_{0,s}$ in the form of $\phi_0$. To estimate the variable importance $\psi_{0,s}$, we use $\psi_{n,s}^* := v_{n}^* - v_{n,s}^*$, which is itself asymptotically linear and nonparametric efficient with influence function $\phi_0-\phi_{0,s}$. When $\psi_{0,s} \neq 0$, that is, under non-null importance, $n^{1/2}(\psi_{n,s}^* - \psi_{0,s})$ converges in distribution to a mean-zero normal random variable with variance $\sigma^2_{0,s} := \Var_{P_0}\{\phi_0(Z) - \phi_{0,s}(Z)\}>0$. In Algorithm \ref{alg:alternative, xfit} in Appendix 3, we propose a cross-fitted estimator $\sigma^2_{n,s}$ of $\sigma^2_{0,s}$ and describe how to conduct inference under non-null importance. Denoting by $z_{q}$ the $q$th quantile of the standard normal distribution, a confidence interval with asymptotic coverage $1 - \alpha$ is $(\psi_{n,s} - z_{1 - \alpha/2} \sigma_{n,s}n^{-1/2}, \psi_{n,s} + z_{1 - \alpha/2} \sigma_{n,s}n^{-1/2})$.
	
	When $\psi_{0,s} = 0$, the influence function of $\psi_{n,s}^*$ degenerates to the zero function. In this case, $\psi_{n,s}^*$ converges to zero at a rate faster than $n^{-1/2}$, and the inferential procedure described in Algorithm \ref{alg:alternative, xfit} may fail. This boundary problem has been a topic of recent study \citep{Dai2022, Lundborg2022, Hudson2023}. Sample splitting \citep{Dai2022, Williamson2021b} provides a simple means of constructing tests of the null hypothesis that $\psi_{0,s} = 0$ as well as confidence intervals that remain valid under that null. The sample splitting procedure, outlined in Algorithm \ref{alg:alternative, ss xfit} in Appendix 3, entails estimating $v_0$ and $v_{0,s}$ using non-overlapping portions of the data. The VIM estimator obtained as the difference between estimators of $v_0$ and $v_{0,s}$ constructed on independent data has non-degenerate behavior under the null and thus can easily be used to achieve valid inference. 
	
	\subsection{Estimation of nuisance parameters}\label{subsec:nuisance estimation}
	
	Due to the forms of the efficient influence functions $\phi_0$ and $\phi_{0,s}$, our proposed procedure requires estimation of the conditional time-to-event distribution function $F_0$ and conditional censoring survival function $G_0$ over the entire interval $(0, \tau]$. Notably, this is true even if the chosen predictiveness measure itself can be shown to depend only on the evaluation of $t \mapsto F_0(t \midd x)$ at landmark time $\tau$, as is the case for AUC\ or Brier score. Therefore, we recommend using a learning approach that targets the entire distribution function rather than its evaluation at a single time point. Furthermore, to reduce the risk of inconsistent estimation due to the use of misspecified models, we focus on flexible learning methods for survival data, including the approaches of \citet{Westling2023} and \cite{Wolock2024}.
	
	Estimation of $f_0$ must be handled on a case-by-case basis. For many predictiveness measures, such as AUC, Brier score, and survival time MSE, $f_0$ can be written in terms of $F_0$, and so an estimate of $F_0$ can be used to produce an estimate of $f_0$ without fitting any additional algorithm. This strategy depends on consistent estimation of $F_0$. \citet{Rubin2007} propose a doubly-robust pseudo-outcome regression procedure to estimate the conditional mean of a transformation of $T$ given $X$. Therefore, when the oracle prediction function takes such a form, the doubly-robust pseudo-outcome approach can be used to produce a doubly-robust estimator of $f_0$. This is a two-step procedure: first, estimates of $F_0$ and $G_0$ are used to construct the pseudo-outcomes, and then, an additional regression is fit to estimate the desired conditional mean. In Section \ref{sec:sims}, we illustrate this procedure for AUC\ variable importance. When an oracle is not available in closed form, as for the C-index, an additional optimization scheme can be performed to produce an estimate of $f_0$. In the Supplementary Material, we outline a numerical optimization approach for the C-index based on gradient boosting. 
	
	When $f_{0,s}(x)$ can be written as $E_{P_0}\{f_{0}(X)\midd X_{-s} = x_{-s}\}$, as for AUC, Brier score, and survival time MSE, we can construct full oracle predictions and then regress them on the reduced feature vector to obtain an estimate of $f_{0,s}$. When numerical optimization is required, as for the C-index, the class of potential optimizers can be restricted to those depending only on the residual covariates $X_{-s}$. 
	
	\section{Numerical experiments}\label{sec:sims}
	
	\subsection{Simulation setup}
	
	We conducted numerical studies to evaluate the performance of our proposed inferential procedure for survival VIMs. In all experiments, we began by generating independent replicates of $(X, T, C)$, where $X$ was a $p$--dimensional covariate vector generated from a multivariate normal distribution with mean vector $(0,\dots, 0)$ and covariance matrix $\Sigma$. The values of $p$ and $\Sigma$ varied between simulation settings. Given covariate vector $X = x$, the event time $T$ and censoring time $C$ were simulated from the log-normal accelerated failure time models 
	\begin{align*}
		\log T = \tfrac{5}{10}x_1 -\tfrac{3}{10}x_2 + \tfrac{1}{10}x_1x_2 - \tfrac{1}{10}x_3x_4 + \tfrac{1}{10}x_1x_5 + \varepsilon_T\ ,\hspace{0.2in} \log C = \beta_{0,C} - \tfrac{2}{10}x_1 + \tfrac{2}{10}x_2 + \varepsilon_C\ ,
	\end{align*}
	where $\varepsilon_T$ and $\varepsilon_C$ were independent standard normal random variables, and where $\beta_{0,C}$ was chosen to achieve the desired censoring rate in each simulation setting, with $X_1$ and $X_2$ informing the censoring mechanism. For each observation, the observed follow-up time $Y:=\min\{T, C\}$ and event indicator $\Delta:=\I(T \leq C)$ were computed. 
	
	We evaluated our estimation procedure using three different strategies for estimating nuisance functions $F_0$ and $G_0$: random survival forests \citep{Ishwaran2008}, survival Super Learner \citep{Westling2023}, and global survival stacking \citep{Wolock2024}. The latter two are meta-learners that estimate an optimal linear combination of base learners using cross-validation. Details on nuisance estimation, including algorithm libraries and selection of tuning parameters, are given in the Supplementary Material. The integrals appearing in $v_n^*$ and $v_{n,s}^*$ were approximated as sums on a finite grid consisting of the observed event times.
	
	\subsection{Overall performance of the proposed procedure}
	
	First, we investigated the performance of our procedure under several qualitatively distinct scenarios, which are summarized in Table \ref{tab:scenarios} with true VIM values for each scenario given in the Supplementary Material.  In all scenarios, we generated 500 random datasets of size $n \in\{500, 750, \dots,1500\}$. We considered VIMs based on AUC, Brier score, and C-index. For the landmark time VIMs, the oracle prediction functions were estimated as $x\mapsto F_n(\tau \midd x)$. To estimate the residual oracle prediction functions, predictions $F_n(\tau \midd X_1),\dots,F_n(\tau \midd X_n)$ were regressed on the reduced covariate vectors $X_{1,-s},\ldots,X_{n,-s}$ using Super Learner with the same library used in global survival stacking. For the C-index, $F_n(\cdot \midd x)$ was plugged into the numerical optimization procedure described in the Supplementary Material to estimate the oracle and residual oracle prediction functions. A custom family implemented in the \texttt{mboost} package was used for numerical optimization. We implemented our estimation procedure with five-fold cross-fitting and also included non-cross-fitted comparators, that is, using the proposed procedures with $K = 1$. Landmark times for AUC\ and Brier score were set to $\tau \in \{0.5, 0.9\}$, corresponding to the 50th and 75th population quantiles of observed event times. The restriction time for the C-index was set to $\tau = 0.9$. We show only a subset of the results for Scenario I here; additional results for Scenario I and all results for Scenarios II--IV can be found in the Supplementary Material.
	
	\begin{table}\centering
		\begin{tabular}{cccccc}
			\toprule
			Scenario & $p$ & Null features & Corr. features & Censoring rate & Results \\\midrule
			I& 25  &Yes &Yes&50\%&Main text, Supp. Material\\
			II & 25& Yes & No & 50\%&Supp. Material\\
			III & 5 & No & No & 50\%& Supp. Material\\
			IV & 25 & Yes & No & 30\%, 40\%, \ldots, 70\% & Supp. Material \\ \bottomrule
		\end{tabular}\caption{Summary of simulation scenarios.}
		\label{tab:scenarios}
	\end{table} 
	
	In Scenario I, we set $p = 25$ and $\Sigma$ to be a $25 \times 25$ matrix with 1 on the diagonal. Off-diagonal elements were set to 0, except for $\Sigma_{1,6} = \Sigma_{6,1} = 0.7$ and $\Sigma_{2,3} = \Sigma_{3,2} = -0.3$. Five of 25 features had non-zero importance; we considered the importance of correlated features $X_1$ and $X_6$, as well as the group importance of $(X_1,X_6)$. While $X_6$ does not directly affect the event time $T$, the true importance of $X_1$ is altered due to its association with $X_6$. We set $\beta_{0,C} = 0$ to achieve a 50\% censoring rate. We used the sample splitting procedure described in Algorithm \ref{alg:alternative, ss xfit} to compute point and standard error estimates, from which we computed nominal 95\% Wald-type confidence intervals. In addition, we computed $p$-values corresponding to the null hypothesis of zero importance versus the one-sided alternative. We evaluated performance using the empirical bias scaled by $n^{1/2}$; the empirical variance scaled by $n$ and divided by the theoretical asymptotic variance; the empirical confidence interval coverage; and the average confidence interval width (for $X_1$ and $(X_1,X_6)$) or empirical rejection probability (for $X_6$).
	
	In Fig. \ref{fig:scenario2_small}, we show the results for estimating the importance of $X_1$, $X_6$ and $(X_1,X_6)$ based on AUC\ at $\tau = 0.5$ in Scenario I. The scaled bias for the cross-fitted estimators is small, with global survival stacking and survival Super Learner generally outperforming the random survival forests estimator. Compared to their cross-fitted counterparts, the non-cross-fitted estimators tend to have inflated bias, especially for estimating the importance of $X_1$ and $X_6$. Besides the non-cross-fitted random survival forests estimator, all estimators have variance roughly proportional to sample size and near the expected asymptotic variance. The confidence interval coverage approaches the nominal level with increasing sample size when cross-fitting is used. The non-cross-fitted procedures are mildly anticonservative for global survival stacking and survival Super Learner and highly anticonservative for random survival forests. For $X_1$ and $(X_1, X_6)$, the width of all confidence intervals decreases with increasing sample size, as expected. For $X_6$, the type I error is generally controlled near the 0.05 level when using cross-fitting. Altogether, these results demonstrate that our proposed procedure has strong performance in samples of realistic size, and they underscore the importance of cross-fitting when using flexible machine learning nuisance estimators. Under the null hypothesis, the sample splitting procedure is well calibrated.
	
	\begin{figure}
		\centering
		\includegraphics[width=\linewidth]{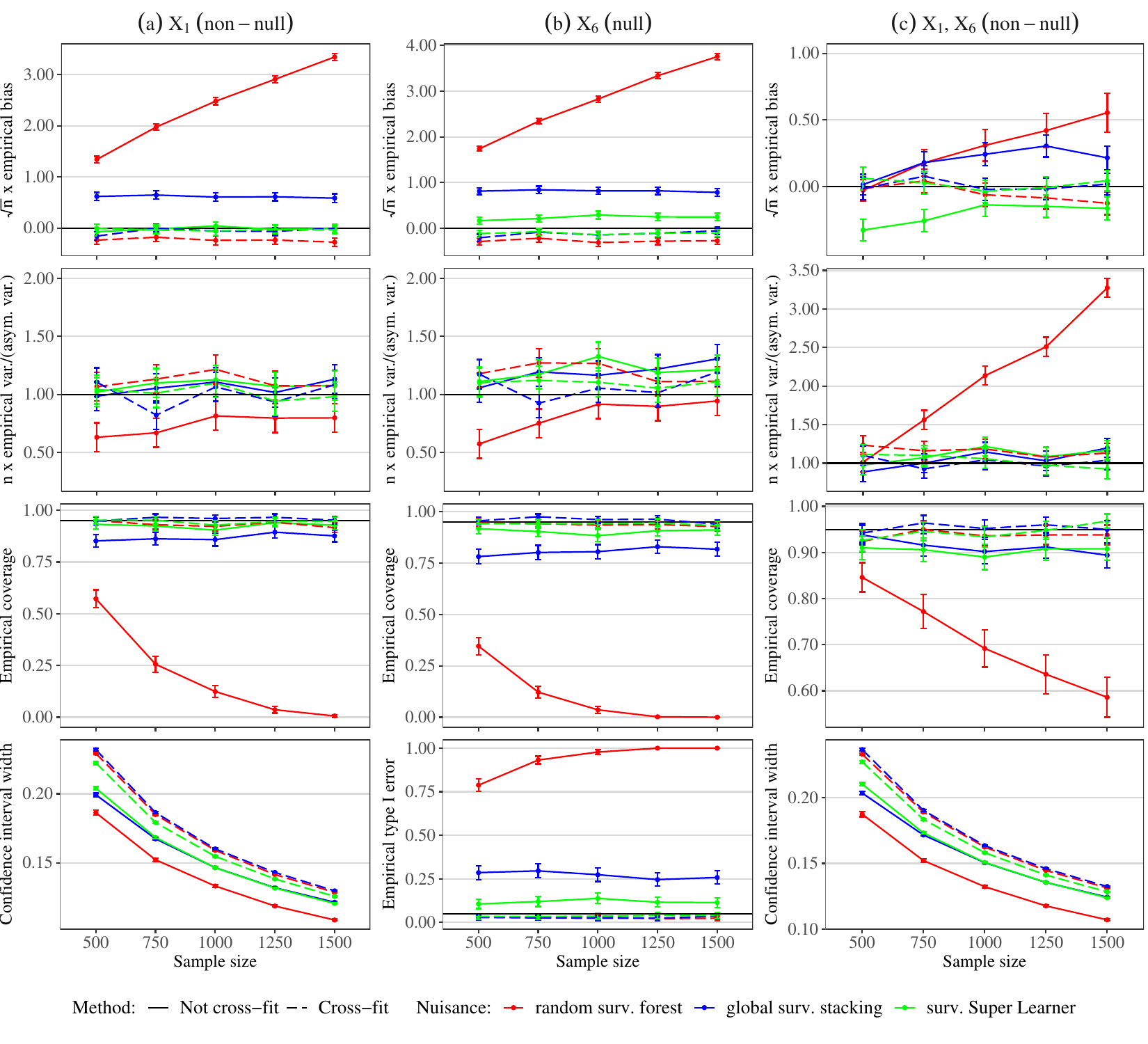}
		\caption{Performance of the one-step AUC VIM estimator in Scenario I in terms of bias, variance, confidence interval coverage, interval width, and type I error. Columns correspond to performance for estimating (a) the non-null importance of $X_1$, (b) the null importance of $X_6$, and (3) the non-null importance of $(X_1,X_6)$. Colors denote different nuisance estimators, which were used to estimate both event and censoring distributions. Solid and dashed lines denote non-cross-fitted and cross-fitted estimators, respectively. Vertical bars represent 95\% confidence intervals taking into account Monte Carlo error.}
		\label{fig:scenario2_small}
	\end{figure}
	
	In the Supplementary Material, we present results in Scenario I for the Brier score and C-index VIMs, as well as all results in Scenario II, in which the features were uncorrelated; Scenario III, in which all features were important and sample splitting was not used; and Scenario IV, in which we varied the censoring rate. Overall, we see that Algorithm \ref{alg:alternative, xfit} performs well when the features of interest have non-null importance. We also see that the operating characteristics of the procedure are largely consistent across censoring levels. Unsurprisingly, the impact of censoring on estimator variance and confidence interval width is larger at the later landmark time, when the censoring rate is higher. The proposed procedure performs similarly well with uncorrelated features as with correlated features. The results also show that, even when the oracle prediction function is not available in closed form, as for the C-index, a numerical optimization approach can yield good results. 
	
	\subsection{Performance of the proposed procedure under nuisance misspecification}
	
	Next, we investigated the robustness properties of our procedure. The doubly-robust behavior of $v_n^*$ is affected by the debiasing approach and by the method used for estimating the oracle prediction function $f_0$; we consider both of these factors here. 
	
	As noted in Section \ref{sec:estimation overview}, there are several ways to debias the plug-in estimator defined as $v_n := V_1(f_n,H_n)/V_2(H_n)$. The double-robustness of $v_{n}^*$, established in Theorem \ref{thm:consistency}, is essentially a consequence of the double-robustness of the one-step estimator $H_{n,k}^*(x_0,t_0)$ uniformly over $(x_0,t_0)$; see Lemma 7 in the Supplementary Material. It is also possible to use the efficient influence function of $P\mapsto V(f_P, P)$ to directly debias $v_n$. However, for measures with $m \geq 2$, this approach generally yields a remainder term that, while second-order, does not involve the conditional censoring distribution. (In special cases, such as when $m = 1$ and $\theta(t_1) = 1$, the two approaches to debiasing $v_n$ yield identical doubly-robust estimators.) Therefore, for this remainder term to tend to zero in probability, consistent estimation of $F_0$ is required. We provide a theoretical illustration of this phenomenon in the Supplementary Material. We refer to our proposed method as the \textit{indirect one-step} estimator, and the method based on debiasing using the efficient influence function of $P \mapsto V(f_P, P)$ as the \textit{direct one-step} estimator. 
	
	Regardless of the chosen debiasing method, consistency of $v_n^*$ depends on consistent estimation of $f_0$. As discussed in Section \ref{subsec:nuisance estimation}, in some cases $f_0$ can be estimated in a doubly-robust manner. For example, for AUC\ and Brier score VIMs, the oracle prediction function $x \mapsto F_0(\tau \midd x)$ can be estimated using the doubly-robust pseudo-outcome regression approach of \citet{Rubin2007}, for which only one of the initial nuisance estimators $F_n$ and $G_n$ must be consistent; details for the pseudo-outcome regression procedure are given in the Supplementary Material.  
	
	To illustrate the different robustness properties of these approaches, we performed a numerical study. We generated 500 random datasets of size $n \in \{250, 500, 1000, 2500, 5000\}$ under the settings of Scenario I. We considered the importance of $X_1$ in terms of AUC\ at $\tau = 0.5$. We compared the properties of the indirect and direct one-step estimators and of two approaches to estimating $f_0$: (i) the conditional distribution function method using $F_n(\tau \midd X_1), \ldots, F_n(\tau \midd X_n)$; and (ii) the doubly-robust pseudo-outcome method, where initial estimators $F_n$ and $G_n$ were plugged into the pseudo-outcome mapping, and the estimated pseudo-outcomes were regressed on the feature vector using Super Learner. Our theory indicated that all estimators described here should exhibit consistency under inconsistent estimation of $G_0$, while only the indirect one-step estimator using doubly-robust pseudo-outcome regression should be robust to inconsistent estimation of $F_0$.
	
	We considered three misspecification settings. In the first, both $F_0$ and $G_0$ were estimated using survival Super Learner. In the second setting, $F_0$ was estimated using survival Super Learner but $G_0$ was deliberately inconsistently estimated. This was achieved as follows: for each $X_i$, we set $G_n(t_1 \midd  X_i),  \ldots, G_n(t_J \midd X_i)$ to an evenly spaced grid of values between 1 and 0.1, where $\{t_1,\ldots,t_J\}$ denotes the grid of times used to approximate integrals appearing in construction of the one-step estimator. Thus, $G_n$ was a valid conditional survival function but was not estimated using the data and was the same for each $X_i$. In the third setting, $G_0$ was estimated using survival Super Learner while $F_0$ was deliberately inconsistently estimated as described (with the grid direction reversed since $F_0$ is a distribution function). We used Algorithm \ref{alg:alternative, ss xfit} with five-fold cross-fitting. For each simulation replicate, we computed the error of the VIM estimator. We computed the empirical bias and empirical variance of the VIM estimator across simulation replicates.
	
	In Fig. \ref{fig:robust} we show the results of this experiment. From the left column, we observe that when neither $F_0$ nor $G_0$ is deliberately inconsistently estimated, all four estimators perform well, with decreasing variance and bias tending to zero with increasing sample size. From the center column, where the censoring estimator $G_n$ is inconsistent, we observe a similar overall pattern. In the right column, where $F_0$ is deliberately inconsistently estimated, only the indirect one-step estimator using the doubly-robust pseudo-outcome method appears consistent; all other estimators have non-vanishing bias. Overall, these results support our theoretical results on the robustness of our proposed procedure to severely inconsistent estimation of either nuisance. We note that the pseudo-outcome approach carries an increased computational burden, as it requires fitting $K$ additional regressions when $K$-fold cross-fitting is used. 
	
	\begin{figure}
		\centering
		\includegraphics[width=\linewidth]{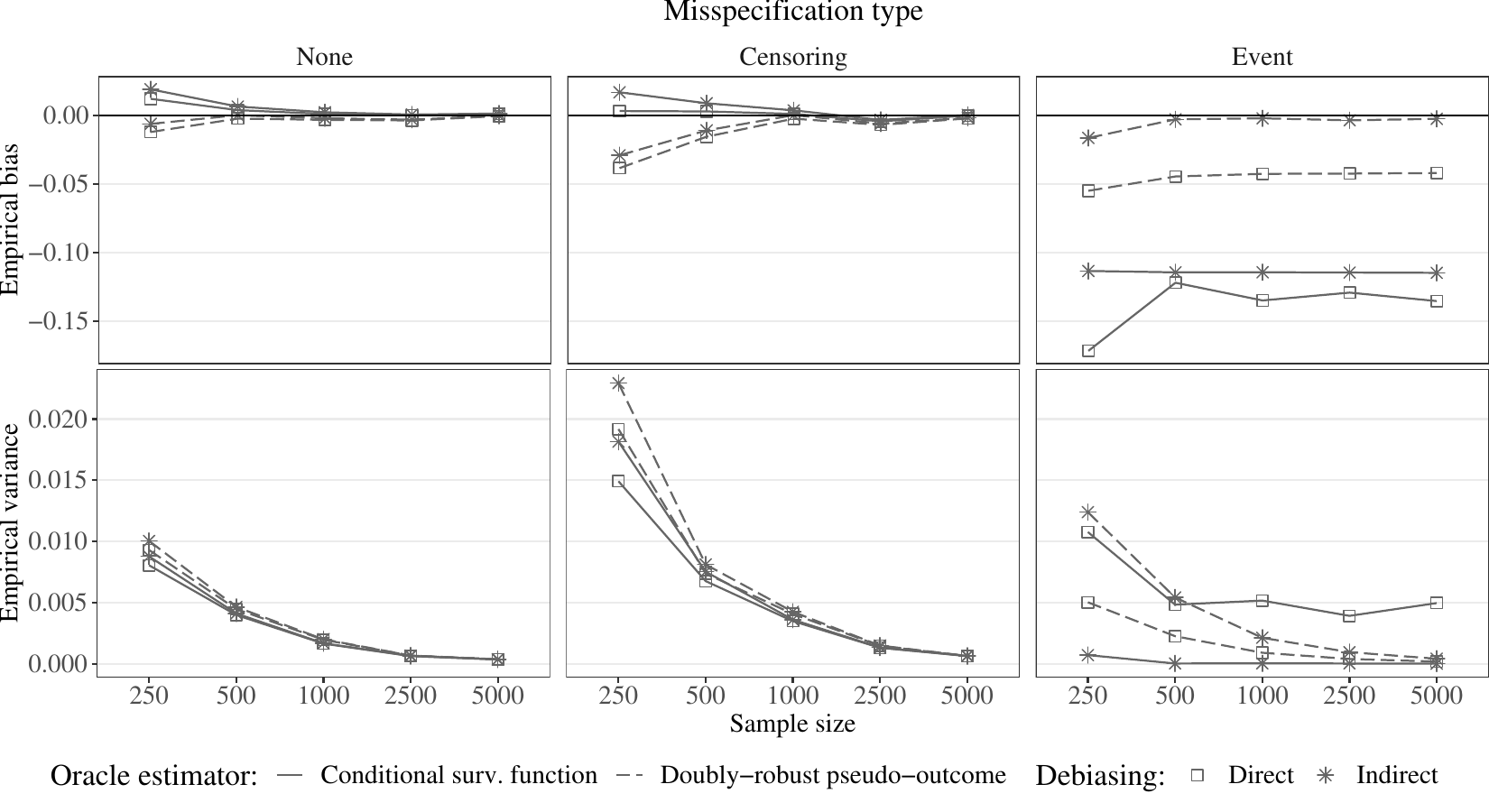}
		\caption{Comparison of robustness of one-step AUC\ VIM estimation procedures. In the left column, both conditional event and censoring distributions were estimated consistently. In the center and right columns, the conditional censoring estimator and conditional event estimator, respectively, were deliberately misspecified. Squares correspond to the direct one-step estimator and stars to the proposed indirect one-step estimator. Solid lines correspond to using the conditional survival function estimator of the oracle prediction function and dashed lines to using the doubly-robust pseudo-outcome approach.}
		\label{fig:robust}
	\end{figure}
	
	\section{Variable importance in HVTN 702}\label{sec:data analysis}
	
	The HVTN 702 vaccine trial \citep{Gray2021} included 5404 participants, with 2704 and 2700 individuals assigned to the vaccine and placebo groups, respectively. Participants were healthy adults between the ages of 18 and 35 years. By the time the trial was terminated when nonefficacy stopping criteria were met, around 5\% of participants had been diagnosed with HIV acquisition and around 7\% had been lost to follow-up; the remainder were not diagnosed with HIV acquisition during follow-up. Secondary analyses of data from this trial included development of a ``baseline risk score" for the 24-month probability of HIV acquisition diagnosis using covariates measured at enrollment, and a subsequent analysis of vaccine efficacy within strata defined by that risk score. The risk score was constructed using a regularized Cox model, but the importance of the baseline covariates for predicting the probability of HIV acquisition diagnosis was never formally assessed. 
	
	In order to investigate variable importance for predicting time to HIV acquisition diagnosis in HVTN 702, we chose to measure predictiveness in terms of landmark time AUC and C-index. This aligns with the primary goal of baseline risk score development in HIV trials: in most trials, feasible time horizons are short enough that only a small percentage of participants acquire HIV, and interest lies in discriminating between those who are at higher versus lower risk of acquisition. We evaluated AUC\ importance at several landmark times to assess if and how variable importance varies over different time horizons. 
	
	We analyzed variable importance in the modified intention-to-treat cohort, consisting of the 5384 participants who underwent randomization and were HIV-1 negative at baseline. The median follow-up time in this cohort was 623 days. Because vaccine efficacy was estimated to be null, with estimated HIV acquisition rates of 3.4 and 3.3 per 100 person-years in the vaccine and placebo arms, respectively, we conducted a pooled analysis of both treatment arms. The original baseline risk score analysis in \citealp{Gray2021} used a landmark time of 24 months. We considered landmark times of 18, 24 and 30 months post-randomization, corresponding to approximately 60\%, 36\% and 13\% of participants still at-risk, respectively. The percentage of participants diagnosed with HIV acquisition at these landmark times were 4.1\%, 5.1\% and 5.4\%, respectively. The majority of censored participants were subject to administrative censoring due to the termination of the trial. The estimated rates of loss to follow-up (early study termination for any reason, including death) at 18, 24 and 30 months were 6.4\%, 7.3\% and 7.5\%, respectively. For the C-index, we used a restriction time of 30 months.
	
	The baseline risk scores in \citet{Gray2021} were constructed separately in participants with male versus female sex assigned at birth. In addition to performing an analysis stratified by sex assigned at birth, we also investigated variable importance in the full cohort. The baseline variables of interest are summarized in Table \ref{tab:features} and detailed fully in \citet{Gray2021}. We considered sex assigned at birth (in the combined cohort only), body mass index, and age, as well as feature groups consisting of geographic variables, variables related to sexual health, variables related to behavior, and variables related to housing. These variable groupings were determined \textit{a priori} on the basis of subject matter knowledge. For convenience, we refer to these as ``feature groups" even when they only include a single feature. Missing covariate values were imputed following the procedure detailed in \citet{Gray2021}.   
	
	\begin{table}
		\centering 
		\begin{tabular}{cp{0.75\linewidth}}
			\toprule
			Feature group & Features included  \\ \midrule
			Geographic region &---\\
			Sex&--- \\
			Age&--- \\
			Body mass index&--- \\
			Sexual health&Prevalent sexually transmitted infection, genital sores, genital discharge \\
			Sexual behavior & Sexual orientation, married or have main sex partner, live with partner, partner has other partners, anal sex, condom use, unprotected sex with alcohol use, sex with HIV+ partner, unprotected sex with HIV+ partner, exchange services for sex \\
			Housing & Urban/rural, formal dwelling, home has 3+ services \\ \bottomrule
		\end{tabular}\caption{Features included in the HVTN 702 VIM analysis. Sex was as assigned at birth.}
		\label{tab:features}
	\end{table} 
	
	We implemented Algorithm \ref{alg:alternative, ss xfit} using five-fold cross-fitting with global survival stacking for nuisance estimation. The learner library was as described for the simulations in Section \ref{sec:sims}. Because cross-fitting and sample splitting introduce additional randomness into the inferential procedure, we followed the proposal of \citet{Dai2022}, which entails performing the procedure multiple times and aggregating the resulting $p$-values in order to stabilize inference. We adapted this proposal to our setting by performing 10 iterations of Algorithm \ref{alg:alternative, ss xfit}, thereby producing $p$-values $\{p_1,\dots,p_{10}\}$, and using the compound Bonferroni-geometric aggregation function proposed by \citet{Vovk2020}, in which the overall $p$-value is computed as $2\min\{10\min(p_1,\ldots,p_{10}, e \tilde{p})\}$ with $\tilde{p}$ the geometric mean of $p_1,\ldots,p_{10}$. This aggregation function has been shown to have good performance for both lightly correlated and highly correlated $p$-values. Confidence intervals were obtained by inverting the resulting hypothesis test. The point estimate presented is the average of the point estimates over random split repetitions. 
	
	We assessed the importance of each feature group relative to that feature group plus geographic variables, comparing a model using only geographic variables to a model using geographic variables plus the feature group of interest. 
	% We analyzed variable importance relative to geographic variables, where the reduced feature vector included only geographic variables, and the full feature vector included geographic variables plus the feature or feature group of interest. 
	This quantifies the gain in predictiveness from inclusion of the feature(s) of interest compared to a simple base model. We adjusted for treatment arm assignment by including it as a covariate in the estimation of the nuisance parameters, but did not otherwise include it as a predictor in the full and residual oracle prediction functions. In the Supplementary Material, we present the results of an analysis of VIM relative to the full covariate vector. 
	
	The analysis results for VIM relative to the feature group of interest plus geographic variables are shown in Fig. \ref{fig:702 marginal}. In the combined cohort, sex assigned at birth is estimated to be the most important feature in terms of both AUC\ and C-index predictiveness, and is deemed statistically significant for AUC\ at all three time horizons; one-sided hypothesis tests of zero importance yielded $p = 0.005$, $p = 0.012$, and $p = 0.004$ at 18, 24, and 30 months, respectively. This is unsurprising, as the HIV incidence estimates in \citet{Gray2021} were 4.3 acquisition events per 100 person-years in females and 1.3 acquisition events per 100 person-years in males. Sexual health features, sexual behavior features, and age are of secondary importance after sex assigned at birth, with housing features and body mass index ranking last. The feature group rankings and magnitudes are relatively stable over the three horizons for AUC. Among females, sexual health features have the largest estimated importance, although the magnitude of estimated importance is small for all feature groups and all predictiveness measures. Among males, sexual behavior features are estimated to be the most important feature group, with substantially larger point estimates than other feature groups, especially for AUC\ at the 18-month time horizon, for which a one-sided test of zero importance yielded $p = 0.034$. Notably, as there were only 37 HIV acquisition diagnoses observed among males, the confidence intervals in the male cohort are substantially wider than those in the other cohorts. Besides sexual behavior features in males, tests of non-zero importance in the female and male cohorts did not reach statistical significance.
	
	\begin{figure}
		\centering
		\includegraphics[width=\linewidth]{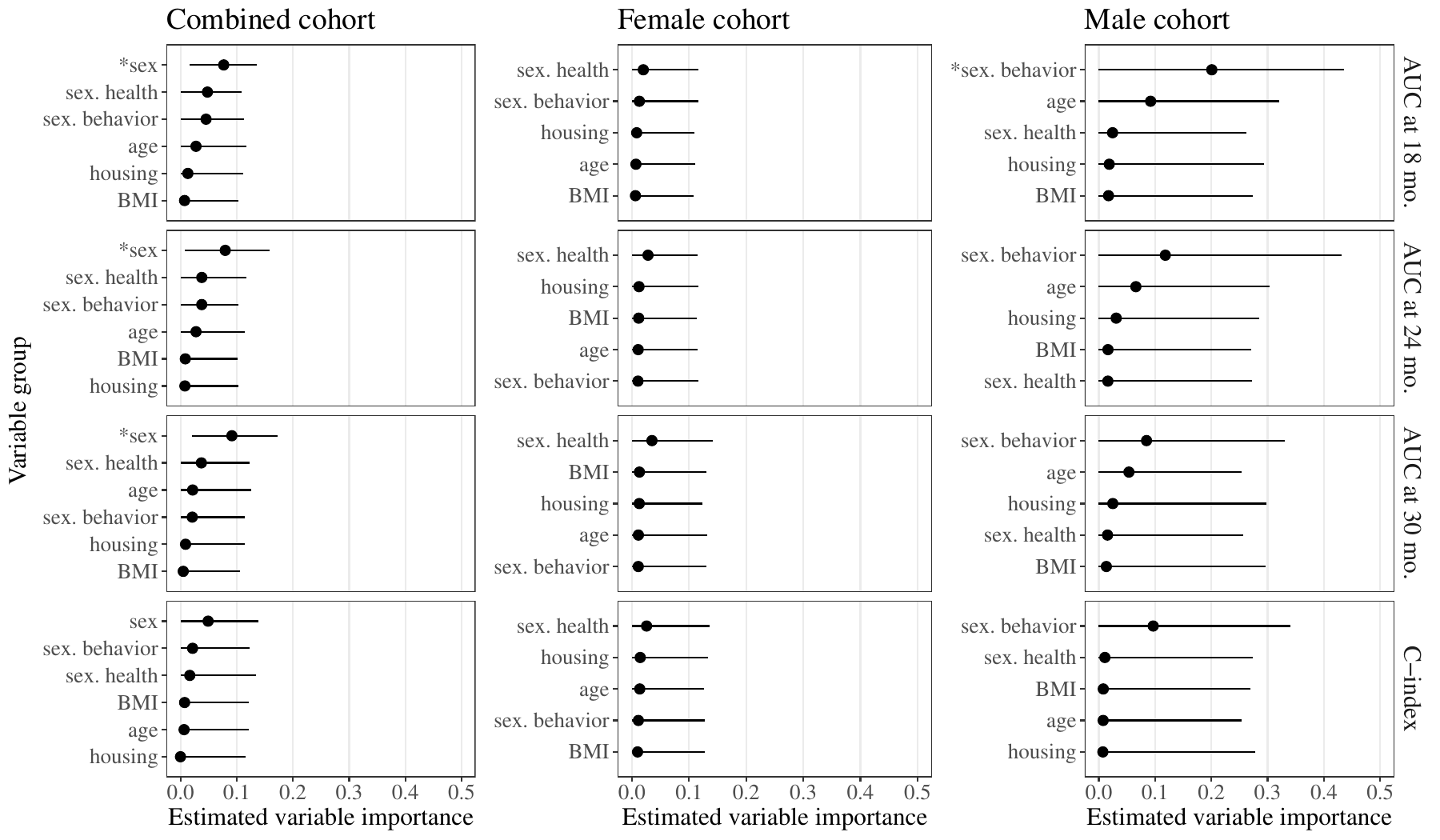}
		\caption{Variable importance in HVTN 702 evaluated relative to the feature group of interest plus geographic features. Rows correspond, from top to bottom, to AUC variable importance evaluated at 18, 24, and 30 months of follow-up, and C-index variable importance. Columns correspond, from left to right, to the combined male and female cohort, female cohort, and male cohort. Feature groups include sex assigned at birth (sex), age, body mass index (BMI), sexual health features, sexual behavior features, and housing features. Feature groups for which a one-sided test of zero importance reached nominal significance at a 0.05 level are marked with an asterisk.}
		\label{fig:702 marginal}
	\end{figure}
	
	\section{Concluding remarks}
	
	We have described a framework for nonparametric estimation and inference on variable importance in survival analysis. We studied a class of VIM parameters encompassing many existing predictiveness measures used in time-to-event settings. As long as a rich enough set of covariates is observed, measures in this class can be identified in terms of the observed data distribution, even under covariate-dependent right censoring. We provided a closed-form EIF for all measures within this class and proposed a debiased machine learning procedure using flexible estimation of nuisance functions. The resulting estimator enjoys doubly-robust consistency with regard to the nuisance functions $F_0$ and $G_0$, and when both $F_0$ and $G_0$ are consistently estimated, it is asymptotically normal when suitably centered and scaled and is also nonparametric efficient.
	
	The C-index predictiveness measure, for which a closed-form oracle prediction function is not readily available, presents theoretical issues that warrant future work. Our proposed inferential procedure requires that estimation of $f_0$ make no first-order contribution to the behavior of the one-step estimator, which can be established for AUC, Brier score, and survival time MSE under mild conditions. While our proposed numerical optimization procedure for the C-index performs well in numerical experiments, whether the second-order condition can be adapted to accommodate measures with no closed-form optimizer remains an open question. 
	
	Another limitation of this work concerns inference under the setting of zero importance. Sample splitting, while pragmatic, is not entirely satisfactory, as splitting the data introduces additional randomness and reduces power to detect non-zero alternatives. A theoretical analysis of the repeated splitting procedure \citep{Dai2022} may be fruitful; previous work on combining $p$-values for testing a single hypothesis has shown that power achieved by a given number of repeated splits and $p$-value aggregation method varies depending on the data-generating mechanism \citep{Vovk2020}. A principled approach to selecting the number of splits and aggregation method in a data-driven manner would be of practical utility. Alternatively, it may be possible to adapt the method of \citet{Hudson2023}, which does not involve sample splitting, to the survival setting.
	
	\section*{Acknowledgement}
	The authors thank the study participants and investigators of the HVTN 702 trial conducted by the HIV Vaccine Trials Network. The authors also thank Dr. Michele Andrasik for valuable scientific guidance. This work was supported by the National Science Foundation Graduate Research Fellowship Program under Grant No. DGE-2140004, by National Heart, Lung, and Blood Institute grant R01-HL137808, and by National Institute of Allergy and Infectious Diseases grants UM1-AI068635 and R37-AI029168. The content is solely the responsibility of the authors and does not necessarily represent the official views of the funding agencies.
	
	\section*{Software and code}
	\label{SM}
	Code to reproduce all simulations and data analyses is available online at \url{https://github.com/cwolock/surv_vim_supplementary}. The proposed methods are implemented in the \texttt{survML} package available at \url{https://github.com/cwolock/survML}.
	
	\section*{Appendix 1}\label{sec:conditions appendix}
	\subsection*{Additional technical conditions}
	We use $\norm{\cdot}_\infty$ and $\norm{\cdot}_v^*$ to denote the supremum norm and uniform sectional variation norm, respectively. We consider $\mathcal{F}$ to be endowed with some norm $\norm{\cdot}_{\mathcal{F}}$. For a generic $P \in \mathcal{M}_{\text{obs}}$ and $f \in \mathcal{F}$, we define the joint distribution function of $(f(X), T)$ evaluated at $(y_0, t_0) \in \mathcal{Y} \times (0,\tau_0]$ as $H_{P,f}: (y_0,t_0) \mapsto \int \I(f(u) \leq y_0)\,F_P(t_0 \midd x)\,Q_P(dx)$.
	As shorthand, we denote $H_{0,f}:=H_{P_0,f}$. We use $S_{n,k}$ and $L_{n,k}$ to denote the estimators of $S_0$ and $L_0$, respectively, that correspond to $F_{n,k}$.
	
	For any fixed event conditional survival function $S$ and corresponding event conditional distribution function $F$, and any fixed censoring conditional survival function $G$, we define $\pi_{S,G}: (z, t) \mapsto F(t \midd x) - \chi_{S,G}(z, t)$. 
	As shorthand, we use $\pi_0$ and $\pi_{n,k}$ to denote $\pi_{S_0, G_0}$ and $\pi_{S_{n,k}, G_{n,k}}$, respectively.
	Next, we define for any $f$ and any $\pi_{S,G}$,
	\begin{align*}
		\Gamma(f,\pi_{S,G})(z_1,\dots,z_m) :=  \int\cdots\int \omega[\{f(x_1), t_1\},\dots,\{f(x_m), t_m\}]\prod_{j=1}^{m}\pi_{S,G}(z_j, dt_j)\ .
	\end{align*}
	Using this notation, we define the following random functions for each $k \in \{1,\dots,K\}$:
	\begin{align*}
		&h_{1,n,k}(z_1,\dots,z_m) := \Gamma(f_{n,k}, \pi_{n,k})(z_1,\dots,z_m) - \Gamma(f_0, \pi_{0})(z_1,\dots,z_m)\ ; \\
		& h_{2,n,k}(z_1,\dots, z_m) := \Gamma(f_0, \pi_{n,k})(z_1,\dots,z_m) - \Gamma(f_0, \pi_0)(z_1,\dots,z_m)\ ;\\
		&h_{3,n,k}(z_1,\dots,z_m) := \int\cdots\int \big(\omega[\{f_{n,k}(x_1), t_1\},\dots,\{f_{n,k}(x_m),t_m\}] \\
		&\hspace{3cm}- \omega[\{f_{0}(x_1), t_1\},\dots,\{f_{0}(x_m),t_m\}]\big)\left\{\prod_{j=1}^{m}\pi_{n,k}(z_j, dt_j) -\prod_{j=1}^{m}\pi_{0}(z_j, dt_j) \right\};\\
		&h_{t,n,k}(x) := S_{n,k}(t \midd x)\int_0^t\left\{\frac{G_0(u \midd x)}{G_{n,k}(u \midd x)} - 1\right\}(L_{n,k} - L_0)(du \midd x)\ .
	\end{align*}
	We also define
	\begin{align*}
		\bar{h}_{n,k,1}&: z \mapsto \int \cdots \int  h_{1,n,k}(z, z_2, \dots, z_m)\prod_{j=2}^{m}P_0(dz_j)\ , \\
		\bar{h}_{n,k,2}&: z \mapsto  \int \cdots \int h_{2,n,k}(z, z_2, \dots, z_m)\prod_{j=2}^{m}P_0(dz_j)\ .
	\end{align*}

	\begin{condition}\label{condition:efficient influence function}
		\upshape
		There exists a dense subset $\mathcal{H}$ of $L_2^0(P_0)$ such that, for each $h \in \mathcal{H}$ and regular univariate parametric submodel $\{P_{0,\epsilon}\} \subset \mathcal{M}_{\text{obs}}$ through $P_0$ at $\epsilon = 0$ and with score for $\epsilon$ equal to $h$ at $\epsilon = 0$, the following conditions hold, with $f_{0,\epsilon}$ denoting $f_{P_{0,\epsilon}}$: (a) $V(f_{0,\epsilon}, P_\epsilon) - V(f_{0,\epsilon}, P_0) = V(f_{0}, P_\epsilon) - V(f_{0},P_0) + o(\epsilon)$; (b) $\epsilon \mapsto V(f_{0,\epsilon}, P_0)$ is differentiable in a neighborhood of $\epsilon = 0$; and (c) the optimizer $f_{0,\epsilon}$ is in $\mathcal{F}$ for small enough $\epsilon$. 
	\end{condition}
	
	\begin{condition}\label{condition:bounded away from zero} 
		\upshape
		There exists $\eta \in (0,\infty)$ such that, with $P_0$--probability tending to one and for $P_0$--almost all $x \in \mathcal{X}$, $G_{n,k}(\tau \midd x) \geq 1/\eta$ and $G_0(\tau \midd x) \geq 1/\eta$ for each $k \in \{1,\dots,K\}$. 
	\end{condition}
	
	\begin{condition}\label{condition:nuisance limits} 
		\upshape
		There exist $G_\infty$ and $S_\infty$ such that, for all $k$, 
		\begin{align*}
			E_{P_0}\left[\sup_{u \in (0, \tau]}\abs{\frac{1}{G_{n,k}(u \midd X)} - \frac{1}{G_\infty(u \midd X)}}\right]^2, \hspace{0.4cm}E_{P_0}\left[\sup_{u \in (0, \tau]}\sup_{v \in [0,u]}\abs{\frac{S_{n,k}(u \midd X)}{S_{n,k}(v \midd X)} - \frac{S_\infty(u \midd X)}{S_\infty(v \midd X)}}\right]^2,
		\end{align*}
		and $\|f_{n,k}(X) - f_0( X)\|_{\mathcal{F}}$ are $o_P(1)$.
	\end{condition}
	\begin{condition}
		\label{condition:double robust} For $P_0$--almost all $x \in \mathcal{X}$, there exist sets $\mathcal{S}_x, \mathcal{G}_x \subseteq (0, \tau]$ such that $\mathcal{S}_x \cup \mathcal{G}_x = (0, \tau]$, $L_\infty(u \midd x)  = L_0(u \midd x)$ for all $u \in \mathcal{S}_x$, and $G_\infty(u \midd x) = G_0 (u \midd x)$ for all $ u \in \mathcal{G}_x$. 
	\end{condition}
	\begin{condition} \label{condition:continuity of V} 
		\upshape
		There exists some constant $J_1$ such that, for each sequence $f_1, f_2,\ldots \in \mathcal{F}$ such that $\|f_j - f_0\|_{\mathcal{F}} \to 0$, $|V(f_j, H_0) - v_0| \leq J_1\|f_j - f_0\|_{\mathcal{F}}$ for each $j$ large enough. 
		
	\end{condition}
	\begin{condition} \label{condition:variation norm of kernel} 
		\upshape
		The functions $\omega$ and $\theta$ are bounded, and additionally, for all $l \leq m$, 
		\begin{align*}
			\sup_{(y_1, t_1), \dots, (y_{l-1}, t_{l-1})}&\norm{\omega_{0,f_0,l}\left\{(y_1, t_1), \dots, (y_{l-1}, t_{l-1}), (\cdot, \cdot)\right\}}_v^*, \quad 
			\sup_{t_1, \dots, t_{l-1}}&\norm{\theta_{0,l}\left(t_1,\dots,t_{l-1}, \cdot\right)}^*_v
		\end{align*}
		are finite,
		where we define $\theta_{0,l}: (t_1,\dots,t_l) \mapsto\int \cdots \int \theta\left(t_1,\dots, t_m\right)\prod_{j=l+1}^{m}H_0(dx_j, dt_j)$ and
		\begin{align*}
			&\omega_{0,f,l}: \{(y_1,t_1),\dots,(y_l,t_l)\} \mapsto\int \cdots \int \omega\left\{(y_1,t_1),\dots, (y_m, t_m)\right\}\prod_{j=l+1}^{m}H_{0,f}(dy_j, dt_j)\ .
		\end{align*}
		Furthermore, with $P_0$--probability tending to one, for each $k \in \{1,\dots,K\}$,
		\begin{align*}\sup_{(y_1, t_1), \dots, (y_{l-1}, t_{l-1})}&\norm{\omega_{0,f_{n,k},l}\left\{(y_1, t_1), \dots, (y_{l-1}, t_{l-1}), (\cdot, \cdot)\right\}}_v^*   < \infty\ .
		\end{align*}
	\end{condition}
	
	\begin{condition}\label{condition:optimality} 
		\upshape
		There exists some constant $J_2$ such that, for each sequence $f_1, f_2, \ldots \in \mathcal{F}$ such that $\|f_j - f_0\|_{\mathcal{F}} \to 0$, $|V(f_j, P_0) - V(f_0, P_0)| \leq J_2\|f_j - f_0\|^2_{\mathcal{F}}$ for  large enough $j$.
	\end{condition}
	
	\begin{condition} \label{condition:rate of convergence for f}  
		\upshape
		For all $k$, $\|f_{n,k} - f_0\|_{\mathcal{F}} = o_P(n^{-1/4})$.
	\end{condition}
	
	\begin{condition}\label{condition:weak consistency} 
		\upshape    
		For all $k$, $\int \bar{h}_{n,k,1}^2(z)P_0(dz)$ and $\int  \bar{h}_{n,k,2}^2(z)P_0(dz)$ are $o_P(1)$. 
	\end{condition}
	
	\begin{condition}\label{condition:second order remainder} 
		\upshape
		For all $k$, $\sup_{t \in (0, \tau]}\int |h_{t,n,k}(x)|Q_0(dx)$ and $\int h_{3,n,k}(z_1,\dots,z_m)\prod_{j=1}^{m}P_0(dz_j)$ are $o_P(n^{-1/2})$. 
	\end{condition}
	
	\section*{Appendix 2}
	\subsection*{Oracle prediction functions for stated example measures}
	\setcounter{example}{0}
	
	\begin{example}[\sc{AUC}]
		\upshape
		The ideal data maximizer of $f \mapsto \mathbbmsl{V}(f, \mathbbmsl{P}_0)$ is the population conditional mean of the indicator variable $\I(T > \tau)$ given $X$, $\mathbbmsl{f}_0: x \mapsto \pr_{\mathbbmsl{P}_0}(T > \tau \midd X = x)$ \citep{Agarwal2014}, and so, $f_0$ is $x\mapsto 1-F_0(\tau\midd x)$. Similarly, $\mathbbmsl{f}_{0,s}$ equals $x \mapsto \pr_{\mathbbmsl{P}_0}(T > \tau \midd X_{-s} = x_{-s})$, and so, $f_{0,s}$ is $x \mapsto E_{P_0}\left\{f_{0}(X) \midd X_{-s} = x_{-s}\right\}$. 
	\end{example}
	
	\begin{example}[\sc{Brier score}]
		\upshape
		The Brier score at time $\tau$ is equivalent to the negative MSE for predicting the binary outcome $\I(T > \tau)$. As such, the oracle prediction function is the conditional survival probability $\mathbbmsl{f}_0: x \mapsto \pr_{\mathbbmsl{P}_0}(T > \tau\midd X = x)$, and so, $f_0$ is $x \mapsto 1-F_0(\tau \midd x)$. As for AUC, $f_{0,s}$ is $x \mapsto E_{P_0}\left\{f_{0}(X) \midd X_{-s} = x_{-s}\right\}$.
	\end{example}
	
	\begin{example}[\sc{Survival time MSE}]
		\upshape
		As for the Brier score, the negative MSE is maximized by the conditional mean, which here corresponds to $x \mapsto E_{\mathbbmsl{P}_0}(T \wedge \tau \midd X = x)$, and so, $f_0: x \mapsto \int_0^\tau \left\{1-F_0(t \midd x)\right\} dt$. The residual oracle prediction function is $x \mapsto E_{\mathbbmsl{P}_0}(T \wedge \tau \midd X_{-s} = x_{-s})$, with identified counterpart given by $f_{0,s}: x \mapsto E_{P_0}\left\{f_{0}(X) \midd X_{-s} = x_{-s}\right\}$. 
	\end{example}
	
	\begin{example}[\sc{C-index}]
		\upshape
		If $f$ is a prediction function such that $f(x_1) \geq f(x_2)$ implies 
		\begin{align*}
			\pr_{\mathbbmsl{P}_0}\left(T_1 < T_2, T_1 \leq \tau\midd X_1 = x_1, X_2 = x_2\right) > \pr_{\mathbbmsl{P}_0}\left(T_2 < T_1, T_2 \leq \tau\midd X_1 = x_1, X_2 = x_2\right) 
		\end{align*}
		for all $(x_1, x_2) \in \mathcal{X}^2$, then $f$ is an oracle prediction function; details are provided in the Supplementary Material. However, \citet{Elgui2023} exhibit a family of distributions for which no such $f$ exists. 
	\end{example}
	
	\section*{Appendix 3}\label{sec:algo appendix}
	
	\subsection*{Algorithms for estimation and inference}
	\begin{algorithm}[h!]
		\caption{Cross-fitted inference on VIM value $\psi_{0,s}$ (non-zero importance)}\label{alg:alternative, xfit}
		\begin{algorithmic}[1]
			\State Select approximation time grid $\mathcal{B} := \{t_1,\dots, t_J\}$ with $t_J \geq \tau$. In following steps, approximate integrals as Riemann sums on $\mathcal{B}$.
			\State Generate $W_n \in \{1,\dots,K\}^n$ by sampling uniformly from $\{1,\dots,K\}$ with replacement. For $k = 1,\dots, K$, denote by $\mathcal{D}_k$ the subset of observations with index in $\mathcal{I}_k := \{i: W_{n,i} =k\}$. 
			\For{$k = 1,\dots,K$}
			\State  Using only data in $\cup_{j \neq k}\mathcal{D}_j$, construct estimators $F_{n,k}$ and $G_{n,k}$ of $F_0$ and $G_0$, respectively, on $\mathcal{B}$, and construct estimators $f_{n,k}$ and $f_{n,k,s}$ of $f_0$ and $f_{0,s}$, respectively.
			\State  Substitute $f_{n,k}$, $f_{n,k,s}$, $F_{n,k}$ and $G_{n,k}$ for $f_0$, $f_{0,s}$, $F_0$ and $G_0$ in the form of $\phi_0$ and $\phi_{0,s}$ to construct $\phi_{n,k}$ and $\phi_{n,k,s}$, respectively.
			\State Compute $H_{n,k}^*:(x_0,t_0)\mapsto n_k^{-1}\sum_{i \in \mathcal{I}_k}\varphi_{n,k,x_0,t_0}(Z_i)$ with  $F_{n,k}$ and $G_{n,k}$ used in $\varphi_{n,k,x_0,t_0}$.
			\State Compute $v_{n,1,k}^* := V_1(f_{n,k}, H_{n,k}^*)$, $v_{n,2,k}^* := V_2(H_{n,k}^*)$, $v_{n,1,k,s}^* := V_1(f_{n,k,s},H_{n,k}^*)$, and $\sigma^2_{n,k,s} := n_k^{-1}\sum_{i \in \mathcal{I}_k}\left\{\phi_{n,k}(Z_i) - \phi_{n,k,s}(Z_i)\right\}^2$.
			\EndFor
			\State Compute VIM point estimate $\psi_{n,s}^* := (\sum_{k=1}^{K}v_{n,2,k}^* )^{-1}\sum_{k=1}^{K}(v_{n,1,k}^* - v_{n,1,k,s}^*)$ and corresponding variance estimate $\sigma^2_{n,s} := K^{-1}\sum_{k=1}^{K}\sigma^2_{n,k,s}$.
		\end{algorithmic}
	\end{algorithm}
	
	\begin{algorithm}[h!]
		\caption{Sample-split, cross-fitted inference on VIM value $\psi_{0,s}$}\label{alg:alternative, ss xfit}
		\begin{algorithmic}[1]
			\State  Select approximation time grid $\mathcal{B} := \{t_1, \dots, t_J\}$, where $t_J \geq \tau$. In following steps, 
			approximate integrals as Riemann sums on $\mathcal{B}$.
			\State Generate $W_n \in \{1,\dots,2K\}^n$ by sampling uniformly from $\{1,\dots,2K\}$ with replacement. For $k = 1,\dots, 2K$, denote by $\mathcal{D}_k$ the subset of observations with index in $\mathcal{I}_k:=\{i: W_{n,i} =k\}$. Let $n_s$ denote the number of observations in $\cup_{k \text{ even}}\mathcal{D}_k$ and set $n_s^* := n - n_k$.
			\For{$k = 1, \dots, K$}
			\State Using only data in $\cup_{j \neq k}\mathcal{D}_j$, construct estimators $F_{n,k}$ and $G_{n,k}$ of $F_0$ and $G_0$, respectively, on $\mathcal{B}$. In addition, construct estimators $f_{n,k}$ and $f_{n,k,s}$ of $f_0$ and $f_{0,s}$, respectively. 
			\State Substitute $f_{n,k}$, $f_{n,k,s}$, $F_{n,k}$ and $G_{n,k}$ for $f_0$, $f_{0,s}$, $F_0$ and $G_0$ in the form of $\phi_0$ and $\phi_{0,s}$ to construct $\phi_{n,k}$ and $\phi_{n,k,s}$, respectively.
			\State Compute $H_{n,k}^*:(x_0,t_0)\mapsto n_k^{-1}\sum_{i \in \mathcal{I}_k}\varphi_{n,k,x_0,t_0}(Z_i)$ with  $F_{n,k}$ and $G_{n,k}$ used in $\varphi_{n,k,x_0,t_0}$.
			\State If $k$ is odd, compute $v_{n,1,k}^* := V_1(f_{n,k}, H_{n,k}^*)$, $v_{n,2,k}^* := V_2(H_{n,k}^*)$, and $\tilde{\sigma}^2_{n,k} := n_k^{-1}\sum_{i \in \mathcal{I}_k}\phi_{n,k}(Z_i)^2$. If $k$ is even, compute $v_{n,1,k,s}^* := V_1(f_{n,k,s}, H_{n,k}^*)$, $v_{n,2,k}^* := V_2(H_{n,k}^*)$, and $\tilde{\sigma}^2_{n,k,s} := n_k^{-1}\sum_{i \in \mathcal{I}_k}\phi_{n,k,s}(Z_i)^2$.
			\EndFor
			\State Compute VIM point estimate $\tilde{\psi}_{n,s} := (\sum_{k=1}^{K}v_{n,2,2k-1})^{-1}\sum_{k=1}^{K}v_{n,1,2k-1} - (\sum_{k=1}^{K}v_{n,2,2k})^{-1}\sum_{k=1}^{K}v_{n,1,2k,s}$ and corresponding variance estimate $\tilde{\sigma}^2_{n,s} := (Kn_s^*)^{-1}\sum_{k=1}^{K}\tilde{\sigma}^2_{n,2k-1} + (Kn_s)^{-1}\sum_{k=1}^{K}\tilde{\sigma}^2_{n,2k,s}$. 
		\end{algorithmic}
	\end{algorithm}
	
	\clearpage
	\singlespacing
	
	\bibliography{refs}
	
	\clearpage
	
	\begin{center}
		{\large \textbf{Supplementary Material}}
	\end{center}
	
	\doublespacing
	
	\noindent In Section \ref{sec:proofs}, we provide proofs of the theoretical results in the main text along with technical lemmas. Section \ref{sec:additional technical details} contains additional technical details, including verification of key conditions for the example predictiveness measures. In Section \ref{sec:gradient boosted c}, we outline a procedure using gradient boosting to numerically maximize the C-index. In Section \ref{sec:addtl sims}, we give details and full results for the main numerical experiments.  In Section \ref{sec:addtl data analysis}, we provide additional analysis results for the HVTN 702 study. In Section \ref{sec:robustness}, we discuss and illustrate the robustness properties of our proposed estimation procedure. Finally, in Section \ref{sec:permutation}, we compare our approach to a permutation-based variable importance procedure.
	
	\setcounter{section}{0}
	\setcounter{figure}{0}
	\setcounter{table}{0}
	\renewcommand{\thesection}{S\arabic{section}}
	\renewcommand{\thefigure}{S\arabic{figure}}
	\renewcommand{\thetable}{S\arabic{table}}
	\section{Proofs of theorems}\label{sec:proofs}
	
	\subsection{Additional notation}
	We use $E_0$ to denote an expectation taken with respect to $P_0$, that is, $E_0[f(Z)] := E_{P_0}[f(Z)]$. We also adopt the empirical process notation $P^mf := \int \cdots \int f(z_1,\dots,z_m)\prod_{j=1}^mP(dz_j)$ for probability measure $P$. Inequalities and equalities involving random quantities are taken to hold pointwise, that is, for every element in the sample space, unless otherwise specified.  
	
	\subsection{Discussion of technical conditions}\label{subsec:regularity conditions}
	Conditions \ref{condition:independent censoring} and \ref{condition:region of integration} stated and discussed in Section \ref{sec:survival VIMs} of the main text concern the identifiability of the ideal data parameter $\mathbbmsl{V}(\mathbbmsl{f}_0, \mathbbmsl{P}_0)$. In this subsection, we provide additional discussion on the conditions under which Theorems \ref{thm:efficient influence function}--\ref{thm:asymptotic linearity} hold. 
	
	Condition \ref{condition:efficient influence function} matches that of Theorem 3 in \citet{Williamson2021b}. Condition \ref{condition:bounded away from zero} requires that $G_n$ and $G_0$ be uniformly bounded away from zero. Condition \ref{condition:nuisance limits} requires that $F_n$ and $G_n$ converge to fixed limits, and that $f_n$ converge to $f_0$. Conditions \ref{condition:nuisance limits} and \ref{condition:double robust} together imply that, for almost all $t \in (0, \tau]$ and $x \in \mathcal{X}$, either $S_n$ or $G_n$ is consistent. Condition \ref{condition:continuity of V} essentially requires that the map $f \mapsto V(f, H_0)$ be continuous about $f_0$ with respect to  $\norm{\cdot}_{\mathcal{F}}$. Condition \ref{condition:variation norm of kernel} requires that $\omega$ and $\theta$ be bounded and also places restrictions on the uniform sectional variation norm of $\omega_{0,f_0,l}$, $\omega_{0, f_{n,k},l}$ and $\theta_{0,f_0,l}$ for all $l \leq m$ with respect to one argument. This is significantly weaker than requiring a finite uniform sectional variation norm with respect to all arguments. Details for stated example measures are provided in Section \ref{sec:additional technical details}. 
	
	Condition \ref{condition:optimality} formalizes the requirement that estimation of $f_0$ have no first-order contribution. We provide details for the example measures in Section \ref{sec:additional technical details}. Condition \ref{condition:rate of convergence for f} requires that $f_0$ be estimated at a sufficiently fast rate. Condition \ref{condition:weak consistency} ensures asymptotic negligibility of a set of empirical process terms. Condition \ref{condition:second order remainder} 
	requires that two remainder terms, depending on $(G_{n,k} - G_0)(L_{n,k} - L_0)$ and $(f_{n,k} - f_0)(\pi_{n,k} - \pi_0)$, tend to zero at rate faster than $n^{-1/2}$.
	
	\subsection{Identification}
	\begin{proof}[*thm:ID]
		We begin by showing that for all $t \in (0, \tau_0]$ and $P_0$--almost all $x \in \mathcal{X}$, $\mathbbmsl{L}_0(t \midd x) = L_0(t \midd x)$ as long as Conditions \ref{condition:independent censoring} and \ref{condition:region of integration}(a) hold. To see this, we begin by using standard probability rules to write
		\begin{align*}
			M_{0,1}(u \midd x) = P_0(Y \leq u, \Delta =1 \midd X = x)&= \mathbbmsl{P}_0(T \leq u, T \leq C \midd X = x) \\
			&= \int_0^u \mathbbmsl{P}_0(C \geq t \midd X=x)\mathbbmsl{F}_0(dt \midd x)\ .
		\end{align*}
		Here, we have used the conditional independence of $T$ and $C$ given $X$. Using the same conditional independence, we have that $1 - M_0(u^- \midd x) = \mathbbmsl{P}_0(C \geq u\midd X = x)\left\{1-\mathbbmsl{F}_0(u^- \midd x)\right\}$. Condition \ref{condition:region of integration}(a) implies that $\mathbbmsl{P}_0(C \geq u \midd X = x) > 0$ for all $u \in (0, \tau_0]$, and since $t \in (0, \tau_0]$ we can write
		\begin{align*}
			L_0(t \midd x) = \int_0^t \frac{M_{0,1}(du \midd x)}{1 - M_0(u^- \midd x)} &= \int_0^t\frac{ \mathbbmsl{P}_0(C \geq u\midd X=x)\mathbbmsl{F}_0(du \midd x)}{\mathbbmsl{P}_0(C \geq u\midd X= x)\left\{1-\mathbbmsl{F}_0(u^- \midd x)\right\}} \\
			&= \int_0^t\frac{\mathbbmsl{F}_0(du \midd x)}{1-\mathbbmsl{F}_0(u^- \midd x)} = \mathbbmsl{L}_0(t \midd x)\ .
		\end{align*}
		It follows that $\mathbbmsl{F}_0(t \midd x) = F_0(t \midd x)$ for all $t \in (0, \tau_0]$ and $P_0$--almost all $x \in \mathcal{X}$. 
		
		Next, we consider $\mathbbmsl{V}_1(f, \mathbbmsl{H}_0)$. We partition $\{\mathcal{X} \times (0, \infty)\}^m$ into regions $\mathcal{R} := \{\mathcal{X} \times (0, \tau_0]\}^m$ and its complement $\mathcal{R}^C$. In light of Condition \ref{condition:region of integration}(b), there exists some real-valued function $\tilde{\omega}$ on $\mathcal{Y}^m$ such that 
		\begin{align*}
			&\mathbbmsl{V}_1(f, \mathbbmsl{H}_0) = \int \cdots \int \omega((f(x_1), t_1), \dots, (f(x_m),t_m))\prod_{j=1}^{m}\mathbbmsl{H}_0(dx_j, dt_j)\\
			&= \int_{\mathcal{R}}\omega((f(x_1), t_1), \dots, (f(x_m),t_m))\prod_{j=1}^{m}\mathbbmsl{H}_0(dx_j, dt_j) \\
			&\hspace{1cm}+ \int_{\mathcal{R}^C}\omega((f(x_1), t_1), \dots, (f(x_m),t_m))\prod_{j=1}^{m}\mathbbmsl{H}_0(dx_j, dt_j)\\
			&= \int_{\mathcal{R}}\omega((f(x_1), t_1), \dots, (f(x_m),t_m))\prod_{j=1}^{m}\mathbbmsl{F}_0(dt_j\midd x_j)\mathbbmsl{Q}_0(dx_j)\\
			&\hspace{1cm}+ \int_{\mathcal{X}^m}\tilde{\omega}(f(x_1), \dots, f(x_m))\prod_{j=1}^{m}\mathbbmsl{Q}_0(dx_j)\ .
		\end{align*}
		Using the fact that $\mathbbmsl{Q}_0 = Q_0$ and that $\mathbbmsl{F}_0(t \midd x) = F_0(t \midd x)$ for $t \in (0, \tau_0]$ and $P_0$--almost all $x \in \mathcal{X}$, it follows that
		\begin{align*}
			&\int_{\mathcal{R}}\omega((f(x_1), t_1), \dots, (f(x_m),t_m))\prod_{j=1}^{m}\mathbbmsl{F}_0(dt_j\midd x_j)\mathbbmsl{Q}_0(dx_j)\\
			&\hspace{2cm}+ \int_{\mathcal{X}^m}\tilde{\omega}(f(x_1), \dots, f(x_m))\prod_{j=1}^{m}\mathbbmsl{Q}_0(dx_j)\\
			&= \int_{\mathcal{R}}\omega((f(x_1), t_1), \dots, (f(x_m),t_m))\prod_{j=1}^{m}F_0(dt_j\midd x_j)Q_0(dx_j) \\
			&\hspace{2cm}+ \int_{\mathcal{X}^m}\tilde{\omega}(f(x_1), \dots, f(x_m))\prod_{j=1}^{m}Q_0(dx_j)\\
			&= \int \cdots \int \omega((f(x_1), t_1), \dots, (f(x_m),t_m))\prod_{j=1}^{m}H_0(dx_j, dt_j) = V_1(f, H_0)\ .
		\end{align*}
		An identical argument applies for $\mathbbmsl{V}_2(\mathbbmsl{H}_0).$
	\end{proof}
	
	\subsection{Efficiency calculations}
	
	\begin{proof}[*thm:efficient influence function] 
		In the following, with a slight abuse of notation we let $P_0(y,\delta \midd x)$ denote the joint distribution function of $(Y,\Delta)$ given $X = x$ evaluated at $(y,\delta)$.
		
		Let $\{P_\epsilon\}$ be a suitably smooth and bounded Hellinger differentiable path with $P_{\epsilon = 0} = P_0$ and score function $\dot{\ell}_0$ at $\epsilon = 0$. We will repeatedly make use of that fact that $\dot{\ell}_0(x, y, \delta) = \dot{\ell}_0(x) + \dot{\ell}_0(y,\delta \midd x)$, and that for any function $b:\mathcal{X}\rightarrow\mathbb{R}$, $
		\iiint b(x) \dot{\ell}_0(y, \delta \midd x)P_0(dy, d\delta \midd x) Q_0(dx) = 0$. 
		We begin by using the quotient rule to write
		\begin{align}
			\left.\frac{d}{d\epsilon}V(f_0, P_\epsilon)\right|_{\epsilon = 0} &= \left.\frac{d}{d\epsilon}\frac{V_1(f_0, H_\epsilon)}{V_2( H_\epsilon)}\right|_{\epsilon = 0} = \frac{\left.\frac{d}{d\epsilon}V_1(f_0, H_\epsilon)\right|_{\epsilon = 0}}{v_{0,2}} - \frac{V_1(f_0,  H_0)\left.\frac{d}{d\epsilon}V_2(H_\epsilon)\right|_{\epsilon = 0}}{v_{0,2}^2}\ . \label{eq:quotient rule}
		\end{align}
		We study the two derivatives above separately.  Under appropriate boundedness conditions, we have that
		\begin{align}
			&\left.\frac{d}{d\epsilon}V_1(f_0,H_\epsilon)\right|_{\epsilon = 0 }\ =\ \left.\frac{d}{d\epsilon}\int\cdots \int \omega\left(((f_0(x_1),t_1),\dots, (f_0(x_m), t_m)\right)\prod_{j=1}^{m}H_\epsilon(dx_j, dt_j)\right|_{\epsilon = 0}\nonumber\\
			&\hspace{1cm}=\ \left.\frac{d}{d\epsilon}\int\cdots \int \omega\left(((f_0(x_1),t_1),\dots, (f_0(x_m), t_m)\right)\prod_{j=1}^{m}F_\epsilon(dt_j \midd x_j)Q_\epsilon(dx_j)\right|_{\epsilon = 0}\nonumber\\
			&\hspace{1cm}=\ \int\cdots \int \omega\left((f_0(x_1),t_1),\dots, (f_0(x_m),t_m)\right) \left.\frac{d}{d\epsilon}\prod_{j=1}^{m}F_\epsilon(d t_j \midd x_j)\right|_{\epsilon = 0}\prod_{j=1}^{m}Q_0(dx_j)\nonumber \\&\hspace{1.7cm}+\int\cdots \int\omega\left((f_0(x_1),t_1),\dots, (f_0(x_m),t_m)\right) \prod_{j=1}^{m}F_0(d t_j \midd x_j)\left.\frac{d}{d\epsilon}\prod_{j=1}^{m}Q_\epsilon(dx_j)\right|_{\epsilon = 0} \label{eq:pathwise derivative calculation}.
		\end{align}
		Using the product rule and symmetry of $\omega$, the second term in (\ref{eq:pathwise derivative calculation}) is equal to
		\begin{align*}
			&m\int\cdots \int \omega\left((f_0(x_1),t_1),\dots, (f_0(x_m),t_m)\right)\prod_{j=1}^{m}F_0(d t_j \midd x_j) \prod_{j=2}^{m}Q_0(dx_j)\left.\frac{d}{d\epsilon}Q_\epsilon(dx_1)\right|_{\epsilon = 0}\\
			&=m\int\cdots \int \omega\left((f_0(x_1),t_1),\dots, (f_0(x_m),t_m)\right)\prod_{j=1}^{m}F_0(d t_j \midd x_j) \prod_{j=2}^{m}Q_0(dx_j)\dot{\ell}_0(x)Q_0(dx_1)\\
			&=m\int\cdots \int \Big\{ \omega\left((f_0(x_1),t_1),\dots, (f_0(x_m),t_m)\right)\\
			& \hspace{3cm}\times \prod_{j=1}^{m}F_0(d t_j \midd x_j) \prod_{j=2}^{m}Q_0(dx_j)\dot{\ell}_0(x_1,y,\delta)P_0(dy, d\delta \midd x_1)Q_0(dx_1)\Big\}\ ,
		\end{align*}
		where the final equality follows from properties of score functions. We can rewrite this last expression as 
		\begin{align*}
			&m\int\cdots \int  \omega\left((f_0(x_1),t_1),\dots, (f_0(x_m),t_m)\right)F_0(d t_1 \midd x_1) \prod_{j=2}^{m}H_0(dx_j, dt_j)\dot{\ell}_0(x,y,\delta)P_0(dy, d\delta, dx_1)\\
			&= m\int \cdots \int \omega_{0,1}\left(x_1,t_1\right)F_0(d t_1 \midd x_1)\dot{\ell}_0(x,y,\delta)P_0(dy, d\delta, dx_1)\ ,
		\end{align*}
		where, as in the main text, we define $\omega_{0,1}(x,t) := \omega_{0,f_0,1}(f_0(x),t)$.
		Therefore, this term contributes 
		\begin{align}
			x \mapsto m\int\omega_{0,1}\left(x,t\right)F_0(d t \midd x)\label{eq:efficient influence function contribution 1}
		\end{align}
		to the efficient influence function. 
		
		Now, by definition we have that
		\begin{align*}
			\left.\frac{d}{d\epsilon}F_\epsilon(t \midd x)\right|_{\epsilon = 0} =- \left.\frac{d}{d\epsilon}\Prodi_{(0,t]}\left\{1 - L_{\epsilon}(du \midd x)\right\}\right|_{\epsilon = 0} .
		\end{align*}
		By Theorem 8 of \citet{Gill1990}, the product integral map $h \mapsto \Theta_h(t) :=  \prodi_{(0, t]}\{1-h(dt)\}$ is Hadamard differentiable with respect to the supremum norm with derivative given by $\alpha \mapsto \Theta_h(t) \int_0^t\frac{\Theta_h(u^-)}{\Theta_h(u)}\alpha(du)$. By the chain rule for functional derivatives, we have that
		\begin{align*}
			\left.\frac{d}{d\epsilon}F_\epsilon(t \midd x)\right|_{\epsilon = 0} = -\left.\frac{d}{d\epsilon}\Prodi_{(0,t]}\left\{1 - L_{\epsilon}(du \midd x)\right\}\right|_{\epsilon = 0} = S_0(t \midd x) \int_0^t\frac{S_0(u^- \midd x)}{S_0(u \midd x)}\left.\frac{d}{d\epsilon}L_\epsilon (du \midd x)\right|_{\epsilon = 0}.
		\end{align*}
		Now, using the quotient rule, we can write
		\begin{align*}
			\left.\frac{d}{d\epsilon}L_\epsilon(t \midd x)\right|_{\epsilon = 0} &= \int_0^t\left.\frac{d}{d\epsilon}\frac{M_{\epsilon, 1}(du \midd x)}{1 - M_\epsilon(u^- \midd x)}\right|_{\epsilon = 0}\\
			&= \int\frac{\I_{[0,t]}(u)\left.\frac{d}{d\epsilon}M_{\epsilon, 1}(du \midd x)\right|_{\epsilon = 0}}{1 - M_0(u^- \midd x)} + \int\frac{\I_{[0,t]}(u)\left.\frac{d}{d\epsilon}M_\epsilon(u^- \midd x)\right|_{\epsilon = 0}M_{0,1}(du \midd x)}{\left\{1 - M_0(u^- \midd x)\right\}^2}\ .
		\end{align*}
		Computing each of the derivatives above, we have that
		\begin{align*}
			\left.\frac{d}{d\epsilon}M_{\epsilon, 1}(u \midd x)\right|_{\epsilon = 0} = \left.\frac{d}{d\epsilon}P_\epsilon(Y \leq u, \Delta =1 \midd X = x)\right|_{\epsilon = 0}&= \left.\frac{d}{d\epsilon}\iint \I_{[0,u]}(y)\delta P_\epsilon(dy, d\delta \midd x)\right|_{\epsilon = 0}\\
			&= \iint \I_{[0,u]}(y)\delta \dot{\ell}_0(y, \delta \midd x)P_0(dy, d\delta \midd x)\ ,
		\end{align*}
		and also that
		\begin{align*}
			&\left.\frac{d}{d\epsilon}M_\epsilon(u^- \midd x)\right|_{\epsilon = 0} = \left.\frac{d}{d\epsilon}P_\epsilon (Y < u \midd X=x)\right|_{\epsilon = 0}= \left.\frac{d}{d\epsilon}\int \I_{[0,u)}(y)M_\epsilon(dy \midd x)\right|_{\epsilon = 0}\\
			&\hspace{2cm}= \int \I_{[0,u)}(y)\dot{\ell}_0(y \midd x)M_0(dy \midd x) \\
			&\hspace{2cm}= \iint \I_{[0,u)}(y)\dot{\ell}_0(y \midd x)P_0(dy, d\delta \midd x)  + \iint \I_{[0,u)}(y)\dot{\ell}_0(\delta \midd y,x)P_0(dy, d\delta \midd x) \\
			&\hspace{2cm}= \iint \I_{[0,u)}(y) \dot{\ell}_0(y, \delta \midd x)P_0(dy, d\delta \midd x)\ ,
		\end{align*}
		where the last two equalities follow from properties of score functions. Therefore, we find that
		\begin{align*}
			\left.\frac{d}{d\epsilon}L_{\epsilon}(t \midd x)\right|_{\epsilon = 0} &= \int\frac{\I_{[0,t]}(u)\int \delta \dot{\ell}_0(u,\delta \midd x)P_0(du,d\delta \midd x)}{1 - M_0(u^- \midd x)} \\
			&\hspace{1cm}+ \int\frac{\I_{[0,t]}(u)\I_{[0,u)}(y)\iint \dot{\ell}_0(y,\delta \midd x)P_0(dy, d\delta \midd x)M_{0,1}(du \midd x)}{\left\{1 - M_0(u^- \midd x)\right\}^2}\\
			&= \iint\frac{\I_{[0,t]}(y)\delta \dot{\ell}_0(y,\delta \midd x)P_0(dy,d\delta \midd x)}{1 - M_0(y^- \midd x)} \\
			&\hspace{1cm}+ \iiint\frac{\I_{[0,t]}(u)\I_{[0,u)}(y)\dot{\ell}_0(y,\delta \midd x)P_0(dy, d\delta \midd x)M_{0,1}(du \midd x)}{\left\{1 - M_0(u^- \midd x)\right\}^2}\\
			&= \iint\frac{\I_{[0,t]}(y)\delta \dot{\ell}_0(y,\delta \midd x)P_0(dy,d\delta \midd x)}{1 - M_0(y^- \midd x)} \\
			&\hspace{1cm}+ \iiint\frac{\I_{[0,t]}(u)\I_{[0,u)}(y)\dot{\ell}_0(y,\delta \midd x)P_0(dy, d\delta \midd x)L_{0}(du \midd x)}{1 - M_0(u^- \midd x)}\ .
		\end{align*}
		Again using the fact that scores are $P_0$--centered and observing that $\I(y > u) = 1 - \I(u \leq y)$, we plug this expression into the derivative of the product integral to obtain
		\begin{align*}
			\left.\frac{d}{d\epsilon}F_{\epsilon}(t \midd x)\right|_{\epsilon=0} &= S_0(t \midd x) \iint\frac{\I_{[0,t]}(u)S_0(u^- \midd x)}{S_0(u \midd x)\left\{1 - M_0(u^- \midd x)\right\}}\delta \dot{\ell}_0(u, \delta \midd x)P_0(du, d\delta \midd x) \\
			&\hspace{0.5cm}-S_0(t \midd x) \iiint\frac{\I_{[0,t]}(u)\I_{[0,u)}(y)S_0(u^- \midd x)}{S_0(u \midd x)\left\{1 - M_0(u^- \midd x)\right\}} \dot{\ell}_0(y, \delta \midd x)P_0(dy, d\delta \midd x)L_0(du \midd x)\\
			&= S_0(t \midd x) \iint\frac{\delta \I_{[0,t]}(y)S_0(y^- \midd x)}{S_0(y \midd x)\left\{1 - M_0(y^- \midd x)\right\}}\dot{\ell}_0(y, \delta \midd x)P_0(dy, d\delta \midd x) \\
			&\hspace{0.5cm}-S_0(t \midd x) \iint\int_0^{t \wedge y}\frac{S_0(u^- \midd x)L_0(du \midd x)}{S_0(u \midd x)\left\{1 - M_0(u^- \midd x)\right\}} \dot{\ell}_0(y, \delta \midd x)P_0(dy, d\delta \midd x)\\
			&= S_0(t \midd x) \iint\frac{\delta \I_{[0,t]}(y)}{S_0(y \midd x)G_0(y \midd x)}\dot{\ell}_0(y, \delta \midd x)P_0(dy, d\delta \midd x) \\
			&\hspace{0.5cm}-S_0(t \midd x) \iint\int_0^{t \wedge y}\frac{L_0(du \midd x)}{S_0(u \midd x)G_0(u \midd x)} \dot{\ell}_0(y, \delta \midd x)P_0(dy, d\delta \midd x)\\
			&=- \iint \chi_0(z,t)\dot{\ell}_0(y, \delta \midd x)P_0(dy, d\delta \midd x)\ .
		\end{align*}
		Next, we note that 
		\begin{align*}
			E_0[\chi_0(Z,t) \midd X = x]&= -S_0(t \midd x)\int \left[\frac{\delta \I_{[0,t]}(y)}{S_0(y \midd x)G_0(y \midd x)} - \int_0^{t \wedge y}\frac{L_0(du \midd x)}{S_0(u \midd x)G_0(u \midd x)}\right]P_0(dy, d\delta \midd x)\\
			&= -S_0(t \midd x)\left[\int_0^t \frac{M_{0,1}(dy \midd x)}{S_0(y \midd x)G_0(y \midd x)} - \int_0^{t}\frac{\left\{1 - M_0(u^- \midd x)\right\}L_0(du \midd x)}{S_0(u \midd x)G_0(u \midd x)}\right]\\
			&= -S_0(t \midd x)\left[\int_0^t \frac{M_{0,1}(dy \midd x)}{S_0(y \midd x)G_0(y \midd x)} - \int_0^{t}\frac{M_{0,1}(du \midd x)}{S_0(u \midd x)G_0(u \midd x)}\right]\ = 0\ .
		\end{align*}
		By properties of score functions, this implies that
		\begin{align*}
			\left.\frac{d}{d\epsilon}F_\epsilon(t \midd x)\right |_{\epsilon=0} &=-\iint \chi_0(z,t)\dot{\ell}_0(y, \delta \midd x)P_0(dy, d\delta \midd x) - \iint \chi_0(z,t)\dot{\ell}_0( x)P_0(dy, d\delta \midd x)\\
			&= -\iint \chi_0(z,t)\dot{\ell}_0(y, \delta,x)P_0(dy, d\delta \midd x)
		\end{align*}since $\iint \chi_0(z,t)\dot{\ell}_0( x)P_0(dy, d\delta \midd x)=0$.
		Hence, using the product rule and symmetry of $\omega$, the first summand in (\ref{eq:pathwise derivative calculation}) equals
		\begin{align*}
			&-m\int\cdots \int \Big\{\omega\left((f_0(x_1),t_1),\dots,(f_0(x_m),t_m)\right) \chi_0(z_1,dt_1)\\
			&\hspace{4cm}\times \prod_{j=2}^{m}H_0(dx_j,dt_j)\dot{\ell}_0(y_1,\delta_1,x_1)P_0(dy_1, d\delta_1 \midd x_1)Q_0(dx_1)\Big\}\\
			&= -m\int\cdots \int \omega_{0,1}\left(x_1,t_1\right) \chi_0(z_1,dt_1)\dot{\ell}_0(y_1,\delta_1,x_1)P_0(dy_1, d\delta_1 \midd x_1)Q_0(dx_1)\ ,
		\end{align*}
		yielding that the efficient influence function contribution from this term is
		\begin{align}
			z \mapsto -m\int\omega_{0,1}\left(x,t\right) \chi_{0}(z,dt)\ .\label{eq:efficient influence function contribution 2}
		\end{align}
		Next, we note that 
		\begin{align*}
			\chi_{0}(z,dt) &= -\left\{\frac{\delta \I_{[0,t]}(y)}{S_0(y \midd x) G_0(y \midd x)} - \int_0^{t \wedge y}\frac{L_0(du \midd x)}{S_0(u \midd x)G_0(u \midd x))}\right\}S_0(dt \midd x)\\
			&\hspace{1cm}-S_0(t \midd x)\frac{\delta \gamma(t - y)}{S_0(y \midd x) G_0(y \midd x)} + S_0(t \midd x)\frac{\I_{[0,y]}(t)L_0(dt \midd x)}{S_0(t \midd x)G_0(t \midd x)}\ ,
		\end{align*}
		where $\gamma$ is the Dirac delta function.
		It follows then that
		\begin{align*}
			\int \omega_{0,1}\left(x,t\right)\chi_{0}(z,dt) &= -\int \omega_{0,1}\left(x,t\right) \left\{\frac{\delta \I_{[0,t]}(y) }{S_0(y \midd x) G_0(y\midd x)} - \int_0^{t \wedge y}\frac{L_0(du \midd x)}{S_0(u \midd x)G_0(u \midd x))}\right\}S_0(dt \midd x)\\
			&\hspace{1cm}-\omega_{0,1}\left(x,t\right)\frac{S_0(y \midd x)\delta}{S_0(y \midd x) G_0(y \midd x)} + \int \omega_{0,1}\left(x,t\right)S_0(t \midd x)\frac{\I_{[0,y]}(t)L_0(dt \midd x)}{S_0(t \midd x)G_0(t \midd x)}\\
			&= -\int \omega_{0,1}\left(x,t\right) \left\{\frac{\delta \I_{[0,t]}(y)}{S_0(y \midd x) G_0(y\midd x)} - \int_0^{t \wedge y}\frac{L_0(du \midd x)}{S_0(u \midd x)G_0(u \midd x))}\right\}S_0(dt \midd x)\\
			&\hspace{1cm}-\omega_{0,1}\left(x,t\right)\frac{\delta}{ G_0(y \midd x)} + \int \omega_{0,1}\left(x,t\right)\frac{\I_{[0,y]}(t)L_0(dt \midd x)}{G_0(t \midd x)}\ .
		\end{align*}
		Taking an expectation over $(Y,\Delta)$ given $X = x$, we have that
		\begin{align*}
			&E_0\left[\int \omega_{0,1}\left(x,t\right)\chi_{0}(Z,dt)\middle| X =x\right] \\
			&= -\int \omega_{0,1}\left(x,t\right) \left\{\int_0^t\frac{M_{0,1}(dy \midd x)}{S_0(y \midd x) G_0(y\midd x)} - \int_0^{t}\frac{\left\{1 - M_0(u^- \midd x)\right\}L_0(du \midd x)}{S_0(u \midd x)G_0(u \midd x))}\right\}S_0(dt \midd x)\\
			&\hspace{1cm}-\int \omega_{0,1}\left(x,t\right)\frac{M_{0,1}(dy \midd x)}{ G_0(y \midd x)} + \int \omega_{0,1}\left(x,t\right)\frac{\left\{1 - M_0(t^- \midd x)\right\}L_0(dt \midd x)}{G_0(t \midd x)}\\
			&= -\int \omega_{0,1}\left(x,t\right) \left\{\int_0^t\frac{M_{0,1}(dy \midd x)}{S_0(y \midd x) G_0(y\midd x)} - \int_0^{t}\frac{M_{0,1}(du \midd x)}{S_0(u \midd x)G_0(u \midd x))}\right\}S_0(dt \midd x)\\
			&\hspace{1cm}-\int \omega_{0,1}\left(x,t\right)\frac{M_{0,1}(dy \midd x)}{ G_0(y \midd x)} + \int \omega_{0,1}\left(x,t\right)\frac{M_{0,1}(dt \midd x)}{G_0(t \midd x)} = 0\ .
		\end{align*}
		The tower rule then implies that (\ref{eq:efficient influence function contribution 2}) is already $P_0$--centered. To center the efficient influence function overall, we subtract $mv_{0,1}$ from (\ref{eq:efficient influence function contribution 1}), yielding the centered gradient $\phi_{\omega,0}$. The derivation of the gradient of $V_2$ is identical, so combining these results with (\ref{eq:quotient rule}) yields the overall gradient
		\begin{align*}
			\phi_0: z \mapsto \frac{\phi_{\omega,0}(z)-v_{0}\phi_{\theta,0}(z)}{v_{0,2}} \ .
		\end{align*}
		It remains to show that $P \mapsto V(f_P, P)$ and $P \mapsto V(f_0, P)$ have the same efficient influence function at $P_0$. This was established in \citet{Williamson2021b}, but we provide details here for completeness. We note that
		\begin{align*}
			V(f_\epsilon, P_\epsilon) &- V(f_0, P_0) = V(f_\epsilon, P_\epsilon) - V(f_0, P_\epsilon) + V(f_0, P_\epsilon) -  V(f_0, P_0)\ .
		\end{align*}
		Under Condition \ref{condition:efficient influence function}(a), we have that $V(f_\epsilon, P_\epsilon) - V(f_0, P_\epsilon)  = V(f_\epsilon, P_0) - V(f_0, P_0)  + o(\epsilon)$. Furthermore, under Conditions \ref{condition:efficient influence function}(b) and \ref{condition:efficient influence function}(c), we have that
		\begin{align*}
			\left.\frac{d}{d\epsilon}V(f_\epsilon, P_0)\right|_{\epsilon = 0} = 0\ ,
		\end{align*}
		and so Taylor's theorem yields that $V(f_\epsilon, P_0) - V(f_0, P_0) = o(\epsilon)$. Since the parameter $P \mapsto V(f_0, P)$ is pathwise differentiable at $P_0$ with canonical gradient $\phi_0$, we can write that $V(f_0, P_\epsilon) -  V(f_0, P_0) = \epsilon \int \phi_0(z)\dot{\ell}_0(z)P_0(dz) + O(\epsilon^2)$. Combining all these observations, we have that $ V_1(f_\epsilon, P_\epsilon) - V(f_0, P_0) = \epsilon \int \phi_0(z)\dot{\ell}_0(z)P_0(dz)  + o(\epsilon)$. Hence, $P \mapsto V(f_P, P)$ is pathwise differentiable at $P_0$ with efficient influence function $\phi_0$. 
	\end{proof}
	
	\begin{lemma}\label{lemma:efficient influence function of joint CDF}
		For any $f \in \mathcal{F}$, if there exists $\eta \in (0, \infty)$ such that $G_0(t_0 \midd x) \geq 1/\eta$ for $P_0$--almost all $x \in \mathcal{X}$ such that $f(x) \leq y_0$, then $P \mapsto H_{P,f}(y_0,t_0)$ is pathwise differentiable at $P_0$ relative to the nonparametric model $\mathcal{M}_{\text{obs}}$, with efficient influence function given by $z \mapsto \bar{\varphi}_{0,f,y_0,t_0}(z):= \varphi_{0,f,y_0,t_0}(z) - H_{P,f}(y_0,t_0)$, where 
		\begin{align*}
			\varphi_{0,f,y_0,t_0}(z) := \I(f(x) \leq y_0)\left\{F_0(t_0 \midd x) - \chi_{0}(z,t_0)\right\}.
		\end{align*}
		If there exists $\eta \in (0, \infty)$ such that $G_0(t_0 \midd x) \geq 1/\eta$ for $P_0$--almost all $x$ such that $x \leq x_0$, then $P \mapsto H_{P}(x_0,t_0)$ is pathwise differentiable at $P_0$ relative to $\mathcal{M}_{\text{obs}}$, with efficient influence function given by $z \mapsto \bar{\varphi}_{0,x_0,t_0}(z):= \varphi_{0,x_0,t_0}(z) - H_{P}(x_0,t_0)$, where 
		\begin{align*}
			\varphi_{0,x_0,t_0}(z) := \I(x \leq x_0)\left\{F_0(t_0 \midd x) - \chi_{0}(z,t_0)\right\}.
		\end{align*}
	\end{lemma}
	
	\begin{proof}[*lemma:efficient influence function of joint CDF]
		As in the proof of Theorem \ref{thm:efficient influence function}, we proceed by direct calculation of the pathwise derivative of the parameter. Again, we let $\{P_\epsilon\}$ be a suitably smooth and bounded Hellinger differentiable path with $P_{\epsilon = 0} = P_0$ and score function $\dot{\ell}_0$ at $\epsilon = 0$. We begin by writing
		\begin{align*}
			H_{\epsilon,f}(y_0,t_0) = \int \I(f(x) \leq y_0)\,F_\epsilon(t_0 \midd x)\,Q_\epsilon(dx)\ .
		\end{align*}
		Under appropriate boundedness conditions, we have that
		\begin{align}
			&\left.\frac{d}{d\epsilon}H_{\epsilon,f}(y_0,t_0)\right|_{\epsilon = 0} = \left.\frac{d}{d\epsilon}\int \I(f(x) \leq y_0)\,F_\epsilon(t_0 \midd x)\,Q_\epsilon(dx) \right|_{\epsilon = 0} \nonumber \\
			&\hspace{0.4cm}= \int \dots \int \I(f(x) \leq y_0)\left.\frac{d}{d\epsilon}F_\epsilon(t_0 \midd x)\right|_{\epsilon = 0}Q_0(dx)  + \int \I(f(x) \leq y_0)\,F_0(t_0 \midd x) \left.\frac{d}{d\epsilon}Q_\epsilon(dx) \right|_{\epsilon = 0}\ .\label{eq:pathwise derivative joint cdf}
		\end{align}
		Using the same argument as in the proof of Theorem \ref{thm:efficient influence function}, the second term in (\ref{eq:pathwise derivative joint cdf}) contributes $x \mapsto \I(f(x) \leq y_0)F_0(t_0 \midd x)$
		to the efficient influence function. For the first term, we recall that $\left.\frac{d}{d\epsilon}F_\epsilon(t_0 \midd x) \right|_{\epsilon = 0} = -\iint \chi_{0}(z,t_0)\dot{\ell}_0(y,\delta,x)P_0(dy,d\delta,dx)$. Therefore, the first term in (\ref{eq:pathwise derivative joint cdf}) contributes $z \mapsto -\I(f(x) \leq y_0)\chi_{0}(z,t_0)$. This portion of the efficient influence function is already $P_0$--centered. To center the efficient influence function overall, we subtract $H_{0,f}(y_0,t_0)$, yielding the overall efficient influence function
		\begin{align*}
			\bar{\varphi}_{0,f}(y_0,t_0): z \mapsto \I(f(x) \leq y_0)\left\{F_0(t_0 \midd x) - \chi_{0}(z,t_0)\right\} - H_{0,f}(y_0,t_0)\ .
		\end{align*}
		The derivation of the efficient influence function of $P \mapsto H_P(x_0,t_0)$ at $P_0$ is identical. 
	\end{proof}
	
	\subsection{Theoretical analysis of the joint distribution function estimator}
	
	For generic conditional event time survival function $S$ and conditional censoring survival function $G$, we define
	\begin{align*}
		\chi_{S,G}(z,t) := -S(t \midd x) \left[\frac{\delta\I_{[0,t]}(y)}{S(y \midd x)G(y \midd x)} - \int_0^{t \wedge y}\frac{L(du \midd x)}{S(u \midd x)G(u \midd x)}\right]
	\end{align*}
	with $L$ the conditional cumulative hazard function corresponding to $S$. We then let $\varphi_{S,G,f,y_0,t_0}:z \mapsto \I(f(x) \leq y_0)\{1 - S(t \midd x) - \chi_{S,G}(z,t_0)\}$. As shorthand, we denote $\chi_{n,k} := \chi_{S_{n,k}, G_{n,k}}$, $\varphi_{n,k,f,y_0,t_0} := \varphi_{S_{n,k}, G_{n,k},f,y_0,t_0}$, $\chi_{\infty} := \chi_{S_{\infty}, G_{\infty}}$, and $\varphi_{\infty,f,y_0,t_0} := \varphi_{S_\infty, G_{\infty},f,y_0,t_0}$. 
	
	\begin{lemma}\label{lemma:duhamel}
		For any conditional event time survival function $S$ and corresponding cumulative hazard function $L$, any conditional censoring survival function $G$, and any $f \in \mathcal{F}$, it holds that
		\begin{align*}
			&P_0\varphi_{S,G,f,y_0,t_0} - H_{0,f}(y_0,t_0) \\
			&\hspace{.8cm}= E_0\left[\I(f(X) \leq y_0)S(t_0\midd X)\int_0^{t_0}\frac{S_0(u^-\midd X)}{S(u \midd X)}\left\{\frac{G_0(u \midd X)}{G(u \midd X)} - 1\right\}(L - L_0)(du \midd X)\right].
		\end{align*}
		
	\end{lemma}
	\begin{proof}[*lemma:duhamel] 
		On one hand, we note that
		\begin{align*}
			&E_0\left[\frac{\Delta\I_{[0,t_0]}(Y)}{S(Y \midd x)G(Y \midd x)} - \int_0^{t_0 \wedge Y}\frac{L(du \midd x)}{S(u \midd x)G(u \midd x)}\,\middle|\, X = x\right]\\
			&\hspace{0.8cm}= -\int_0^{t_0}\frac{S_0(y^- \midd x)G_0(y \midd x)}{S(y \midd x)G(y \midd x)}(L - L_0)(du \midd x)\ ,
		\end{align*}
		which implies that $E_0\left[\chi_{S,G}(Z,t) \midd X = x\right] = S(t \midd x)\int_0^{t_0}\frac{S_0(y^- \midd x)G_0(y \midd x)}{S(y \midd x)G(y \midd x)}(L - L_0)(du \midd x)$.
		On the other hand, $F(t_0 \midd x) - F_0(t_0 \midd x) = S_0(t_0 \midd x) - S(t_0 \midd x)$, and in view of the Duhamel equation (Theorem 6 of \citealp{Gill1990}), we have that
		\begin{align*}
			S_0(t_0 \midd x) - S(t_0 \midd x) = S(t_0 \midd x)\int_0^{t_0}\frac{S_0(u^- \midd x)}{S(u \midd x)}(L  - L_0)(du \midd x)\ .
		\end{align*}
		Combining these results, we find that
		\begin{align*}
			&F(t_0 \midd x) - F_0(t_0 \midd x) - E_0\left[\chi_{S,G}(Z,t_0) \midd X = x\right]\\
			&\hspace{0.8cm}= S(t_0 \midd x)\int_0^{t_0}\frac{S_0(u^- \midd x)}{S(u \midd x)}(L  - L_0)(du \midd x) - S(t_0\midd x)\int_0^{t_0}\frac{S_0(u^- \midd x)G_0(u \midd x)}{S(u \midd x)G(u \midd x)}(L - L_0)(du \midd x)\\
			&\hspace{0.8cm}= -S(t_0 \midd x) \int_0^{t_0}\frac{S_0(u^- \midd x)}{S(u \midd x)}\left\{\frac{G_0(u \midd x)}{G(u \midd x)} - 1\right\} (L - L_0)(du \midd x)\ .
		\end{align*}
		Finally, applying the tower property, we conclude that
		\begin{align*}
			&P_0\varphi_{S,G,f,y_0,t_0} - H_{0,f}(y_0,t_0) \\
			&\hspace{0.8cm}= E_0\left[\I(f(X) \leq y_0)\left\{F(t_0 \midd X) - \chi_{S,G}(Z,t_0)\right\}\right] - E_0\left[\I(f(X) \leq y_0)F_0(t_0 \midd X)\right]\\
			&\hspace{.8cm}= E_0\left[\I(f(X) \leq y_0)\left\{F(t_0 \midd X) - \chi_{S,G}(Z,t_0) - F_0(t_0 \midd X)\right\}\right]\\
			&\hspace{.8cm}= E_0\left[\I(f(X) \leq y_0)E_0\left[F(t_0 \midd X) - \chi_{S,G}(Z,t_0) - F_0(t_0 \midd X) \midd X = x\right]\right]\\
			&\hspace{.8cm}= -E_0\left[\I(f(X) \leq y_0)S(t_0 \midd X) \int_0^{t_0}\frac{S_0(u^- \midd X)}{S(u \midd X)}\left\{\frac{G_0(u \midd X)}{G(u \midd X)} - 1\right\} (L - L_0)(du \midd X)\right].
		\end{align*}
	\end{proof}
	
	\begin{lemma}\label{lemma:DR}
		If Condition \ref{condition:double robust} holds, then $P_0\varphi_{\infty, f,y_0,t_0} = H_{0,f}(y_0,t_0)$ for any $f \in \mathcal{F}$.
	\end{lemma}
	\begin{proof}[*lemma:DR] 
		By Lemma \ref{lemma:duhamel}, we have that
		\begin{align*}
			&P_0\varphi_{\infty,f,y_0,t_0} - H_{0,f}(y_0,t_0) \\
			&\hspace{.8cm}= E_0\left[\I(f(X) \leq y_0)S(t_0\midd X)\int_0^{t_0}\frac{S_0(u^-\midd X)}{S_\infty(u \midd X)}\left\{\frac{G_0(u \midd X)}{G_\infty(u \midd X)} - 1\right\}(L_\infty - L_0)(du \midd X)\right].
		\end{align*}
		As long as $t_0 \in (0, \tau]$, for each value $X = x$, we can use Condition \ref{condition:double robust} to decompose the interval $[0,t_0]$ into $\mathcal{S}_x\cup\mathcal{G}_x$. For any $ u \in \mathcal{S}_x$, we have that $L_\infty = L_0$ so that $(L_\infty - L_0)(du \midd x) = 0$. For any $u \notin \mathcal{S}_x$, we have by assumption that  $u \in \mathcal{G}_x$ so that $G_\infty = G_0$ and $\frac{G_0(u \midd x)}{G_\infty(u \midd x)} - 1 = 0$. Therefore, the integral over $(0, t_0] \subseteq \mathcal{S}_x\cup\mathcal{S}^C_x$ is equal to 0. 
	\end{proof}
	
	We now give some results regarding the large-sample behavior of the cross-fitted joint distribution function estimator. We denote by $W_n \in \{1,\dots, K\}^n$ a random vector generated by sampling uniformly from $\{1,\dots,K\}$ with replacement. We let $\mathcal{D}_k$ denote the subset of observations with index in $\mathcal{I}_k := \{i: W_{n,i} = k\}$ for $k = 1,\dots, K$ and let $n_k := |\mathcal{I}_k|$. We let $S_{n,k}$ denote an estimator of $S_0$ constructed using the data $\cup_{j\neq k}\mathcal{D}_j$, and likewise for $G_{n,k}$ and $f_{n,k}$. We let $P_{n,k} := (S_{n,k}, G_{n,k})$. For generic $f \in \mathcal{F}$, we define $H^*_{n,k,f}(y_0,t_0):= n_k^{-1}\sum_{i \in \mathcal{I}_{k}}\varphi_{n,k,f,y_0,t_0}(Z_i)$.

	\begin{lemma}\label{lemma:sup L2 discrepancy of nuisances}
		Under Condition \ref{condition:bounded away from zero}, there exists a universal constant $J_\eta$ depending only on $\eta$ such that, for each $k$ and for any $f \in \mathcal{F}$, it holds that
		\begin{align*}
			&P_0\left(\sup_{y_0 \in \mathcal{Y}, t_0 \in (0, \tau]}\abs{\varphi_{n,k,f,y_0,t_0} - \varphi_{\infty,f,y_0,t_0}}\right)^2\\
			&\leq J^2_\eta\left[E_0\left[\sup_{u \in (0, \tau], v \in [0,u]}\abs{\frac{S_\infty(u \midd X)}{S_\infty(v \midd X)} - \frac{S_{n,k}(u \midd X)}{S_{n,k}(v \midd X)}} \right]^2\vee E_0\left[\sup_{u \in (0, \tau]}\abs{\frac{1}{G_{n,k}(u \midd X)} - \frac{1}{G_\infty(u \midd X)}}\right]^2\right]
		\end{align*}
		with probability tending to one. 
	\end{lemma}
	\begin{proof}[*lemma:sup L2 discrepancy of nuisances] 
		Throughout this proof, for a random function $g_{n,k}$, $E_0[g_{n,k}(Z)]$ is understood to be an expectation taken with respect to the distribution of the random data unit $Z$, i.e., $E_0[g_{n,k}(Z)] = \int g_{n,k}(z)P_0(dz)$.
		
		We start with the decomposition $\varphi_{n,k,f,y_0,t_0}(z) - \varphi_{\infty,f,y_0,t_0}(z) = \I(f(x) \leq y_0)\sum_{j=1}^{5}A_{n,k,j}(t_0)(z)$, where we define
		\begin{align*}
			&A_{n,k,1}(t_0)(z) := F_{n,k}(t_0 \midd x) - F_\infty(t_0 \midd x)\ ;\\
			&A_{n,k,2}(t_0)(z) := \frac{\delta\I_{[0,t_0]}(y)}{G_\infty(y \midd x)}\left\{\frac{S_{n,k}(t_0 \midd x)}{S_{n,k}(y \midd x)} - \frac{S_\infty(t_0 \midd x)}{S_\infty(y \midd x)}\right\};\\
			&A_{n,k,3}(t_0)(z) := \frac{\delta\I_{[0,t_0]}(y)S_{n,k}(t_0 \midd x)}{S_{n,k}(y \midd x)}\left\{\frac{1}{G_{n,k}(y \midd x)} - \frac{1}{G_\infty(y \midd x)}\right\};\\
			&A_{n,k,4}(t_0)(z) := -\int_0^{t_0 \wedge y}\left\{\frac{1}{G_{n,k}(u \midd x)} - \frac{1}{G_\infty(u \midd x)}\right\}\frac{S_{\infty}(t_0 \midd x)L_{\infty}(du \midd x)}{S_\infty(u \midd x)}\ ;\\
			&A_{n,k,5}(t_0)(z) := -\int_0^{t_0 \wedge y} \frac{1}{G_{n,k}(u \midd x)}\left\{\frac{S_{n,k}(t_0 \midd x)L_{n,k}(du \midd x)}{S_{n,k}(u \midd x)} -\frac{S_{\infty}(t_0 \midd x)L_{\infty}(du \midd x)}{S_{\infty}(u \midd x)}\right\}.
		\end{align*}
		Using the triangle inequality, we have that
		\begin{align*}
			&\left[P_0\left(\sup_{y_0 \in \mathcal{Y}, t_0 \in (0, \tau]}\abs{\varphi_{n,k,f,y_0,t_0} - \varphi_{\infty,f,y_0,t_0}}\right)^2\right]^{\frac{1}{2}} \\
			&\hspace{0.8cm}\leq \sum_{j=1}^{5}\left(E_0\left[\sup_{y_0 \in \mathcal{Y}, t_0 \in (0, \tau]}\abs{\I(f(X)\leq y_0)A_{n,k,j}(t_0)(Z)}\right]^2\right)^{\frac{1}{2}}\\
			&\hspace{0.8cm}\leq \sum_{j=1}^{5}\left(E_0\left[\sup_{t_0 \in (0, \tau]}A_{n,k,j}^2(t_0)(Z)\right]\right)^{\frac{1}{2}}= \sum_{j=1}^{5}\left(P_0\left[\sup_{t_0 \in (0, \tau]}A_{n,k,j}^2(t_0)\right]\right)^{\frac{1}{2}}.
		\end{align*}
		We proceed by bounding each term individually. First, since $S_{n,k}(0 \midd x) = S_\infty(0 \midd x) = 1$, we have that
		\begin{align*}
			P_0\left[\sup_{t_0 \in (0, \tau]}A_{n,k,1}^2(t_0)\right] &= E_0\left[\sup_{t_0 \in (0, \tau]}\left\{S_\infty(t_0 \midd X) - S_{n,k}(t_0 \midd X)\right\}^2 \right] \\
			&= E_0\left[\sup_{t_0 \in (0, \tau]}\left\{\frac{S_\infty(t_0 \midd X)}{S_\infty(0 \midd X)} - \frac{S_{n,k}(t_0 \midd X)}{S_{n,k}(0 \midd X)}\right\}^2 \right]\\
			&\leq \sup_{u \in [0,t_0]} E_0\left[\sup_{t_0 \in (0, \tau]}\left\{\frac{S_\infty(t_0 \midd X)}{S_\infty(u \midd X)} - \frac{S_{n,k}(t_0 \midd X)}{S_{n,k}(u \midd X)}\right\}^2\right]\\
			&\leq E_0\left[\sup_{t_0 \in (0, \tau], u \in [0,t_0]}\left\{\frac{S_\infty(t_0 \midd X)}{S_\infty(u \midd X)} - \frac{S_{n,k}(t_0 \midd X)}{S_{n,k}(u \midd X)}\right\}^2\right]\\
			&\leq  E_0\left[\sup_{u \in (0, \tau], v \in [0,u]}\abs{\frac{S_\infty(u \midd X)}{S_\infty(v \midd X)} - \frac{S_{n,k}(u \midd X)}{S_{n,k}(v \midd X)}} \right]^2.
		\end{align*}
		For $A_{n,k,2}$, we use the fact that $1/G_\infty \leq \eta$, so that
		\begin{align*}
			P_0\left[\sup_{t_0 \in (0, \tau]}A_{n,k,2}^2(t_0)\right] &= E_0\left[\sup_{t_0 \in (0, \tau]}\left\{\frac{\delta\I_{[0,t_0]}(Y)}{G_\infty(Y \midd X)}\left(\frac{S_{n,k}(t_0 \midd X)}{S_{n,k}(Y \midd X)} - \frac{S_\infty(t_0 \midd X)}{S_\infty(Y \midd X)}\right)\right\}^2\right]\\
			&\leq \eta^2 E_0\left[\sup_{t_0 \in (0, \tau]}\left\{\frac{S_{n,k}(t_0 \midd X)}{S_{n,k}(Y \midd X)} - \frac{S_\infty(t_0 \midd X)}{S_\infty(Y \midd X)}\right\}^2\right]\\
			&\leq \eta^2 E_0\left[\sup_{t_0 \in (0, \tau], u \in [0,t_0]}\left\{\frac{S_\infty(t_0 \midd X)}{S_\infty(u \midd X)} - \frac{S_{n,k}(t_0 \midd X)}{S_{n,k}(u \midd X)}\right\}^2\right]\\
			&\leq \eta^2 E_0\left[\sup_{u \in (0, \tau], v \in [0,u]}\abs{\frac{S_\infty(u \midd X)}{S_\infty(v \midd X)} - \frac{S_{n,k}(u \midd X)}{S_{n,k}(v \midd X)}} \right]^2.
		\end{align*}
		Next, for $A_{n,k,3}$, we note that $y \leq t_0$ implies $S_{n,k}(y \midd x) \geq S_{n,k}(t_0 \midd x)$ so that  $\frac{\I_{[0, t_0]}(y) S_{n,k}(t_0 \midd x)}{S_{n,k}(y)} \leq 1$. Hence, we have that
		\begin{align*}
			P_0\left[\sup_{t_0 \in (0, \tau]}A_{n,k,3}^2(t_0)\right] &= E_0\left[\sup_{t_0 \in (0, \tau]}\left\{\frac{\delta\I_{[0,t_0]}(Y)S_{n,k}(t_0 \midd X)}{S_{n,k}(Y \midd X)}\left(\frac{1}{G_{n,k}(Y \midd X)} - \frac{1}{G_\infty(Y \midd X)}\right)\right\}^2\right]\\
			&\leq E_0\left[\sup_{u \in (0, \tau]}\abs{\frac{1}{G_{n,k}(u \midd X)} - \frac{1}{G_\infty(u \midd X)}}\right]^2.
		\end{align*}
		For $A_{n,k,4}$, we have that
		\begin{align*}
			&P_0\left[\sup_{t_0 \in (0, \tau]}A_{n,k,4}^2(t_0)\right] \\
			&\hspace{.8cm}=E_0\left[\sup_{t_0 \in (0, \tau]}\left\{\int_0^{t_0 \wedge Y}\left(\frac{1}{G_{n,k}(u \midd X)} - \frac{1}{G_\infty(u \midd X)}\right)\frac{S_\infty(t_0 \midd X)L_\infty(du \midd X)}{S_\infty(u \midd X)}\right\}^2\right]\\
			&\hspace{.8cm}\leq E_0\left[\sup_{u \in (0, \tau]}\left(\frac{1}{G_{n,k}(u \midd X)} - \frac{1}{G_\infty(u \midd X)}\right)^2\sup_{t_0 \in (0, \tau]}\left\{\int_0^{t_0 \wedge Y}\frac{S_\infty(t_0 \midd X)L_\infty(du \midd X)}{S_\infty(u \midd X)}\right\}^2\right].
		\end{align*}
		Using the backwards equation (Theorem 5 of \citealp{Gill1990}), we have that $\int_0^t \frac{S(t)L(du)}{S(u)} = 1 - S(t)$ for any survival function $S$ and corresponding cumulative hazard $L$. Continuing from above, we can then write
		\begin{align*}
			P_0\left[\sup_{t_0 \in (0, \tau]}A_{n,k,4}^2(t_0)\right] &\leq E_0\left[\sup_{u \in (0, \tau]}\left(\frac{1}{G_{n,k}(u \midd X)} - \frac{1}{G_\infty(u \midd X)}\right)^2\sup_{t_0 \in (0, \tau]}\left\{1 - S_\infty(t_0 \wedge T \midd X)\right\}^2\right]\\
			&\leq E_0\left[\sup_{u \in (0, \tau]}\left(\frac{1}{G_{n,k}(u \midd X)} - \frac{1}{G_\infty(u \midd X)}\right)^2\right]\\
			&\leq  E_0\left[\sup_{u \in (0, \tau]}\abs{\frac{1}{G_{n,k}(u \midd X)} - \frac{1}{G_\infty(u \midd X)}}\right]^2.
		\end{align*}
		For $A_{n,k,5}$, we first define $\alpha_{n,k,t}(u \midd x) := \frac{S_{n,k}(t \midd x)}{S_{n,k}(u \midd x)}$ and $\alpha_{\infty, t}(u \midd x) := \frac{S_\infty(t \midd x)}{S_\infty (u \midd x)}$. Again, using the backwards equation, we have that $\alpha_{n,k,t}(du \midd x) = \frac{S_{n,k}(t \midd x)L_{n,k}(du \midd x)}{S_{n,k}(u \midd x)}$ and $\alpha_{\infty,t}(du \midd x) = \frac{S_{\infty}(t \midd x)L_{\infty}(du \midd x)}{S_{\infty}(u \midd x)}$. Using integration by parts, we can then write that, with probability tending to one,
		\begin{align*}
			&P_0\left[\sup_{t_0 \in (0, \tau]}A_{n,k,5}^2(t_0)\right] = E_0\left[\sup_{t_0 \in (0, \tau]}\left\{\int_0^{t_0 \wedge Y} \frac{1}{G_{n,k}(u \midd X)}\left(\alpha_{n,k,t_0}(du \midd X) - \alpha_{\infty, t_0}(du \midd X)\right)\right\}^2\right]\\
			&= E_0\Bigg[\sup_{t_0 \in (0, \tau]}\bigg\{\frac{\alpha_{n,k,t_0}(t_0 \wedge Y \midd X) - \alpha_{\infty,t_0}(t_0 \wedge Y \midd X)}{G_{n,k}(t_0 \wedge Y \midd X)} - \frac{\alpha_{n,k,t_0}(0 \midd X) - \alpha_{\infty,t_0}(0 \midd X)}{G_{n,k}(0 \midd X)} \\
			&\hspace{3cm}+ \int_0^{t \wedge Y}\left(\alpha_{n,k,t_0}(u \midd X) - \alpha_{\infty, t_0}(u \midd X)\right)\frac{G_{n,k}(du \midd X)}{G_{n,k}(u \midd X)^2}\bigg\}^2\Bigg]\\
			&= E_0\Bigg[\sup_{t_0 \in (0, \tau]}\bigg\{\frac{\alpha_{n,k,t_0}(t_0 \wedge Y \midd X) - \alpha_{\infty,t_0}(t_0 \wedge Y \midd X)}{G_{n,k}(t_0 \wedge Y \midd X)} - S_{n,k}(t_0 \midd X) + S_{\infty}(t_0 \midd X) \\
			&\hspace{3cm}+ \int_0^{t \wedge Y}\left(\alpha_{n,k,t_0}(u \midd X) - \alpha_{\infty, t_0}(u \midd X)\right)\frac{G_{n,k}(du \midd X)}{G_{n,k}(u \midd X)^2}\bigg\}^2\Bigg]\\
			&\leq E_0\Bigg[\sup_{t_0 \in (0, \tau]}\bigg\{\eta^2\abs{\frac{S_{n,k}(t_0 \midd X)}{S_{n,k}(t_0 \wedge Y \midd X)} -\frac{S_\infty(t_0 \midd X)}{S_\infty(t_0 \wedge Y \midd X)}} + \abs{\frac{S_{n,k}(t_0 \midd X)}{S_{n,k}(0 \midd X)} -  \frac{S_{\infty}(t_0 \midd X)}{S_{\infty}(0 \midd X)}} \\
			&\hspace{3cm}+ \eta^2\abs{ \int_0^{t \wedge Y}\left(\alpha_{n,k,t_0}(u \midd X) - \alpha_{\infty, t_0}(u \midd X)\right)G_{n,k}(du \midd X)}\bigg\}^2\Bigg]\\
			&\leq E_0\Bigg[\sup_{t_0 \in (0, \tau]}\bigg\{\eta^2\sup_{u \in [0,t_0]}\abs{\frac{S_{n,k}(t_0 \midd X)}{S_{n,k}(u\midd X)} -\frac{S_\infty(t_0 \midd X)}{S_\infty(u \midd X)}} + \sup_{u \in [0,t_0]}\abs{\frac{S_{n,k}(t_0 \midd X)}{S_{n,k}(u \midd X)} -  \frac{S_{\infty}(t_0 \midd X)}{S_{\infty}(u \midd X)}} \\
			&\hspace{3cm}+ \eta^2\sup_{u \in [0,t_0]}\abs{\frac{S_{n,k}(t_0 \midd X)}{S_{n,k}(u\midd X)} -\frac{S_\infty(t_0 \midd X)}{S_\infty(u \midd X)}}\bigg\}^2\Bigg]\\
			&\leq (1 + 2\eta^2)^2E_0\left[\sup_{u \in (0, \tau], v \in [0,u]}\abs{\frac{S_\infty(u \midd X)}{S_\infty(v \midd X)} - \frac{S_{n,k}(u \midd X)}{S_{n,k}(v \midd X)}} \right]^2.
		\end{align*}
		The claim therefore holds with $J_\eta = 4+\eta+2\eta^2$. 
	\end{proof}
	
	In the following, for any given function class $\mathcal{H}$, positive real number $\epsilon$, and norm $\norm{\cdot}$, we let the covering number $N(\epsilon, \mathcal{H}, \norm{\cdot})$ denote the number of $\norm{\cdot}$--balls of radius no larger than $\epsilon$ needed to cover $\mathcal{H}$.
	
	\begin{lemma}\label{lemma:Donsker class}
		Let $S$, $G$ and $f$ be fixed with $G(\tau \midd x) > 1/\eta$ for all $x \in \mathcal{X}$ and some $\eta > 0$. Define the class of influence functions $\mathcal{H}_{S,G,f} := \left\{z \mapsto \varphi_{S,G, f,y_0,t_0}(z): y_0 \in \mathcal{Y}, t_0 \in (0, \tau]\right\}$. Then, we have that
		\begin{align*}
			\sup_\Pi \int_0^1\left[1 + N(4\epsilon(1+\eta)^2, \mathcal{H}_{S,G,f}, L_2(\Pi))\right]^{\frac{1}{2}}d\epsilon
		\end{align*}
		is bounded above by a constant not depending on $S$, $G$ or $f$, where the supremum is taken over all probability measures $\Pi$ on the sample space of the observed data. 
	\end{lemma}
	\begin{proof}[*lemma:Donsker class] 
		We begin by noting that for any $h \in \mathcal{H}_{S,G,f}$, we can write $h_{y_0,t_0}(z) = \I(f(x) \leq y_0)\pi_{S,G}(z, t_0)$. Next, we define function classes $\mathcal{H}_{S,G,\tau} := \left\{z \mapsto \pi_{S,G}(z, t_0): t_0 \in (0, \tau]\right\}$ and $\mathcal{H}_{\mathcal{Y}} := \left\{x \mapsto \I(f(x) \leq y_0): y_0 \in \mathcal{Y}\right\}$. Let $\Pi$ denote a generic distribution on the sample space of the observed data. 
		
		For fixed $S$ and $G$, Lemma 5 of \citet{Westling2023} implies that, for any $\epsilon \in (0, 1]$, 
		\begin{align*}
			\sup_\Pi N(\epsilon\norm{H_\tau}_{\Pi,2}, \mathcal{H}_{S,G,\tau}, L^2(\Pi)) < 32\epsilon^{-10}\ ,
		\end{align*}
		where $H_\tau := 2(1 + \eta) $ is a natural envelope for $\mathcal{H}_{S,G,\tau}$. Furthermore, as long as $\mathcal{Y}$ is a bounded subset of $\mathbb{R}$, Lemma 4 of \citet{Westling2023} implies that $\sup_\Pi N(\epsilon, \mathcal{H}_{\mathcal{Y}}, L^2(\Pi)) < 1/\epsilon^2$. The natural envelope for $\mathcal{H}_{\mathcal{Y}}$ is $H_{\mathcal{Y}} := 1$. Therefore, $H_{S,G,f}:= 2(1 + \eta)$ is an envelope of $\mathcal{H}_{S,G,f}$. Next, we note that 
		\begin{align*}
			\mathcal{H}_{S,G,f} \subseteq \left\{h_1h_2: h_1 \in \mathcal{H}_{S,G,\tau}, h_2 \in \mathcal{H}_{\mathcal{Y}} \right\}.
		\end{align*}
		Given that $P_0\left(H_{\tau}^2H_{\mathcal{Y}}^2\right) = 4(1 + \eta)^2 < \infty$ and both function classes are uniformly bounded, Lemma 5.1 of \citet{VdV2006} implies that 
		\begin{align*}
			\sup_\Pi N(2\epsilon\left(1 + \eta\right), \mathcal{H}_{S,G,f}, L_2(\Pi)) \leq 32\epsilon^{-12}\ .
		\end{align*}
		This implies that there exists a constant $\tilde{\kappa} > 0$ not depending on $S$, $G$, $f$ or $\epsilon$ such that
		\begin{align*}
			\log \sup_\Pi N(2\epsilon\left(1 + \eta\right), \mathcal{H}_{S,G,f}, L_2(\Pi)) \leq \tilde{\kappa}\log\left(\frac{1}{\epsilon}\right).
		\end{align*}
		Then, we can write
		\begin{align*}
			\sup_\Pi \int_0^1&\left\{1 + \log N(2\epsilon(1+\eta), \mathcal{H}_{S,G,f}, L_2(\Pi))\right\}^{\frac{1}{2}}d\epsilon \\
			&\leq  \int_0^1\sup_\Pi \left\{1 + \log N(2\epsilon(1+\eta), \mathcal{H}_{S,G,f}, L_2(\Pi))\right\}^{\frac{1}{2}}d\epsilon\\
			&\leq \int_0^1 \left\{1 +  \log\sup_\Pi N(2\epsilon(1+\eta), \mathcal{H}_{S,G,f}, L_2(\Pi))\right\}^{\frac{1}{2}}d\epsilon\leq \int_0^1 \left\{1 + \tilde{\kappa}\log \left(\frac{1}{\epsilon}\right)\right\}^{\frac{1}{2}}d\epsilon\ .
		\end{align*}
		Then, because $1 + \tilde{\kappa}\log (\frac{1}{\epsilon}) > 1$ for $\epsilon \in (0,1)$, we have that
		\begin{align*}
			\int_0^1 \left\{1 + \tilde{\kappa}\log \left(\frac{1}{\epsilon}\right)\right\}^{\frac{1}{2}}d\epsilon &\leq  \int_0^1 \left\{1 + \tilde{\kappa}\log \left(\frac{1}{\epsilon}\right)\right\}d\epsilon = 1 - \tilde{\kappa}\int_0^1\log \epsilon d\epsilon = 1 + \tilde{\kappa}\ .
		\end{align*}
		Therefore the uniform entropy integral $\sup_\Pi \int_0^1\left\{1 + \log N(2\epsilon(1+\eta), \mathcal{H}_{S,G,f}, L_2(\Pi))\right\}^{\frac{1}{2}}d\epsilon $
		is bounded above by a constant not depending on $S$, $G$ or $f$. 
	\end{proof}
	
	\begin{lemma}\label{lemma:control of empirical process term}
		If Conditions \ref{condition:bounded away from zero} and \ref{condition:nuisance limits} hold, then for each $k$ and for any $f \in \mathcal{F}$, we have that
		\begin{align*}
			\sup_{y_0 \in \mathcal{Y}, t \in (0, \tau]}\abs{n_k^{1/2}(\mathbb{P}_{n,k} - P_0)\left(\varphi_{n,k,f,y_0,t_0} - \varphi_{\infty,f,y_0,t_0}\right)} = o_P(1)\ .
		\end{align*}
	\end{lemma}
	\begin{proof}[*lemma:control of empirical process term] 
		We begin by defining $\pi_{\infty} := \pi_{S_\infty, G_\infty}$. We note that $\varphi_{n,k,f,y_0,t_0}(z) = \I(f(x) \leq y_0)\pi_{n,k}(z,t_0)$ and $\varphi_{\infty,f, y_0,t_0}(z) = \I(f(x) \leq y_0)\pi_{\infty}(z,t_0)$. We use $\mathbb{G}_{n,k}$ to denote the empirical process $n_k^{1/2}\left(\mathbb{P}_{n,k} - P_0\right)$.  
		
		We use the tower property to write
		\begin{align*}
			&E_0\left[\sup_{y_0 \in \mathcal{Y}, t_0 \in (0, \tau]}\abs{\mathbb{G}_{n,k}\left(\varphi_{n,k,f,y_0,t_0} - \varphi_{\infty,f,y_0,t_0}\right)}\right] \\
			&\hspace{.8cm}=E_0\left[E_0\left[\sup_{y_0 \in \mathcal{Y}, t_0 \in (0, \tau]}\abs{\mathbb{G}_{n,k}\left(\varphi_{n,k,f,y_0,t_0} - \varphi_{\infty,f,y_0,t_0}\right)}\,\middle|\, \cup_{j\neq k}\mathcal{D}_j\right] \right] \\
			&\hspace{.8cm}= E_0\left[E_0\left[\sup_{h \in \mathcal{H}_{n,k,f}}\abs{\mathbb{G}_{n,k} h}\,\middle|\, \cup_{j\neq k}\mathcal{D}_j\right]\right],
		\end{align*}
		where we define $\mathcal{H}_{n,k,f} := \left\{z \mapsto \varphi_{n,k,f,y_0,t_0}(z) - \varphi_{\infty,f,y_0,t_0}(z): y_0 \in \mathcal{Y}, t_0 \in (0, \tau]\right\}$. Conditioning on the training data $\cup_{j \neq k}\mathcal{D}_j$, the nuisance estimators $S_{n,k}$ and $G_{n,k}$ are fixed functions. We note that $\bar{H}_{n,k,f}: z \mapsto  \sup_{y_0 \in \mathcal{Y}, t_0 \in (0, \tau]}\abs{\varphi_{n,k,f,y_0,t_0}(z) - \varphi_{\infty,f,y_0,t_0}(z)}$ is a natural envelope function for $\mathcal{H}_{n,k,f}$. Then, by the same argument as used in Lemma \ref{lemma:Donsker class}, we have that the uniform entropy integral
		\begin{align*}
			J(1, \mathcal{H}_{n,k,f}) := \sup_\Pi \int_0^1\left[1 + \log N(\epsilon\norm{\bar{H}_{n,k,f}}_{\Pi,2}, \mathcal{H}_{n,k,f}, L_2(\Pi))\right]^{\frac{1}{2}}d\epsilon
		\end{align*}
		is bounded above by a constant  depending on neither $n$ nor $k$. By Theorem 2.14.1 of \citet{vandervaart1996}, there exists a constant $\kappa^*$  depending on neither $n$ nor $k$ such that
		\begin{align}
			&E_0\left[E_0\left[\sup_{h \in \mathcal{H}_{n,k,f}}\abs{\mathbb{G}_{n,k} h}\, \middle|\, \cup_{j\neq k}\mathcal{D}_j\right]\right] \leq \kappa^*E_0\left[E_0\left[\bar{H}_{n,k,f}(Z)^2\,\middle|\, \cup_{j \neq k}\mathcal{D}_j\right]^{\frac{1}{2}}\right] \nonumber \\ 
			&\hspace{1cm}= \kappa^*E_0\left[E_0\left[\sup_{y_0 \in \mathcal{Y}, t_0 \in (0, \tau]}\left\{\varphi_{n,k,f,y_0,t_0}(Z) - \varphi_{\infty,f, y_0,t_0}(Z)\right\}^2\,\middle|\, \cup_{j \neq k}\mathcal{D}_j\right]^{\frac{1}{2}}\right] \nonumber \\ 
			&\hspace{1cm}\leq \kappa^*\left\{E_0\left[E_0\left[\sup_{y_0 \in \mathcal{Y}, t_0 \in (0, \tau]}\left\{\varphi_{n,k,f,y_0,t_0}(Z) - \varphi_{\infty, f,y_0,t_0}(Z)\right\}^2\,\middle|\, \cup_{j \neq k}\mathcal{D}_j\right]\right]\right\}^{\frac{1}{2}}, \label{eq:iterated upper bound}
		\end{align}
		where the second inequality follows from Jensen's inequality. 
		Next, let $U_{n,k}$ denote
		\begin{align*}
			&J_\eta^2\left[E_0\left[\sup_{u \in (0, \tau], v \in [0,u]}\abs{\frac{S_\infty(u \midd X)}{S_\infty(v \midd X)} - \frac{S_{n,k}(u \midd X)}{S_{n,k}(v \midd X)}} \right]^2\vee E_0\left[\sup_{u \in (0, \tau]}\abs{\frac{1}{G_{n,k}(u \midd X)} - \frac{1}{G_\infty(u \midd X)}}\right]^2\right]
		\end{align*}
		with the expectation taken over the random data unit $X$ and not over the random functions $S_{n,k}$ and $G_{n,k}$. By Lemma \ref{lemma:sup L2 discrepancy of nuisances}, we have that
		\begin{align*}
			E_0\left[\sup_{y_0 \in \mathcal{Y}, t_0 \in (0, \tau]}\left\{\varphi_{n,k,f,y_0,t_0}(Z) - \varphi_{\infty, f,y_0,t_0}(Z)\right\}^2\,\middle|\, \cup_{j \neq k}\mathcal{D}_j\right] \leq U_{n,k}
		\end{align*}
		with probability tending to one. Combining this with (\ref{eq:iterated upper bound}), we observe that
		\begin{align*}
			E_0\left[\sup_{y_0 \in \mathcal{Y}, t_0 \in (0, \tau]}\abs{\mathbb{G}_{n,k}\left(\varphi_{n,k,f,y_0,t_0} - \varphi_{\infty,f,y_0,t_0}\right)}\right] \leq \kappa^*\left\{E_0\left[U_{n,k}\right]\right\}^{\frac{1}{2}}
		\end{align*}
		with probability tending to one. We note that $U_{n,k}$ is a uniformly bounded sequence of random variables converging in probability to 0 under Condition \ref{condition:nuisance limits}. This implies that $\kappa^*\left\{E_0\left[U_{n,k}\right]\right\}^{\frac{1}{2}} \to 0$.  
		Finally, applying Markov's inequality, for any $\epsilon > 0 $, we have that
		\begin{align*}
			P_0&\left(\sup_{y_0 \in \mathcal{Y}, t_0 \in (0, \tau]}\abs{\mathbb{G}_{n,k}\left(\varphi_{n,k,f,y_0,t_0} - \varphi_{\infty,f,y_0,t_0}\right)}> \epsilon\right) \\
			&\hspace{2cm}\leq \frac{1}{\epsilon}E_0\left[\sup_{y_0 \in \mathcal{Y}, t_0 \in (0, \tau]}\abs{\mathbb{G}_{n,k}\left(\varphi_{n,k,f,y_0,t_0} - \varphi_{\infty,f,y_0,t_0}\right)}\right] \to 0\ .
		\end{align*}
	\end{proof}
	
	\begin{lemma}\label{lemma:uniform consistency of joint}
		If Conditions \ref{condition:bounded away from zero}--\ref{condition:double robust} hold, then for each $k$ it holds that
		\begin{align*}
			\sup_{y_0 \in \mathcal{Y}, t_0 \in (0, \tau]}\abs{H^*_{n,k,f_{n,k}}(y_0,t_0) - H_{0,f_{n,k}}(y_0,t_0)} \stackrel{\textnormal{p}}{\longrightarrow} 0 \ .
		\end{align*}
	\end{lemma}
	\begin{proof}[*lemma:uniform consistency of joint] 
		For each $k$ and any $f \in \mathcal{F}$, we define the plug-in estimator  $H_{n,k,f}(y_0,t_0) := \frac{1}{n_k}\sum_{i \in \mathcal{I}_k}\I(f(X_i) \leq y_0)F_{n,k}(t_0 \midd X_i)$. We then have that
		\begin{align*}
			&H_{n,k,f_{n,k}}(x_0,t_0) - H_{0,f_{n,k}}(y_0,t_0) \\
			&\hspace{1.8cm}= (\mathbb{P}_{n,k}-P_0)\bar{\varphi}_{n,k,f_{n,k},y_0,t_0}  + R_{f_{n,k},y_0,t_0}(P_{n,k}, P_0) - \mathbb{P}_{n,k}\bar{\varphi}_{n,k,f_{n,k},y_0,t_0}\ ,
		\end{align*}
		where, for generic $f \in \mathcal{F}$, we define 
		$R_{f,y_0,t_0}(P_{n,k}, P_0):= H_{n,k,f}(y_0,t_0) - H_{0,f}(y_0,t_0) + P_0\bar{\varphi}_{n,k,f,y_0,t_0}$.
		We can therefore write
		\begin{align*}
			H_{n,k,f_{n,k}}^*(y_0,t_0) - H_{0,f_{n,k}}(y_0,t_0) = (\mathbb{P}_{n,k}- P_0)\bar{\varphi}_{n,k,f_{n,k},y_0,t_0} + R_{f_{n,k},y_0,t_0}(P_{n,k}, P_0)\ . 
		\end{align*}
		Applying the triangle inequality yields that
		\begin{align*}
			&\sup_{y_0 \in \mathcal{Y}, t_0 \in (0, \tau]}\abs{H^*_{n,k,f_{n,k}}(y_0,t_0) - H_{0,f_{n,k}}(y_0,t_0)} \\
			&\hspace{1.8cm}\leq \sup_{y_0 \in \mathcal{Y}, t_0 \in (0, \tau]}\abs{(\mathbb{P}_{n,k} - P_0)\bar{\varphi}_{n,k,f_{n,k},y_0,t_0}} + \sup_{y_0 \in \mathcal{Y}, t_0 \in (0, \tau]}\abs{R_{f_{n,k},y_0,t_0}(P_{n,k}, P_0)}\ .
		\end{align*}
		Similarly as in the proof of Lemma \ref{lemma:control of empirical process term}, we use the tower property to write
		\begin{align*}
			&E_0\left[\sup_{y_0 \in \mathcal{Y}, t_0 \in (0, \tau]}\abs{(\mathbb{P}_{n,k} - P_0)\bar{\varphi}_{n,k,f_{n,k},y_0,t_0} }\right] = E_0\left[E_0\left[\sup_{h \in \mathcal{H}_{n,k,f_{n,k}}}\abs{(\mathbb{P}_{n,k} - P_0) h}\,\middle|\, \cup_{j\neq k}\mathcal{D}_j\right]\right],
		\end{align*}
		where we define $\mathcal{H}_{n,k,f_{n,k}} := \left\{z \mapsto \bar{\varphi}_{n,k,f_{n,k},y_0,t_0}(z): y_0 \in \mathcal{Y}, t_0 \in (0, \tau]\right\}$. Conditioning on the training data $\cup_{j \neq k}\mathcal{D}_j$, the nuisance estimators $S_{n,k}$, $G_{n,k}$ and $f_{n,k}$ are fixed functions. Then, by the same argument as used in Lemma \ref{lemma:Donsker class}, we have that the uniform entropy integral
		\begin{align*}
			J(1, \mathcal{H}_{n,k,f_{n,k}}) := \sup_\Pi \int_0^1\left[1 + \log N(2\epsilon(1 + \eta), \mathcal{H}_{n,k,f_{n,k}}, L_2(\Pi))\right]^{\frac{1}{2}}d\epsilon
		\end{align*}
		is bounded above by a constant depending on neither $n$ nor $k$. Theorem 2.14.1 of \citet{vandervaart1996} implies that there exists a constant $\kappa^*$  depending on neither $n$ nor $k$ such that
		\begin{align}
			&E_0\left[E_0\left[\sup_{h \in \mathcal{H}_{n,k,f_{n,k}}}\abs{(\mathbb{P}_{n,k} - P_0)h}\, \middle|\, \cup_{j\neq k}\mathcal{D}_j\right]\right] \nonumber \leq 2\kappa^*(1 + \eta)n^{-1/2}\ . \nonumber 
		\end{align}
		Therefore, applying Markov's inequality, for any $\epsilon > 0$, we have that
		\begin{align*}
			P_0&\left(\sup_{y_0 \in \mathcal{Y}, t_0 \in (0, \tau]}\abs{(\mathbb{P}_{n,k} - P_0)\bar{\varphi}_{n,k,f_{n,k}, y_0, t_0}} > \epsilon\right)\\
			&\hspace{2cm} \leq \frac{1}{\epsilon}E_0\left[\sup_{y_0 \in \mathcal{Y}, t_0 \in (0, \tau]}\abs{(\mathbb{P}_{n,k} - P_0)\bar{\varphi}_{n,k,f_{n,k},y_0,t_0}}\right] \to 0
		\end{align*} 
		Finally, because $P_0\bar{\varphi}_{n,k,f_{n,k},y_0,t_0} = P_0\varphi_{n,k,f_{n,k},y_0,t_0} - H_{n,k,f_{n,k}}(y_0,t_0)$, we note that 
		\begin{align}
			R_{f_{n,k},y_0,t_0}(P_{n,k}, P_0) &= H_{n,k,f_{n,k}}(y_0,t_0) - H_{0,f_{n,k}}(y_0,t_0) + P_0\varphi_{n,k,f_{n,k},y_0,t_0} - H_{n,k,f_{n,k}}(y_0,t_0) \nonumber \\
			&= P_0\varphi_{n,k,f_{n,k},y_0,t_0} - H_{0,f_{n,k}}(y_0,t_0)\ .\label{eq:joint cdf linearization remainder}
		\end{align}
		Under Condition \ref{condition:double robust}, we have that $H_{0,f_{n,k}}(y_0,t_0) = P_0\varphi_{\infty,f_{n,k},y_0,t_0}$. Therefore, Lemma \ref{lemma:DR} implies that $R_{f_{n,k},y_0,t_0}(P_{n,k}, P_0)  = P_0\left(\varphi_{n,k,f_{n,k},y_0,t_0} - \varphi_{\infty,f_{n,k},y_0,t_0}\right)$. Under Conditions \ref{condition:bounded away from zero} and \ref{condition:nuisance limits}, Lemma \ref{lemma:sup L2 discrepancy of nuisances} implies that 
		\begin{align*}
			\sup_{y_0 \in \mathcal{Y}, t_0 \in (0, \tau]}\abs{R_{f_{n,k},y_0,t_0}(P_{n,k}, P_0)} = o_P(1)\ . 
		\end{align*}
	\end{proof}
	
	\begin{lemma}\label{lemma:convergence of joint process}
		If Conditions \ref{condition:bounded away from zero} and \ref{condition:nuisance limits} hold with $S_\infty = S_0$ and $G_\infty = G_0$, and Condition \ref{condition:second order remainder} holds, then for each $k$ and each $f \in \mathcal{F}$, it holds that
		\begin{align*}
			\sup_{y_0 \in \mathcal{Y}, t_0 \in (0, \tau]}\abs{H^*_{n,k,f}(y_0,t_0) - H_{0,f}(y_0,t_0) - \mathbb{P}_{n,k}\bar{\varphi}_{0,f,y_0,t_0}} = o_P(n_k^{-1/2})\ .
		\end{align*}
		In particular, $\{n_k^{1/2}(H^*_{n,k,f}(y_0,t_0) - H_{0,f}(y_0,t_0)): y_0 \in \mathcal{Y}, t_0 \in (0, \tau]\}$ converges weakly relative to the supremum norm to a tight mean-zero Gaussian process with covariance $((y_1,t_1),(y_2,t_2)) \mapsto P_0(\bar{\varphi}_{0,f,y_1,t_1}\bar{\varphi}_{0,f,y_2,t_2})$. 
	\end{lemma}
	\begin{proof}[*lemma:convergence of joint process] 
		For any $f \in \mathcal{F}$, we can decompose
		\begin{align*}
			&H_{n,k,f}(y_0,t_0) - H_{0,f}(y_0,t_0) \\
			&\hspace{.8cm}= (\mathbb{P}_{n,k}-P_0)\bar{\varphi}_{\infty,f,y_0,t_0} + C_{n,k,f,y_0,t_0}(P_{n,k}, P_\infty) + R_{f,y_0,t_0}(P_{n,k}, P_0) - \mathbb{P}_{n,k}\bar{\varphi}_{n,k,f,y_0,t_0}\ ,
		\end{align*}
		where we have defined $C_{n,k,f,y_0,t_0}(P_{n,k}, P_\infty):= n_k^{-\frac{1}{2}}\mathbb{G}_{n,k}(\bar{\varphi}_{n,k,f,y_0,t_0} - \bar{\varphi}_{\infty,f,y_0,t_0})$. We can therefore write
		\begin{align*}
			H_{n,k,f}^*(y_0,t_0) - H_{0,f}(y_0,t_0) = (\mathbb{P}_{n,k}- P_0)\bar{\varphi}_{\infty,f,y_0,t_0} + C_{n,k,f,y_0,t_0}(P_{n,k}, P_0) + R_{f,y_0,t_0}(P_{n,k}, P_0)\ . 
		\end{align*}
		Because $F_\infty = F_0$, $S_\infty = S_0$, $L_\infty = L_0$, and $G_\infty = G_0$, we have that $\bar{\varphi}_{\infty,f} = \bar{\varphi}_{0,f}$. First, this implies that $P_0\bar{\varphi}_{\infty,f,y_0,t_0} = 0$. Second, we can use the triangle inequality to write that
		\begin{align*}
			&\sup_{y_0 \in \mathcal{Y}, t_0 \in (0, \tau]}\abs{H^*_{n,k,f}(y_0,t_0) - H_{0,f}(y_0,t_0) - \mathbb{P}_{n,k}\bar{\varphi}_{0,f,y_0,t_0}} \\
			&\hspace{0.8cm}\leq \sup_{y_0 \in \mathcal{Y}, t_0 \in (0, \tau]}\abs{C_{n,k,f,y_0,t_0}(P_{n,k}, P_0)} + \sup_{y_0 \in \mathcal{Y}, t_0 \in (0, \tau]}\abs{R_{f,y_0,t_0}(P_{n,k}, P_0)}\ .
		\end{align*}
		Under Conditions \ref{condition:bounded away from zero} and \ref{condition:nuisance limits}, Lemma \ref{lemma:control of empirical process term} implies that the leading term on the right-hand side above is $o_P(n_k^{-1/2})$. For the second term, we can use (\ref{eq:joint cdf linearization remainder}) and Lemma \ref{lemma:duhamel} to write that
		\begin{align*}
			&R_{f,y_0,t_0}(P_{n,k}, P_0) = P_0\varphi_{n,k,f,y_0,t_0} - H_{0,f}(y_0,t_0) \\
			&\hspace{0.8cm}= E_0\left[\I(f(X) \leq y_0)S_{n,k}(t_0\midd X)\int_0^{t_0}\frac{S_0(u^-\midd X)}{S_{n,k}(u \midd X)}\left\{\frac{G_0(u \midd X)}{G_{n,k}(u \midd X)} - 1\right\}(L_{n,k} - L_0)(du \midd X)\right].
		\end{align*}
		Therefore, it follows that
		\begin{align*}
			&\sup_{y_0 \in \mathcal{Y}, t_0 \in (0, \tau]}\abs{R_{f,y_0,t_0}(P_{n,k}, P_0)} \\
			&\leq   \sup_{y_0 \in \mathcal{Y}, t_0 \in (0, \tau]}\abs{E_0\left[\I(f(X) \leq y_0)S_{n,k}(t_0\midd X)\int_0^{t_0}\frac{S_0(u^-\midd X)}{S_{n,k}(u \midd X)}\left\{\frac{G_0(u \midd X)}{G_{n,k}(u \midd X)} - 1\right\}(L_{n,k} - L_0)(du \midd X)\right]}\\
			&\leq E_0\left[\sup_{y_0 \in \mathcal{Y}, t_0 \in (0, \tau]}\abs{\I(f(X) \leq y_0)S_{n,k}(t_0\midd X)\int_0^{t_0}\frac{S_0(u^-\midd X)}{S_{n,k}(u \midd X)}\left\{\frac{G_0(u \midd X)}{G_{n,k}(u \midd X)} - 1\right\}(L_{n,k} - L_0)(du \midd X)}\right]\\
			&\leq  E_0\left[\sup_{t_0 \in (0, \tau]}\abs{S_{n,k}(t_0\midd X)\int_0^{t_0}\frac{S_0(u^-\midd X)}{S_{n,k}(u \midd X)}\left\{\frac{G_0(u \midd X)}{G_{n,k}(u \midd X)} - 1\right\}(L_{n,k} - L_0)(du \midd X)}\right].
		\end{align*}
		This upper bound is $o_P(n_k^{-1/2})$ under Condition \ref{condition:second order remainder}. 
		
		We conclude that $\sup_{y_0 \in \mathcal{Y}, t_0 \in (0, \tau]}\lvert H^*_{n,k,f}(y_0,t_0) - H_{0,f}(y_0,t_0) - \mathbb{P}_{n,k}\bar{\varphi}_{0,f,y_0,t_0}\rvert = o_P(n_k^{-1/2})$. Lemma \ref{lemma:Donsker class} implies that the class of influence functions $\left\{\bar{\varphi}_{0,f,y_0,t_0}: y_0 \in \mathcal{Y}, t_0 \in (0, \tau]\right\}$ is a uniformly bounded $P_0$--Donsker class, and so, $\{n_k^{1/2}(H^*_{n,k,f}(y_0,t_0) - H_{0,f}(y_0,t_0)): y_0 \in \mathcal{Y}, t_0 \in (0, \tau]\}$ converges weakly relative to the supremum norm to a tight mean-zero Gaussian process with covariance function $(y_1,t_1),(y_2,t_2) \mapsto P_0(\bar{\varphi}_{0,f,y_1,t_1}\bar{\varphi}_{0,f,y_2,t_2})$, as claimed. 
	\end{proof}
	
	\subsection{Asymptotic analysis of VIM estimator}
	For each $k$, we define the remainder terms 
	\begin{align*}
		&r_{\omega,n,k} := V_1(f_0, H_{n,k}^*) - V_1(f_0, H_0) - m \int\omega_{0,f_0,1}\left(f_0(x_1), t_1\right)\left(H_{n,k}^* - H_0\right)(dx_1,dt_1)\ ;\\
		&r_{\theta,n,k} := V_2(H_{n,k}^*) - V_2(H_0)- m \int\theta_{0,1}\left( t_1\right)\left(H_{n,k}^* - H_0\right)(dx_1,dt_1)\ .
	\end{align*}
	
	\begin{lemma}\label{lemma:V stat remainder}
		If Conditions \ref{condition:bounded away from zero}--\ref{condition:nuisance limits} hold with $S_\infty = S_0$ and $G_\infty = G_0$, and Conditions \ref{condition:variation norm of kernel} and \ref{condition:second order remainder} also hold, then $r_{\omega,n,k}$ and $r_{\theta,n,k}$ are $o_P(n_k^{-1/2})$.
	\end{lemma}
	\begin{proof}[*lemma:V stat remainder] 
		Throughout this proof, we repeatedly make use of several facts. First, \citet{Gill1993} established that there exists a constant $0<\kappa<\infty$ such that, for any c\`{a}dl\`{a}g functions $q_1$ and $q_2$, it holds that
		\begin{align}
			\abs{\int q_1(u) q_2(du)} \leq \kappa \norm{q_2}_\infty \norm{q_1}_v^*\ . \label{eq:gill variation norm fact}
		\end{align}
		Next, for a function $\beta: \mathbb{R}^{d + p} \to \mathbb{R}$ and measure $\alpha$, we define 
		\begin{align*}
			\beta_{-p}(u_1,\dots,u_d) := \int \cdots \int \beta(u_1,\dots, u_d, u_{d + 1}, \dots, u_{d + p})\alpha(du_{d + 1}, \dots, du_{d + p})\ .
		\end{align*} 
		We claim that
		\begin{align}
			\norm{\beta_{-p}}_{v}^* \leq \sup_{u_{d + 1}, \dots, u_{d + p}} \norm{\beta(\cdot, u_{d + 1}, \dots, u_{d + p})}_v^*\norm{\alpha}_v\ , \label{eq:partial variation norm fact}
		\end{align}
		where $\norm{\cdot}_v$ denotes the variation norm. By definition, we have that $\norm{\beta_{-p}}_v^* = \sup_s \sup_{r_s}\norm{\beta_{-p}}_{v,r_s}$, where with some abuse of notation we write $\norm{\beta_{-p}}_{v,r_s} := \sup_{u_{-r_s}}\int \abs{\beta_{-p}(du_{r_s}, u_{-r_{s}})}$ and $r_s\subseteq\{1,2,\ldots,d\}$. Hence, we can write that
		\begin{align*}
			\norm{\beta_{-p}}_{v}^* &=  \sup_s \sup_{r_s}\sup_{u_{-r_s}}\int \abs{\int \cdots \int\beta(du_{r_s}, u_{-r_{s}},u_{d + 1}, \dots, u_{d + p})\alpha(du_{d + 1}, \dots, du_{d + p})} \\
			&\leq \left[\sup_s \sup_{r_s}\sup_{u_{-r_s}}\sup_{u_{d + 1}, \dots, u_{d + p}}\int \abs{\beta(du_{r_s}, u_{-r_{s}},u_{d + 1}, \dots, u_{d + p})} \right]\int \cdots \int \abs{\alpha(du_{d + 1}, \dots, du_{d + p})}\\
			&= \sup_{u_{d + 1}, \dots, u_{d + p}} \norm{\beta(\cdot, u_{d + 1}, \dots, u_{d + p})}_v^*\norm{\alpha}_v\ .
		\end{align*}
		Suppose that $g: \R^{d + p }\to \R$ is a given function that can be written as the product $g = g_1g_2$ for some functions $g_1,g_2:\mathbb{R}^{d+p}\rightarrow\mathbb{R}$. Then, elementary calculations suffice to show that 
		\begin{align}
			\norm{g}_v^* \leq  2\norm{g_1}_v^*\norm{g_2}_v^*\ .\label{eq:product variation norm fact}
		\end{align}
		
		With these facts in hand, we begin by showing that for any $f \in \mathcal{F}$, $n_k^{-\frac{1}{2}}\norm{ \mathbb{G}_{n,k}\bar{\varphi}_{0,f,\cdot,\cdot}}^*_v = O_P(1)$. Here and after, for any fixed $y_0$, we take the variation norm of $t_0 \mapsto n_k^{-\frac{1}{2}}\norm{\mathbb{G}_{n,k}\bar{\varphi}_{0,f,y_0,t_0}}_{v}^*$ to be over the interval $(0, \tau]$, since the integrands $\omega$ and $\theta$ are not functions of $t_0$ when $t_0 > \tau$. Because we can write $(\mathbb{P}_{n,k} - P_0)\bar{\varphi}_{0,f,y_0,t_0} = \frac{1}{n_k}\sum_{i \in \mathcal{I}_k}\bar{\varphi}_{0,f,y_0,t_0}(Z_i)$, the triangle inequality yields that $n_k^{-\frac{1}{2}}\norm{\mathbb{G}_{n,k}\bar{\varphi}_{0,f,\cdot,\cdot}}^*_v \leq \frac{1}{n_k}\sum_{i \in \mathcal{I}_k}\norm{\bar{\varphi}_{0,f,\cdot,\cdot}(Z_i)}_v^*$. We then note that $\varphi_{0,f,y_0,t_0}(z) = \I(f(x) \leq y_0)\pi_{0}(z,t_0)$. For fixed $x$, the uniform sectional variation norm of $y_0 \mapsto \I(f(x) \leq y_0)$ is 1, so in light of (\ref{eq:product variation norm fact}), for any fixed $z$, we have that $\norm{\bar{\varphi}_{0,f,\cdot,\cdot}(z)}_v^* \leq 2\norm{\pi_{0}(z,\cdot)}_v^* = 2\norm{\pi_0(z,\cdot)}_v$, where we have replaced the uniform sectional variation norm with the variation norm since $t_0 \mapsto \pi_{0}(z,t_0)$ is both univariate and bounded on $(0,\tau]$. Now, for any fixed $Z = z$, we have
		\begin{align*}
			&\norm{\pi_{0}(z,\cdot)}_v = \int_0^\tau \abs{\left\{F_0(dt \midd x) - \chi_0(z,dt) \right\}}\\
			&\leq \int_0^\tau \abs{F_0(dt \midd x)} + \int_0^\tau \abs{\left\{\frac{\delta \I_{[0,y]}(t)}{S_0(y \midd x)G_0(y \midd x)} - \int_0^{t \wedge y}\frac{L_0(du \midd x)}{S_0(u \midd x)G_0(u \midd x)}  \right\}S_0(dt \midd x)}\\
			&\hspace{1.5cm}+ \frac{\delta\I(y \leq \tau)}{G_0(y \midd x)} + \int_0^{y\wedge \tau} \abs{\frac{L_0(dt \midd x)}{G_0(t \midd x)}}\\
			&\leq 1 + \frac{\delta}{G_0(\tau \midd x)} + \sup_{t \in (0, \tau]}\abs{\int_0^{t \wedge y}\frac{L_0(du \midd x)}{G_0(u \midd x)}}\int \abs{S_0(dt \midd x)} + \frac{\delta}{G_0(\tau \midd x)} + \int_0^{y \wedge \tau}\abs{\frac{L_0(dt \midd x)}{G_0(t \midd x)}}\\
			&\leq  1 + \frac{2\delta}{G_0(\tau \midd x)}  + 2\int_0^\tau\abs{\frac{L_0(dt \midd x)}{G_0(t \midd x)}}\ .
		\end{align*}
		This implies, in view of Condition \ref{condition:bounded away from zero}, that the random variable $\|\pi_0(Z,\cdot)\|^*_{v}$ has finite mean and variance under $P_0$, and so, in particular, we have that
		\begin{align*}
			\frac{1}{n_k}\sum_{i \in \mathcal{I}_k}\norm{\pi_{0}(Z_i, \cdot)}^*_v = O_P(1)\ .
		\end{align*}
		Therefore, we find that $n_k^{-\frac{1}{2}}\norm{\mathbb{G}_{n,k}\bar{\varphi}_{0,f,\cdot,\cdot}}_v^* = O_P(1)$, as claimed.
		
		Next, for generic $f \in \mathcal{F}$, we define pointwise $r_{n,H_{0,f}}(y_0,t_0) = H_{n,k,f}^*(y_0,t_0) - H_{0,f}(x_0,t_0) - n_k^{-\frac{1}{2}}\mathbb{G}_{n,k}\bar{\varphi}_{0,f,x_0,t_0}$.  By the triangle inequality, we have that
		\begin{align*}
			\norm{r_{n,H_{0,f}}}_v^* \leq \|n_k^{-\frac{1}{2}}\mathbb{G}_{n,k}\bar{\varphi}_{0,f,\cdot,\cdot}\|_v^* + \norm{H^*_{n,k,f} - H_{0,f}}_v^*\leq n_k^{-\frac{1}{2}}\norm{\mathbb{G}_{n,k}\bar{\varphi}_{0,f,\cdot,\cdot}}_v^* + 2  = O_P(1)\ .
		\end{align*}
		
		We now analyze the remainder terms $r_{\theta,n,k}$ and $r_{\omega,n,k}$. For $m = 1$, we have $r_{\omega,n,k} = 0$ for $m = 1$, and for $m \geq 2$, we have that
		\begin{align*}
			r_{\omega,n,k} &= \sum_{l=2}^{m}A_{l,m}\int \cdots \int \omega_{0,f_0,l}((f_0(x_1),t_1),\dots, (f_0(x_l),t_l))\prod_{j=2}^{l}(H^*_{n,k} - H_0)(dx_j, dt_j)\\
			&= \sum_{l=2}^{m}A_{l,m}\int \cdots \int \omega_{0,f_0,l}((y_1,t_1),\dots, (y_l,t_l))\prod_{j=2}^{l}(H^*_{n,k,f_0} - H_{0,f_0})(dy_j, dt_j)\ ,
		\end{align*}
		where the coefficients are defined recursively via the relationship $A_{l + 1,m} := \sum_{i=l}^{m-1}A_{l,i}$ for $l = 1,2,\dots, m - 1$ with initialization $A_{1,m} := m$. For any $l \geq 2$, we can write the corresponding summand in the above sum as $r_{n,\omega,k,1} + r_{n,\omega,k,2} + r_{n,\omega,k,3}$, where we define
		\begin{align*}
			r_{n,\omega,k,1}&:=A_{l,m} n_k^{-l/2}\int\cdots \int \omega_{0,f_0,l}\left((y_1, t_1),\dots,(y_l, t_l)\right)\prod_{j=1}^{l}\mathbb{G}_{n,k}\bar{\varphi}_{0,f_0,dy_j, dt_j} \\
			r_{n,\omega,k,2}&:=A_{l,m}\sum_{j=1}^{l-1}{l \choose j}\int\cdots \int \Bigg\{ \omega_{0,f_0,l}\left((y_1, t_1),\dots,(y_l, t_l)\right) \\
			&\hspace{5cm} \times \prod_{s=1}^{l-j}n_k^{-\frac{1}{2}}\mathbb{G}_{n,k}\bar{\varphi}_{0, f_0,dy_s, dt_s} \prod_{s=l-j+1}^{l}r_{n,H_{0,f_0}}(dy_s, dt_s)\Bigg\}\\
			r_{n,\omega,k,3}&:=A_{l,m}\int\cdots \int\omega_{0,f_0,l}\left((y_1, t_1),\dots,(y_l, t_l)\right)\prod_{j=1}^{l}r_{n,H_{0,f_0}}(dy_j, dt_j)\ .
		\end{align*}
		We bound each of these terms separately. The leading term $r_{n,\omega,k,1}$ is a degenerate $V$--statistic and hence is $o_P(n^{-1/2})$ \citep{Serfling1980}. For $r_{n,\omega,k,2}$, there exists a constant $\kappa_1$ such that, for all $1 \leq j \leq l-1$, we have that 
		\begin{align}
			&A_{l,m}{l \choose j}\int\cdots \int \left|\omega_{0,f_0,l}\left((y_1,t_1),\dots, (y_l,t_l)\right)\prod_{s=1}^{l-j}n_k^{-\frac{1}{2}}\mathbb{G}_{n,k}\bar{\varphi}_{0,f_0,dy_s,dt_s}\prod_{s=l-j+1}^{l}r_{n,H_{0,f_0}}(dy_s, dt_s)\right| \nonumber  \\
			&\leq \kappa_1 \|r_{n,H_{0,f_0}}\|_\infty\norm{q_{n,j}}_v^*
			\label{eq:cross norm remainder theta}
		\end{align}
		with
		\begin{align*}
			q_{n,j}: (y_l,t_l) \mapsto  \int \omega_{0,f_0,l}\left((y_1,t_1),\dots,(y_l,t_l)\right)\prod_{s=1}^{l-j}n_k^{-\frac{1}{2}}\mathbb{G}_{n,k}\bar{\varphi}_{0,f_0,dy_s,dt_s}\prod_{s=l-j+1}^{l-1}r_{n,H_{0,f_0}}(dy_s, dt_s)
		\end{align*}
		and where the inequality is due to (\ref{eq:gill variation norm fact}). Now, using (\ref{eq:partial variation norm fact}) and (\ref{eq:product variation norm fact}), we note that $\norm{q_{n,j}}_v^*$ is bounded above by
		\begin{align*}
			2^{l-1}n_k^{-\frac{l-j}{2}}\sup_{(y_1,t_1), \dots, (y_{l-1},t_{l-1})}\norm{ \omega_{0,f_0,l}((y_1,t_1), \dots , (y_{l-1},t_{l-1}), \cdot)}^*_v\norm{\mathbb{G}_{n,k}\bar{\varphi}_{0,f_0, \cdot, \cdot}}_v^{l-j}\|r_{n, H_{0,f_0}}\|_v^{j-1}\ .
		\end{align*}
		Under Condition \ref{condition:variation norm of kernel}, this upper bound is $O_P(1)$, and so, the upper bound in (\ref{eq:cross norm remainder theta}) is $O_P(\|r_{n,H_{0,f_0}}\|_\infty)$. 
		For $r_{n,\omega,k,3}$, we apply an analogous argument, replacing $n_k^{-\frac{1}{2}}\mathbb{G}_{n,k}\pi_{0,y_0,t_0}$ with $r_{n,H_{0,f_0}}(y_0,t_0)$, to conclude that $r_{n,\omega,k,3} = O_P(\|r_{n,H_{0,f_0}}\|_\infty)$.
		Under Conditions \ref{condition:bounded away from zero}, \ref{condition:nuisance limits} (with $S_\infty = S_0$ and $G_\infty = G_0$) and \ref{condition:second order remainder}, Lemma \ref{lemma:convergence of joint process} yields that $\|r_{n,H_{0,f_0}}\|_\infty = o_P(n^{-1/2})$, and so, both $r_{n,\omega,k,2}$ and $r_{n,\omega,k,3}$ are $o_P(n^{-1/2})$.
		
		We now study  $r_{\theta,n,k}$.  Again, we have $r_{\theta,n,k} = 0$ for $m = 1$, and for $m \geq 2$, we have that
		\begin{align}
			r_{\theta,n,k} = \sum_{l=2}^{m}A_{l,m}\int \cdots \int \theta_{0,l}(t_1,\dots, t_{l})\prod_{j=1}^{l}(H^*_{n,k} - H_0)(dx_j, dt_j)\ ,\label{eq:r theta remainder}
		\end{align}
		For any $l \geq 2$, we can write the corresponding summand in (\ref{eq:r theta remainder}) as $r_{n,\theta,k,1} + r_{n,\theta,k,2} + r_{n,\theta,k,3}$, where we define
		\begin{align*}
			r_{n,\theta,k,1}&:=A_{l,m}n_k^{-l/2}\int\cdots \int \theta_{0,l}\left(t_1,\dots,t_l\right)\prod_{j=1}^{l}\mathbb{G}_{n,k}\bar{\varphi}_{0, f_0,dx_j, dt_j} \\
			r_{n,\theta,k,2}&:=A_{l,m}\sum_{j=1}^{l-1}{l \choose j}\int\cdots \int \theta_{0,l}\left(t_1,\dots, t_l\right)\prod_{s=1}^{l-j}n_k^{-\frac{1}{2}}\mathbb{G}_{n,k}\bar{\varphi}_{0,f_0,dx_s,dt_s}\prod_{s=l-j+1}^{l}r_{n,H_{0,f_0}}(dx_s, dt_s)\\
			r_{n,\theta,k,3}&:=A_{l,m}\int\cdots \int \theta_{0,l}\left(t_1,\dots,t_l\right)\prod_{j=1}^{l}r_{n,H_{0,f_0}}(dx_j, dt_j)\ .
		\end{align*}
		Each of these three terms can be analyzed identically as their analogs above. We find that $r_{n,\theta,k,1} = o_P(n^{-1/2})$ because it is a degenerate $V$--statistic.
		Under Conditions \ref{condition:bounded away from zero}, \ref{condition:nuisance limits} (with $S_\infty = S_0$ and $G_\infty = G_0$) and \ref{condition:second order remainder}, Lemma \ref{lemma:convergence of joint process} implies that both $r_{n,\theta,k,2}$ and $r_{n,\theta,k,3}$ are $o_P(n^{-1/2})$.
	\end{proof}
	
	\begin{proof}[*thm:consistency]
		We begin by expanding $v_{n,1}^*$ around $v_{0,1}$ as
		\begin{align*}
			v_{n,1}^* &- v_{0,1} = \frac{1}{K}\sum_{k=1}^{K}\left\{V_1(f_{n,k}, H_{0}) - V_1(f_0, H_0)\right\} + \frac{1}{K}\sum_{k=1}^{K}\left\{V_1(f_{n,k}, H^*_{n,k}) - V_1(f_{n,k}, H_0)\right\}.
		\end{align*}
		Under Condition \ref{condition:continuity of V}, we have that $\abs{V_1(f_{n,k}, H_0) - V_1(f_0, H_0)} \leq J_1 \norm{f_{n,k} - f_0}_{\mathcal{F}} = o_P(1)$ for each $k$, and hence, it follows that $\frac{1}{K}\sum_{k=1}^{K}\left\{V_1(f_{n,k}, H_{0}) - V_1(f_0, H_0)\right\} = o_P(1)$. Next, we note that
		\begin{align*}
			&\abs{V_1(f_{n,k}, H^*_{n,k}) - V_1(f_{n,k}, H_0)} \\
			&\hspace{.5cm}= \abs{\int \dots \int \omega\left((f_{n,k}(x_1), t_1), \dots, (f_{n,k}(x_m),t_m)\right)\left\{\prod_{j=1}^{m}H^*_{n,k}(dx_j, dt_j) -\prod_{j=1}^{m}H_{0}(dx_j, dt_j) \right\}}\ .
		\end{align*}
		By a telescoping sum argument, we have that \begin{align*}
			&\prod_{j=1}^{m}H^*_{n,k}(dx_j, dt_j) -\prod_{j=1}^{m}H_{0}(dx_j, dt_j)\\
			&\hspace{.5cm}=\sum_{l=1}^{m}\left[\prod_{j=1}^{l-1}H^*_{n,k}(dx_j,dt_j)\left\{H_{n,k}^*(dx_l,dt_l) - H_0(dx_l, dt_l)\right\}\prod_{j=l+1}^{m}H_0(dx_j, dt_j) \right],
		\end{align*} which allows us to write that $|V_1(f_{n,k}, H^*_{n,k}) - V_1(f_{n,k}, H_0)|$ is equal to
		\begin{align*}
			&\left|\sum_{l=1}^m\int\cdots \int\bigg\{ \omega\left((f_{n,k}(x_1), t_1), \dots, (f_{n,k}(x_m),t_m)\right)\right.\\
			&\hspace{3cm} \times \left.\prod_{j=1}^{l-1}H^*_{n,k}(dx_j,dt_j)\{H_{n,k}^*(dx_l,dt_l) - H_0(dx_l, dt_l)\}\prod_{j=l+1}^{m}H_0(dx_j, dt_j)\bigg\}\right| \\
			&= \left|\sum_{l=1}^m\int q_{n,l}(f_{n,k}(x_l),t_l)\{H_{n,k}^*(dx_l,dt_l) - H_0(dx_l, dt_l)\}\right|\\
			&=  \left|\sum_{l=1}^m\int q_{n,l}(y_l,t_l)\{H_{n,k,f_{n,k}}^*(dy_l,dt_l) - H_{0,f_{n,k}}(dy_l, dt_l)\}\right|\\
			&\leq \sum_{l=1}^m\left|\int q_{n,l}(y_l,t_l)\{H_{n,k,f_{n,k}}^*(dy_l,dt_l) - H_{0,f_{n,k}}(dy_l, dt_l)\}\right|\, \leq \bar{\kappa}\|H_{n,k,f_{n,k}}^*-H_{0,f_{n,k}}\|_{\infty}\sum_{l=1}^{m}\|q_{n,l}\|_v^*
		\end{align*}
		in view of the triangle inequality and (\ref{eq:gill variation norm fact}), for some constant $0<\bar{\kappa}<\infty$ and with $q_{n,l}$ defined as $(y_l,t_l)\mapsto \int\cdots\int \omega_{0,f_{n,k},l}((y_1,t_1),\ldots,(y_l,t_l))\prod_{j=1}^{l-1}H_{n,k,f_{n,k}}^*(dy_j,dt_j)$. 
		Then, using (\ref{eq:partial variation norm fact}) and (\ref{eq:product variation norm fact}), we can write that, for each $l \in \{1,2,\dots,m\}$,
		\begin{align*}
			\norm{q_{n,l}}_v^* &\leq 2^{l-1}\sup_{(y_1,t_1),\dots,(y_{l-1}, t_{l-1})}\norm{\omega_{0,f_{n,k}l}((y_1,t_1), \dots, (y_{l-1}, t_{l-1}), \cdot)}_v^* \norm{H_{n,k,f_{n,k}}^*}_v^{l-1}\\
			&= 2^{l-1}\sup_{(y_1,t_1),\dots,(y_{l-1}, t_{l-1})}\norm{\omega_{0,f_{n,k},l}((y_1,t_1), \dots, (y_{l-1}, t_{l-1}), \cdot)}_v^* \ .
		\end{align*}
		Each of these terms is bounded in light of Condition \ref{condition:variation norm of kernel}, and so, Lemma \ref{lemma:uniform consistency of joint} implies that $\lVert H_{n,k,f_{n,k}}^* - H_{0,f_{n,k}}\rVert_\infty \sum_{l=1}^{m}\norm{q_{n,l}}_v^* = o_P(1)$. 
		Therefore, $V_1(f_{n,k}, H^*_{n,k}) - V_1(f_{n,k}, H_0) = o_P(1)$ for each $k$, and we conclude that 
		\begin{align*}
			\frac{1}{K}\sum_{k=1}^{k}\left\{V_1(f_{n,k}, H^*_{n,k}) - V_1(f_{n,k}, H_0)\right\} = o_P(1)\ . 
		\end{align*}For $v_{n,2}^*$, we have that $v_{n,2}^* - v_{0,2} = \frac{1}{K}\sum_{k=1}^{K}\{V_2(H^*_{n,k}) - V_2(H_0)\}$. An identical argument as above yields that 
		$ V_2(H^*_{n,k}) - V_2(H_0) = o_P(1)$ for each $k$. Therefore, it holds that $\frac{1}{K}\sum_{k=1}^{K}\{V_2(H^*_{n,k}) - V_2(H_0)\} = o_P(1)$, and we conclude that $v_{n,2}^* \stackrel{\text{p}}{\longrightarrow}v_{0,2}$. By Slutsky's lemma, we finally have that $v_{n,1}^*/v_{n,2}^* \stackrel{\text{p}}{\longrightarrow} v_0$.
	\end{proof}
	
	\begin{proof}[*thm:asymptotic linearity]
		We begin by expanding $v_{n,1}^*$ around $v_{0,1}$ as
		\begin{align*}
			v_{n,1}^* &-v_{0,1} = \frac{1}{K}\sum_{k=1}^{K}\left\{V_1(f_0, H^*_{n,k}) - V_1(f_0, H_0)\right\} + \frac{1}{K}\sum_{k=1}^{K}\left\{V_1(f_{n,k}, H_0) -V_1(f_0, H_0)\right\} + r_n\ ,
		\end{align*}
		where $r_{n} := \frac{1}{K}\sum_{k=1}^{K}[\{V_1(f_{n,k}, H^*_{n,k}) - V_1(f_{n,k}, H_0)\} - \{V_1(f_0, H^*_{n,k}) - V_1(f_0, H_0)\}]$. Under Conditions \ref{condition:optimality} and \ref{condition:rate of convergence for f}, for each $k$, we have that $\abs{V_1(f_{n,k}, H_0) - V_1(f_0, H_0)} \leq J_2 \norm{f_{n,k} - f_0}^2_{\mathcal{F}} = o_P(n^{-1/2})$. Next, applying Lemma \ref{lemma:V stat remainder}, we have that 
		\begin{align*}
			&V_1(f_0, H^*_{n,k}) - V_1(f_0, H_0)\\
			&\hspace{.1cm}= m \iint \omega_{0,f_0,1}\left(f_0(x_1), t_1\right)(H^*_{n,k} - H_0)(dx_1,dt_1) + o_P(n_k^{-1/2})\\
			&\hspace{.1cm}= m \iint \omega_{0,f_0,1}\left(y_1, t_1\right)(H^*_{n,k,f_0} - H_{0,f_0})(dy_1,dt_1) + o_P(n_k^{-1/2})\\
			&\hspace{.1cm}=m\left\{\frac{1}{n_k}\sum_{i \in \mathcal{I}_k}\int \omega_{0,f_0,1}\left(f_0(X_i), t_1\right)\pi_{0}(Z_i,dt_1) +  \iint \omega_{0,f_0,1}\left(y_1, t_1\right)r_{n, H_{0,f_0}}(dx_1,dt_1) - V_1(f_0, H_0)\right\}\\
			&\hspace{1cm}+ o_P(n_k^{-1/2})\\
			&\hspace{.1cm}= \frac{m}{n_k}\sum_{i \in \mathcal{I}_k}\left\{\int \omega_{0,f_0,1}\left(f_0(X_i),t_1\right)\left(F_0(dt_1 \midd X_i)-\chi_{0}(Z_i,dt_1)\right)- V_1(f_0,  H_0)\right\} + o_P(n_k^{-1/2})\\
			&\hspace{.5cm}= \mathbb{P}_{n,k}\phi_{\omega,0} + o_P(n_k^{-1/2})\ ,
		\end{align*}
		where the fourth equality follows by noting that, under Condition \ref{condition:variation norm of kernel}, there exists a constant $\kappa$ such that
		\begin{align*}
			\abs{m  \iint \omega_{0,f_0,1}\left(y_1, t_1\right)r_{n, H_{0,f_0}}(dx_1,dt_1)} \leq \kappa\norm{\omega_{0,f_0,1}}_v^*\|r_{n,H_{0,f_0}}\|_\infty = o_P(n_k^{-1/2})\ .
		\end{align*}
		Then, using the triangle inequality, we can write that
		\begin{align*}
			\abs{\left\{\frac{1}{K}\sum_{k=1}^{K}V_1(f_0, H^*_{n,k}) - V_1(f_0, H_0)\right\} -\mathbb{P}_n \phi_{\omega,0}} &\leq \max_k\abs{\frac{n}{Kn_k}-1}\cdot\mathbb{P}_n\phi_{\omega,0} + \frac{1}{K}\sum_{k=1}^{K}o_P(n_k^{-1/2}) \\
			&= O_P(n^{-1}) + o_P(n^{-1/2}) = o_P(n^{-1/2})\ .
		\end{align*}
		To analyze $r_n$, we note that $P_0^m \Gamma(f_0, \pi_0) = V_1(f_0, H_0)$. We also use the fact that for a function $b:\mathbb{R}^{p+1}\rightarrow\mathbb{R}$, $\int b(x,t)H^*_{n,k}(dx,dt) = \frac{1}{n_k}\sum_{i \in \mathcal{I}_k}\int b(X_i,t)\pi_{n,k}(Z_i, dt)$. Then, for each $k$, we have that
		\begin{align*}
			&\left\{V_1(f_{n,k}, H^*_{n,k}) - V_1(f_{n,k}, H_0)\right\} - \left\{V_1(f_0, H^*_{n,k}) - V_1(f_0, H_0)\right\}\\
			&= \mathbb{P}_{n,k}^m\Gamma(f_{n,k}, \pi_{n,k}) - P_0^m\Gamma(f_{n,k}, \pi_0) - \mathbb{P}_{n,k}^m\Gamma(f_0, \pi_{n,k}) + P_0^m\Gamma(f_0, \pi_0)\\
			&= (\mathbb{P}_{n,k}^m - P_0^m)\left\{\Gamma(f_{n,k}, \pi_{n,k}) - \Gamma(f_0, \pi_0)\right\} - (\mathbb{P}_{n,k}^m - P_0^m)\left\{\Gamma(f_0, \pi_{n,k}) - \Gamma(f_0, \pi_0)\right\}\\
			&\hspace{1cm}+ P_0^m\left\{\Gamma(f_{n,k}, \pi_{n,k}) - \Gamma(f_{n,k}, \pi_0) - \Gamma(f_0, \pi_{n,k}) + \Gamma(f_0, \pi_0)\right\}\\
			&= (\mathbb{P}_{n,k}^m - P_0^m)\left\{\Gamma(f_{n,k}, \pi_{n,k}) - \Gamma(f_0, \pi_0)\right\} - (\mathbb{P}_{n,k}^m - P_0^m)\left\{\Gamma(f_0, \pi_{n,k}) - \Gamma(f_0, \pi_0)\right\}\\
			&\hspace{1cm}+ E_0\Bigg[\int\cdots\int\Big\{ \omega\left((f_{n,k}(X_1), t_1),\dots,(f_{n,k}(X_m),t_m)\right)\\
			&\hspace{2.5cm}-\omega\left((f_{0}(X_1), t_1),\dots,(f_{0}(X_m),t_m)\right)\Big\}\Big\{\prod_{j=1}^{m}\pi_{n,k}(Z_j, dt_j)-\prod_{j=1}^{m}\pi_{0}(Z_j, dt_j)\Big\}\Bigg]\\
			&= (\mathbb{P}_{n,k}^m - P_0^m)h_{1,n,k} - (\mathbb{P}_{n,k}^m - P_0^m)h_{2,n,k}\\
			&\hspace{1cm}+ E_0\Bigg[\int\cdots\int\Big\{ \omega\left((f_{n,k}(X_1), t_1),\dots,(f_{n,k}(X_m),t_m)\right)\\
			&\hspace{2.5cm}-\omega\left((f_{0}(X_1), t_1),\dots,(f_{0}(X_m),t_m)\right)\Big\}\Big\{\prod_{j=1}^{m}\pi_{n,k}(Z_j, dt_j)-\prod_{j=1}^{m}\pi_{0}(Z_j, dt_j)\Big\}\Bigg]\ .
		\end{align*}
		The expectations in the display above are taken with respect to $(Z_1, Z_2, \dots, Z_m)$. The trailing term is $o_P(n^{-1/2})$ under Condition \ref{condition:second order remainder}. 
		
		We let $\mathbb{P}_{n,k,*}^mh_{1,n,k}$ denote the U-statistic analog of $P_{n,k}^mh_{1,n,k}$, that is, 
		\begin{align*}
			\mathbb{P}_{n,k,*}h_{1,n,k} &:= {n \choose m}^{-1}\sum_{\underline{i}_m \in I_{m,n_k}}h_{1,n,k}(Z_{i_1}, \dots Z_{i_m})\ ,
		\end{align*}
		where $I_{m,n_k} := \left\{\underline{i}_m \subseteq \mathcal{I}_k: i_1 < i_2 < \cdots < i_m\right\}$. Results from \citet{Serfling1980} allow us to establish that if $var_{0}[\bar{h}_{1,n,k}(Z) \midd \cup_{j \neq k}\mathcal{D}_j] > 0$, then it holds that 
		\begin{align*}
			var_{0}\left[n_k^{1/2}\mathbb{P}^m_{n,k,*}h_{1,n,k}\,\middle|\, \cup_{j \neq k}\mathcal{D}_j\right] &= m^2 var_0[\bar{h}_{1,n,k}(Z)\midd \cup_{j \neq k}\mathcal{D}_j] + O(n_k^{-1})\ .
		\end{align*}
		Therefore, for any $\epsilon > 0$, we can apply Chebyshev's inequality to get that
		\begin{align*}
			0 \leq P_0\left(\abs{n_k^{1/2}(\mathbb{P}^m_{n,k,*} -P_0^m)h_{1,n,k}} > \epsilon \,\middle|\, \cup_{j \neq k}\mathcal{D}_j\right) &\leq \frac{1}{\epsilon^2}var_{0}\left[n_k^{1/2}\mathbb{P}^m_{n,k,*}h_{1,n,k}\, \middle|\, \cup_{j \neq k}\mathcal{D}_j\right] \\
			&= \frac{m^2var_{0}[\bar{h}_{1,n,k}(Z)\midd \cup_{j \neq k}\mathcal{D}_j] + O(n_k^{-1})}{\epsilon^2}\\
			&\leq \frac{m^2E_0[\bar{h}^2_{1,n,k}(Z)\midd \cup_{j \neq k}\mathcal{D}_j] + O(n_k^{-1})}{\epsilon^2}\ .
		\end{align*}
		Then, in light of Condition \ref{condition:weak consistency}, we have that $P_0(\lvert n_k^{1/2}(\mathbb{P}^m_{n,k,*} -P_0^m)h_{1,n,k}\rvert > \epsilon \midd \cup_{j \neq k}\mathcal{D}_j) = o_P(1)$. This holds for any realization of $\cup_{j \neq k}\mathcal{D}_{j}$, and since probabilities are uniformly bounded, we find that
		\begin{align*}
			P_0\left(\abs{n^{1/2}_k(\mathbb{P}^m_{n,k,*}-  P_0^m)h_{1,n,k}} > \epsilon\right) = E_0\left[P_0\left(\abs{n_k^{1/2}(\mathbb{P}^m_{n,k,*} -P_0^m)h_{1,n,k}} > \epsilon \,\middle|\, \cup_{j \neq k}\mathcal{D}_j\right)\right] = o(1)\ .
		\end{align*}
		When $h_{1,n,k}$ is fixed, we have that $ (\mathbb{P}_{n,k}^m - \mathbb{P}_{n,k,*}^m)h_{1,n,k} = o_P(n_k^{-1/2})$ \citep{Serfling1980}. This allows us to write
		\begin{align*}
			P_0\left(\abs{n^{1/2}_k(\mathbb{P}^m_{n,k}-  \mathbb{P}_{n,k,*}^m)h_{1,n,k}} > \epsilon\right) =  E_0\left[P_0\left(\abs{n_k^{1/2}(\mathbb{P}^m_{n,k} -\mathbb{P}_{n,k,*}^m)h_{1,n,k}} > \epsilon \,\middle|\, \cup_{j \neq k}\mathcal{D}_j\right)\right]  = o(1)\ .
		\end{align*}
		Finally, we can write that
		\begin{align*}
			(\mathbb{P}_{n,k}^m - P_0^m)h_{1,n,k}  &= (\mathbb{P}_{n,k,*}^m - P_0^m)h_{1,n,k} + (\mathbb{P}_{n,k}^m - \mathbb{P}_{n,k,*}^m)h_{1,n,k}= o_P(n_k^{-1/2}) + o_P(n_k^{-1/2}) =  o_P(n_k^{-1/2})\ .
		\end{align*}
		Since $n/n_k \stackrel{\text{p}}{\longrightarrow} K $, we have that $(\mathbb{P}^m_{n,k}-  P_0^m)h_{1,n,k} = o_P(n^{-1/2})$. An identical argument holds for $(\mathbb{P}^m_{n,k}-  P^m_0)h_{2,n,k}$. Thus, for each $k$, it holds that
		\begin{align*}
			\left\{V_1(f_{n,k}, H^*_{n,k}) - V_1(f_{n,k}, H_0)\right\} - \left\{V_1(f_0, H^*_{n,k}) - V_1(f_0, H_0)\right\} = o_P(n^{-1/2})\ .
		\end{align*}
		We conclude that $r_n = o_P(n^{-1/2})$, and thus, that $v_{n,1}^* -V(f_0, H_0) = \mathbb{P}_n\phi_{\omega,0} + o_P(n^{-1/2})$. 
		
		Now, we observe that $v_{n,2}^* - v_{0,2} = \frac{1}{K}\sum_{k=1}^{K}\{V_2(H^*_{n,k}) - V_2(H_0)\}$.  By Lemma \ref{lemma:V stat remainder}, we have that 
		\begin{align*}
			&V_2(H^*_{n,k}) - V_2(H_0) = m \iint  \theta_{0,1}\left( t_1\right)(H^*_{n,k} - H_0)(dx_1,dt_1) + o_P(n_k^{-1/2})\\
			&= m\left\{\frac{1}{n_k}\sum_{i \in \mathcal{I}_k}\int  \theta_{0,1}\left(t_1\right)\left(F_0(dt_1 \midd X_i) - \chi_{0}(Z_i,dt_1)\right)- v_{0,2}\right\}+ o_P(n_k^{-1/2})= \mathbb{P}_{n,k}\phi_{\theta,0} + o_P(n_k^{-1/2})\ .
		\end{align*}
		Using the same argument as above, we have that $\lvert \frac{1}{K}\sum_{k=1}^{K}V_2(H^*_{n,k}) - V_2(H_0) - \mathbb{P}_n\phi_{\theta,0}\rvert = o_P(n^{-1/2})$, and hence, that $v_{n,2}^* - v_{0,2} = \mathbb{P}_n\phi_{\theta,0} + o_P(n^{-1/2})$. 
		
		Finally, an application of the delta method yields that $v_{n}^* -v_0 = (\mathbb{P}_n-P_0)\phi_0+ o_P(n^{-1/2})$. Therefore, $v_{n}^*$ is an asymptotically linear estimator of $v_{0}$ with influence function equal to $\phi_0$. Under Condition \ref{condition:efficient influence function}, Theorem 1 holds and $\phi_0$ is the efficient influence function of $P \mapsto V(f_P, P)$ at $P_0$ relative to $\mathcal{M}_{\text{obs}}$, so that $v_n^*$ is nonparametric efficient.
	\end{proof}
	
	\section{Additional technical details}\label{sec:additional technical details}
	
	\subsection{Derivation of oracle prediction functions for given example measures}
	\setcounter{example}{0}
	
	\begin{example}[\sc{AUC}]
		\upshape
		For a binary outcome $D$, \citet{Agarwal2014} considered the AUC\ $V(f,P) := P(f(X_1) > f(X_2) \midd D_1 = 1, D_2 = 0)$ and showed that the oracle prediction function is $x \mapsto E_P(D \midd X = x)$. Adapting this to the binary outcome $\I(T > \tau)$, we have the oracle $\mathbbmsl{f}_0: x \mapsto \mathbbmsl{P}_0(T > \tau \midd X = x) = 1 - \mathbbmsl{F}_0(\tau \midd x)$.
	\end{example}

	\begin{example}[\sc{Brier score}]
		\upshape
		The Brier score at time $\tau$ is simply the negative MSE\ for the binary outcome $\I(T > \tau)$, which is maximized by the conditional mean $x \mapsto \mathbbmsl{P}_0(T > \tau \midd X=x) = 1 -\mathbbmsl{F}_0(\tau \midd x)$. 
	\end{example}

	\begin{example}[\sc{Survival time MSE}]
		\upshape
		As for the Brier score, the negative MSE\ is maximized by the conditional mean, which in this case is $x \mapsto E_{\mathbbmsl{P}_0}(T \wedge \tau \midd X = x)$. This can be written in terms of $\mathbbmsl{F}_0$ as $\mathbbmsl{f}_0: x \mapsto \int_0^\tau \left\{1-\mathbbmsl{F}_0(t \midd x)\right\} dt$.
	\end{example}
	
	\begin{example}[\sc{C-index}]
		\upshape
		If there exists a function $f$ such that $f(x_1) > f(x_2)$ implies 
		\begin{align}
			\mathbbmsl{P}_0(T_1 < T_2, T_1 \leq \tau\midd X_1 = x_1, X_2 = x_2) \geq \mathbbmsl{P}_0(T_2 < T_1, T_2 \leq \tau\midd X_1 = x_1, X_2 = x_2) \label{eq:optimal risk ordering}
		\end{align}
		for all $(x_1,x_2)$, then $f$ is an oracle prediction function. To see this, we define 
		\begin{align*}
			c_{0,f}: x_1,x_2 \mapsto E_{\mathbbmsl{P}_0}&[\I(f(x_1) > f(x_2))\I(T_1 < T_2, T_1 \leq \tau) \\
			&+ \I(f(x_2) \geq f(x_1))\I(T_2 < T_1, T_2 \leq \tau)\midd X_1 = x_1, X_2 = x_2]\ .
		\end{align*}
		By the tower property, we have that $E_{\mathbbmsl{P}_0}[c_{0,f}(X_1,X_2)] = \mathbbmsl{V}(f, \mathbbmsl{P}_0)$. Next, we note that for all $(x_1,x_2)$ it holds that
		\begin{align*}
			&c_{0,f}(x_1,x_2) \\
			&=E_{\mathbbmsl{P}_0}[\I(f(x_1) > f(x_2))\I(T_1 < T_2, T_1 \leq \tau) \\
			&\hspace{2cm}+ (1 - \I(f(x_1) > f(x_2)))\I(T_2 < T_1, T_2 \leq \tau)\midd X_1 = x_1, X_2 = x_2 ]\\
			&= E_{\mathbbmsl{P}_0}[\I(T_2 < T_1, T_2 \leq \tau) \\
			&\hspace{2cm}+ \I(f(x_1) > f(x_2))\left\{\I(T_1 < T_2, T_1 \leq \tau) - \I(T_2 < T_1, T_2 \leq \tau)\right\}\midd X_1 = x_2, X_2 = x_2]\ .
		\end{align*}
		Let $f_0$ be a prediction function satisfying (\ref{eq:optimal risk ordering}). For any other $f^*\in \mathcal{F}$ and any $(x_1,x_2)$, we can write
		\begin{align*}
			& c_{0,f_0}(x_1,x_2) - c_{0, f^*}(x_1,x_2)\\
			&= \left\{\I(f_0(x_1) > f_0(x_2)) - \I(f^*(x_1) > f^*(x_2))\right\}\{\mathbbmsl{P}_0(T_1 < T_2, T_1 \leq \tau \midd X_1 = x_1, X_2 = x_2) \\
			&\hspace{8cm}- \mathbbmsl{P}_0(T_2 < T_1, T_2 \leq \tau \midd X_1 = x_1, X_2 = x_2)\}\geq 0 \ .
		\end{align*}
		Here, we have used that fact that when $\I(f_0(x_1) > f_0(x_2)) = 1$ --- and therefore the leading term above is nonnegative --- the trailing term is positive, and when $\I(f_0(x_1) > f_0(x_2)) = 0$ --- and therefore the leading term above is nonpositive --- the trailing term is negative. Applying the tower property yields that
		\begin{align*}
			\mathbbmsl{V}(f_0, \mathbbmsl{P}_0) - \mathbbmsl{V}(f, \mathbbmsl{P}_0) = E_{\mathbbmsl{P}_0}[c_{0,f_0}(X_1,X_2) - c_{0,f^*}(X_1,X_2)] \geq 0 \ .
		\end{align*}
		
	\end{example}
	
	\subsection{Explicit form of efficient influence function for given example measures}
	\setcounter{example}{0}
	
	\begin{example}[\sc{AUC}]
		\upshape
		We have that
		\begin{align*}
			&\omega_{0,1}(x,t) = \frac{1}{2}\int \left\{\I(f_0(x) > f_0(x_2), t \leq \tau, t_2 > \tau) + \I(f_0(x_2) > f_0(x), t_2 \leq \tau, t > \tau) \right\} H_0(dx_2, dt_2)\\
			&\hspace{0.8cm}= \frac{1}{2}\int \left\{\I(f_0(x) > f_0(x_2), t \leq \tau)\left\{1 - F_0(\tau \midd x_2)\right\} + \I(f_0(x_2) > f_0(x), t > \tau)F_0(\tau \midd x_2)\right\}Q_0(dx_2)\\
			&\hspace{0.8cm}= \frac{1}{2}E_0\left[\I(f_0(x) > f_0(X), t \leq \tau)\left\{1 - F_0(\tau \midd X)\right\} + \I(f_0(X) > f_0(x), t > \tau)F_0(\tau \midd X) \right];\\
			&\theta_{0,1}(t) = \frac{1}{2}\int \left[\I( t \leq \tau, t_2 > \tau) + \I(t_2 \leq \tau, t > \tau) \right] H_0(dx_2, dt_2)\\
			&\hspace{1cm}= \frac{1}{2}\int \left\{\I(t \leq \tau)\left\{1 - F_0(\tau \midd x_2)\right\} + \I(t > \tau)F_0(\tau \midd x_2) \right\}Q_0(dx_2)\\
			&\hspace{1cm}= \frac{1}{2}E_0\left[\I(t \leq \tau)\left\{1 - F_0(\tau \midd X)\right\} + \I(t > \tau)F_0(\tau \midd X) \right].
		\end{align*}
		Noting that $\int_0^\tau F_0(dt \midd x) = F_0(\tau \midd x)$, $\int_\tau^\infty F_0(dt \midd x) = 1 - F_0(\tau \midd x)$,  $\int_0^\tau \chi_0(z,dt) = \chi_{0,z}(\tau)$, and $\int_\tau^\infty \chi_0(z,dt) = -\chi_{0,z}(\tau)$, we have that
		\begin{align*}
			\phi_{\omega,0}:z &\mapsto E_0\biggr[\I(f_0(x) > f_0(X))\left\{F_0(\tau \midd x) - \chi_{0}(z,\tau)\right\}\left\{1 - F_0(\tau \midd X)\right\} \\
			&\hspace{2cm}+ \I(f_0(X) > f_0(x))F_0(\tau \midd X)\left\{1 - F_0(\tau \midd x) + \chi_{0}(z,\tau)\right\} \biggr] - 2v_{0,1}\ ;\\
			\phi_{\theta,0}:z &\mapsto E_0\biggr[\left\{F_0(\tau \midd x) - \chi_{0}(z,\tau)\right\}\left\{1 - F_0(\tau \midd X)\right\} \\
			&\hspace{2cm}+ F_0(\tau \midd X)\left\{1 - F_0(\tau \midd x) + \chi_{0}(z,\tau)\right\} \biggr] - 2v_{0,2}\ .
		\end{align*}
	\end{example}
	
	\begin{example}[\sc{Brier score}]
		\upshape
		Because $m = 1$ and $\theta(t)  = 1$, we have that $\omega_{0,1}(x,t) = \omega(f_0(x), t)$ and $\theta_{0,1}(t) = 1$. Therefore, it follows that
		\begin{align*}
			\phi_{\omega,0}:z &\mapsto \int \left\{f_0(x)  - \I(t > \tau)\right\}^2\left\{F_0(dt \midd x) - \chi_0(z,dt)\right\}-v_{0,1}\ .
		\end{align*}
	\end{example}
	
	\begin{example}[\sc{Survival time MSE}]
		\upshape
		As for the Brier score, we have that $m = 1$ and $\theta(t) = 1$, and so, it follows that
		\begin{align*}
			\phi_{\omega,0}:z \mapsto -\int \left\{f_0(x) - (t \wedge \tau )\right\}^2\left\{F_0(dt \midd x) - \chi_0(z,dt)\right\} - v_{0,1}\ .
		\end{align*}
	\end{example}
	
	\begin{example}[\sc{C-index}]
		\upshape
		We have that
		\begin{align*}
			&\omega_{0,1}(x,t) = \frac{1}{2}\int \left\{\I(f_0(x) > f_0(x_2), t \leq t_2, t \leq \tau) + \I(f_0(x_2) > f_0(x), t_2 \leq t, t_2 \leq \tau) \right\}H_0(dx_2, dt_2)\\
			&\hspace{1cm}= \frac{1}{2}E_0\left[\int \left[\I(f_0(x) > f_0(X), t\leq t_2, t \leq \tau) + \I(f_0(X) > f_0(x), t_2 \leq t, t_2 \leq \tau) \right]F_0(dt_2 \midd X) \right];\\
			&\theta_{0,1}(t)  = \frac{1}{2}\int \left[\I(t \leq t_2, t \leq \tau) + \I(t_2 \leq t, t_2 \leq \tau) \right] H_0(dx_2, dt_2)\\
			&\hspace{1cm}= \frac{1}{2}E_0\left[\int \left[\I(t \leq t_2, t \leq \tau) + \I( t_2 \leq t, t_2 \leq \tau) \right]F_0(dt_2 \midd X) \right].
		\end{align*}
		Therefore, it follows that
		\begin{align*}
			\phi_{\omega,0}&:z \mapsto E_0\biggr[\int \Big\{\I(f_0(x) > f_0(X), t \leq t_2, t \leq \tau)+\I(f_0(X) > f_0(x), t_2 \leq t, t_2 \leq \tau) \Big\} \\
			&\hspace{3cm} \times\left\{F_0(dt \midd x) - \chi_0(z,dt)\right\}F_0(dt_2 \midd X) \biggr] - 2v_{0,1}\ ;\\
			\phi_{\theta,0}&:z \mapsto E_0\biggr[\int \Big\{\I(t \leq t_2, t \leq \tau)+ \I( t_2 \leq t_1, t_2 \leq \tau) \Big\} \\
			&\hspace{3cm}\times\left\{F_0(dt \midd x) - \chi_0(z,dt)\right\}F_0(dt_2 \midd X)\biggr] - 2v_{0,2}\ .
		\end{align*}
	\end{example}
	
	\subsection{Verification of Conditions \ref{condition:variation norm of kernel} and \ref{condition:optimality} for given example measures}
	\setcounter{example}{0}
	
	\begin{example}[\sc{AUC}]
		\upshape
		\citet{Williamson2021b} considered AUC for binary outcomes, and showed that Condition \ref{condition:optimality} holds for $J_2 = \frac{2\kappa}{\xi_0(1 - \xi_0)}$ with $\xi_0 = \mathbbmsl{P}_0(T \leq \tau)$ and $\norm{\cdot}_\mathcal{F}$ the supremum norm, under the margin condition $
		P_0\left(\abs{F_0(\tau \midd X_1) - F_0(\tau \midd X_2)} < s\right) \leq \kappa s$
		for some $0 < \kappa < \infty$ and all $s$ small.
		
		For any fixed $t_1$, we can write $
		\norm{\theta(t_1,\cdot)}_v = \int \abs{\theta(t_1,dt_2)}= \I(t_1> \tau)$, as well as $\norm{\theta(t_1, \cdot)}_\infty = 1$,
		and so, we have that $\sup_{t_1}\norm{\theta(t_1,\cdot)}_v^* = 1$. 
		Next, we note that, for any fixed $(y_1,t_1)$, we can write $\omega$ as the product of functions $(y_2,t_2)\mapsto\I(y_1 > y_2)$ and $(y_2,t_2)\mapsto\I(t_2 \leq \tau)\I(t_1 > \tau)$. The uniform sectional variation norm of both functions is 1 for fixed $(y_1,t_1)$, and so, we have that $\sup_{y_1,t_1}\norm{\omega((y_1,t_1),\cdot)}_v^* < \infty$. 
	\end{example}
	
	\begin{example}[\sc{Brier score}]
		\upshape
		We have that $\abs{V(f, P_0) - V(f_0, P_0)}  = E_0\left\{f(X) - f_0(X)\right\}^2$ as long as $x\mapsto F_0(\tau \midd x)$ falls in $\mathcal{F}$. Therefore, Condition \ref{condition:optimality} holds with $J_2 = 1$ and $\norm{\cdot}_{\mathcal{F}}$ taken to be the $L_2(P_0)$ or supremum norm.
		
		We have that $\omega(y,t) = y^2 - 2y\I(t > \tau) + \I(t > \tau)$. By the triangle inequality, we have that $\|\omega\|_v^* \leq \|y \mapsto y^2\|_v^* + 4\|y \mapsto y\|_v^* + 1$. This quantity is finite if, for example, $\mathcal{Y}$ is a bounded subset of $\mathbb{R}$. Because the outcome of interest $\I(T > \tau)$ is binary, we may reasonably take $\mathcal{F}$ to be a class of functions mapping to $\mathcal{Y} = [0,1]$, in which case the variation norm is finite.  
	\end{example}
	
	\begin{example}[\sc{Survival time MSE}]
		\upshape
		We have that $\abs{V(f, P_0) - V(f_0, P_0)}  = E_0\left\{f(X) - f_0(X)\right\}^2$ as long as $x\mapsto E_0(T \wedge \tau \midd X = x)$ falls in $\mathcal{F}$. Therefore, Condition \ref{condition:optimality} holds with $J_2 = 1$ and $\norm{\cdot}_{\mathcal{F}}$ taken to be the $L_2(P_0)$ or supremum norm.
		
		We have that $\omega(y,t) = y^2 - 2y(t \wedge \tau) + (t \wedge \tau)$. We note that $\int_0^\tau \abs{dt} = \tau$, and so, by the triangle inequality, we have that $\|\omega\|_v^* \leq \|y \mapsto y^2\|_v^* + 4\tau\|y \mapsto y\|_v^* + \tau$. As for the Brier score above, this quantity is finite if, for example, $\mathcal{Y}$ is a bounded subset of $\mathbb{R}$. In this example, the outcome of interest is $T \wedge \tau$, so we may reasonably take $\mathcal{F}$ to be a class of functions mapping to $\mathcal{Y} = [0,\tau]$, in which case the variation norm is finite.  
	\end{example}
	
	\begin{example}[\sc{C-index}]
		\upshape
		For any fixed $t_1$, we can write
		$\norm{\theta(t_1,\cdot)}_v = \int \abs{\theta(t_1,dt_2)} = 1$, as well as $\norm{\theta(t_1,\cdot)}_\infty = 1$,
		and so, we have that $\sup_{t_1}\norm{\theta(t_1,\cdot)}_v^* = 1$. Next, for any fixed $(y_1,t_1)$, we can write $\omega$ as the product of functions $(y_2,t_2)\mapsto \I(y_1 > y_2)$ and $(y_2,t_2)\mapsto\I(t_2 \leq t_1, t_2 \leq \tau)$. The variation norm of each of these functions is 1 for fixed $(y_1,t_1)$, and so, as for the AUC, we find that $\sup_{y_1,t_1}\|\omega((y_1,t_1),\cdot)\|_v^* < \infty$. 
		
		As discussed in the main text, Condition \ref{condition:optimality} appears difficult to verify when $f_0$ is not available in closed form. 
	\end{example}
	
	\section{Numerical optimization of the C-index}\label{sec:gradient boosted c}
	Due to the popularity of the C-index as a predictiveness measure in survival analysis, several authors have proposed learning a prediction function by directly optimizing the C-index. Existing approaches to C-index maximization have been based on either the Harrell's C-index statistic \citep{Harrell1982} or the inverse-probability-of-censoring-weight-based C-index estimator of \citet{Uno2011}. Due to the lack of smoothness of the indicator function $(x_1,x_2)\mapsto \I(f(x_1) > f(x_2))$, the C-index is not amenable to gradient-based optimization techniques, and so, existing methods typically use a smoothed, differentiable approximation or lower bound. \citet{Raykar2008} cast maximization of the C-index as a ranking problem and propose a differentiable lower bound on the C-index, which is then optimized using a conjugate gradients algorithm. \citet{Chen2012} propose a smoothed C-index approximation based on the sigmoid function and optimize it via gradient boosted machines. \citet{Mayr2014} use a similar smoothed objective function and likewise use gradient boosting with linear models as base learners. Specifically, \citet{Mayr2014} consider the (non-normalized) smoothed C-index
	$E_{\mathbbmsl{P}_0}\left[\I(T_1 < T_2) h_\zeta(f(X_2) -  f(X_1))\right]$, where $h_\zeta: s \mapsto \left\{1 + \exp(s/\zeta)\right\}$ is the sigmoid function and $\zeta$ is a tuning parameter determining the smoothness of the approximation. A $\tau$-restricted smoothed C-index is then given by 
	\begin{align*}
		E_{\mathbbmsl{P}_0}\left[\I(T_1 \leq T_2, T_1 \leq \tau) h_\zeta(f(X_2) -  f(X_1))\right].
	\end{align*}
	Under the assumption that $C$ and $X$ are independent, \citet{Mayr2014} use IPC weights to identify this parameter as
	\begin{align*}
		E_0\left[h_\zeta(f(X_2) - f(X_1))\frac{\Delta_1\I(Y_1 < Y_2, Y_1 \leq \tau)}{\tilde{G}_0(Y_1)^2}\right],
	\end{align*}
	where $\tilde{G}_0(t)$ is the identified marginal survival function of $C$ evaluated at $t$. We leverage our identification of the C-index based on the assumption of conditionally independent censoring, which yields the smoothed objective function
	\begin{align*}
		\iiiint h_\zeta(f(X_2) - f(X_1)) \I(t_1 < t_2, t_1 \leq \tau)H_0(dx_1, dt_1)H_0(dx_2,dt_2)\ . 
	\end{align*}
	The empirical objective function to be numerically maximized is then given by
	\begin{align}
		f\mapsto \sum_{i=1}^{n}\sum_{j=1}^{n}h_\zeta(f(X_j) - f(X_i))w_{i,j}\ ,\label{eq:objective}
	\end{align}
	where $w_{i,j} := \iint \I(t_i < t_j, t_i \leq \tau)F_n(dt_i \midd X_i)F_n(dt_j \midd X_j)$ is a weight. The gradient of this objective function with respect to $f(X_j)$ is given by
	\begin{align}
		\sum_{i=1}^{n}\sum_{j=1}^{n}w_{i,j}\frac{\exp(\frac{f(X_j) - f(X_i)}{\zeta})}{\zeta\left\{1 + \exp(\frac{f(X_j) - f(X_i)}{\zeta})\right\}^2}\ .\label{eq:gradient}
	\end{align}
	In the numerical experiments we present in Section \ref{sec:sims} and Section \ref{sec:addtl sims}, we used the empirical objective (\ref{eq:objective}) and gradient (\ref{eq:gradient}) to define a custom family in the \texttt{mboost} gradient boosting \texttt{R} package. The boosting tuning parameters \texttt{mstop} (number of boosting iterations) and \texttt{nu} (learning rate), as well as the smoothing parameter $\zeta$, were selected by five-fold cross-validation. The objective function used for cross-validation was the unsmoothed plug-in estimate of the C-index $\sum_{i=1}^{n}\sum_{j=1}^{n}\I(f(X_j) - f(X_i))w_{i,j}$. 
	
	Calculation of the double sums in (\ref{eq:objective}) and (\ref{eq:gradient}) can be computationally expensive. In practice, we find that subsampling observations for the boosting procedure can greatly decrease computation times, and the performance of the procedure in simulation is relatively insensitive to the choice of subsample size; see Section \ref{sec:addtl sims}.
	
	\section{Simulation details and additional results in Scenarios I--IV}\label{sec:addtl sims}
	
	\subsection{Details on data-generating mechanism and true VIMs}
	
	In all experiments, we generated the data sequentially as follows:
	\begin{enumerate}
		\item draw $X\sim MVN(0, \Sigma)$; 
		\item draw $\varepsilon_T \sim N(0,1)$ and $\varepsilon_C \sim N(0, 1)$, independent of each other and independent of $X$; 
		\item set $\log(T) = \tfrac{5}{10}X_1 -\tfrac{3}{10}X_2 + \tfrac{1}{10}X_1X_2 - \tfrac{1}{10}X_3X_4 + \tfrac{1}{10}X_1X_5 + \varepsilon_T$;
		\item set $\log(C) = \beta_{0,C} - \tfrac{2}{10}X_1 + \tfrac{2}{10}X_2  + \varepsilon_C$;
		\item set $Y := \min\{T, C\}$ and $\Delta = \I(T \leq C)$.
	\end{enumerate}
	The values of $p$ and $\Sigma$ depend on the scenario. Under this data-generating mechanism, the conditionally independent censoring assumption Condition \ref{condition:independent censoring} holds. The conditional distributions of both $T$ and $C$ fall in the class of accelerated failure time models.
	
	In Scenario I, we set $p = 25$ and set $\Sigma$ such that $X_1, \ldots, X_{25}$ each had variance 1, with $\Sigma_{16} = 0.7$ and $\Sigma_{23} = -0.3$. In Scenario II, we set $p = 25$ and $\Sigma$ to be the identity matrix. In Scenario III, we set $p = 5$ and $\Sigma$ to be the identity matrix. In Scenario IV, we set $p = 5$ and $\Sigma$ to be the identity matrix. To achieve censoring rates of $\{30\%, 40\%, 50\%, 60\%, 70\%\}$, we set $\beta_{0,C}$ equal to $\{0.85, 0.5, 0, -0.45, -0.85\}$, respectively.
	
	In the following, we define
	\begin{align*}
		W&:=(X_1, X_2, X_3, X_4, X_5, X_1X_2, X_3X_4, X_1X_5, X_6, \ldots, X_{25})\ ,\\
		\beta_{T} &:= (\tfrac{5}{10}, -\tfrac{3}{10}, 0, 0,0, \tfrac{1}{10}, -\tfrac{1}{10}, \tfrac{1}{10}, 0, \ldots, 0)\ ,
	\end{align*}
	allowing us to compactly write $\log T = W\beta_{T} + \varepsilon_{T}$. The true survival function $\mathbbmsl{S}_0(t \midd x)$ is given by
	\begin{align*}
		P(T > t \midd X = x) &= P(\log(T) > \log(t) \midd X = x) = P(W\beta_T + \varepsilon_T > \log(t) \midd X = x) \\
		&= P(\varepsilon_T > \log(t) - w\beta_T) = 1 - \Omega(\log(t) - w\beta_T)\ .
	\end{align*}
	where $\Omega$ is the standard normal distribution function. The reduced-dimension conditional survival function omitting $X_s$ is given by
	\begin{align*}
		P(T > t \midd X_{-s} = x_{-s}) &= P(\log(T) > \log(t) \midd X_{-s} = x_{-s})\\
		&= P(W\beta_T + \varepsilon_T > \log(t) \midd X_{-s} = x_{-s}) 
	\end{align*}
	Specifically, for $s = 1, s = 2$, $s = (1,6)$, and $s = 6$, we have
	\begin{align*}
		&P(T > t \midd X_{-1} = x_{-1}) = P\{(\tfrac{5}{10} + \tfrac{1}{10}x_2 + \tfrac{1}{10}x_5)X_1 + \varepsilon_T > \log(t) + \tfrac{3}{10}x_2 + \tfrac{1}{10}x_3x_4 \midd X_{-1} = x_{-1}\}\ ,\\
		&P(T > t \midd X_{-2} = x_{-2}) \\
		&\hspace{2cm}= P\{(-\tfrac{3}{10} + \tfrac{1}{10}x_1 )X_2 + \varepsilon_T > \log(t) - \tfrac{5}{10}x_1 + \tfrac{1}{10}x_3x_4 - \tfrac{1}{10}x_1x_5 \midd X_{-2} = x_{-2}\}\ ,\\
		&P(T > t \midd X_{-(1,6)} = x_{-(1,6)})\\
		&\hspace{2cm}=P\{(\tfrac{5}{10} + \tfrac{1}{10}x_2 + \tfrac{1}{10}x_5)X_1 + \varepsilon_T > \log(t) + \tfrac{3}{10}x_2 + \tfrac{1}{10}x_3x_4 \midd X_{-(1,6)} = x_{-(1,6)}\}\ ,\\
		&P(T > t \midd X_{-6} = x_{-6}) = P(T > t \midd X = x) = 1 - \Omega(\log(t) - w\beta_T)\ .
	\end{align*}
	The conditional distributions required to compute the first three probabilities above are given by
	\begin{align*}
		X_1 \midd X_2, \ldots, X_{25} &\sim N\{\Sigma_{16}X_6, (1 - \Sigma_{16}^2)\}\ ,\\
		X_2 \midd X_1, X_3, \ldots, X_{25} &\sim N\{\Sigma_{23}X_3, (1 - \Sigma_{23}^2)\}\ ,\\
		X_1  \midd X_2, \ldots, X_5, X_7,\ldots, X_{25} &\sim N(0, 1)\ .
	\end{align*}
	We note that $\varepsilon_T$ is a standard normal random variable independent of $X$, and so
	\begin{align*}
		&(\tfrac{5}{10} + \tfrac{1}{10}x_2 + \tfrac{1}{10}x_5)X_1 + \varepsilon_T \midd X_{-1} = x_{-1} \\
		&\hspace{4cm}\sim N\{(\tfrac{5}{10} + \tfrac{1}{10}x_2 + \tfrac{1}{10}x_5)\Sigma_{16}x_6, (\tfrac{5}{10} + \tfrac{1}{10}x_2 + \tfrac{1}{10}x_5)^2(1 - \Sigma_{16}^2)\}\ ,\\
		&(-\tfrac{3}{10} + \tfrac{1}{10}x_1 )X_2 + \varepsilon_T \midd X_{-2} \sim N\{(-\tfrac{3}{10} + \tfrac{1}{10}x_1)\Sigma_{23}x_3, (-\tfrac{3}{10} + \tfrac{1}{10}x_1)^2(1 - \Sigma_{23}^2)\}\ ,\\
		&(\tfrac{5}{10} + \tfrac{1}{10}x_2 + \tfrac{1}{10}x_5)X_1 + \varepsilon_T \midd X_{-(1,6)} \sim N(0, (\tfrac{5}{10} + \tfrac{1}{10}x_2 + \tfrac{1}{10}x_5)^2)\ .
	\end{align*}
	
	For landmark time VIMs, the true oracle and residual oracle prediction functions are characterized by the conditional survival function $\mathbbmsl{S}_0(t \midd x)$. For the C-index, because the distribution of $T$ given $X$ is normal, a valid oracle prediction function is given by the conditional mean $x \mapsto E_{\mathbbmsl{P}_0}\left(T \midd X = x\right)$. This is due to the fact that, for two independent normally distributed random variables $T_1$ and $T_2$ with respective means $\mu_1$ and $\mu_2$ and equal variances, $T_1$ stochastically dominates $T_2$ if $\mu_1 > \mu_2$. This implies that $\mathbbmsl{P}_0(T_1 > T_2) > 1/2$, and hence, the mean falls in $\mathcal{F}_0$, the class of oracle prediction functions. The true VIM values for all simulations are given in Table \ref{tab:true vims main}.
	
	\begin{table}
		\centering \begin{tabular}{cc cc cc c} 
			\toprule
			Scenario& Feature& \multicolumn{5}{c}{Importance measure} \\ \midrule
			&&\multicolumn{2}{c}{AUC\ at $\tau$}&\multicolumn{2}{c}{Brier score at $\tau$}&C-index\\    
			&&$\tau = 0.5$&0.9&0.5&0.9&\\ 
			I & $X_1$&0.051&0.049&0.011&0.014&0.042\\
			& $X_6$&0&0&0&0&0\\
			& $(X_1, X_6)$&0.116&0.113&0.021&0.028&0.095\\
			II & $X_1$&0.117&0.113&0.021&0.028&0.093\\
			& $X_6$&0&0&0&0&0\\
			III& $X_1$&0.117&0.113&0.021&0.028&0.093\\
			&$X_2$ & 0.034 & 0.033 & 0.008 & 0.010 &0.029\\
			IV& $X_1$&0.117&0.113&0.021&0.028&0.093\\ \bottomrule
		\end{tabular}\caption{Approximate values of $\psi_{0,s}$ for numerical experiments. These parameter values were approximated using a Monte Carlo approach with sample size $2\times10^7$.}
		\label{tab:true vims main}
	\end{table} 
	
	\subsection{Details on nuisance parameter and oracle prediction function estimation}\label{sec:nuisance estimation details}
	
	Table \ref{tab:nuisance estimators} describes the algorithms used to estimate $S_0$ and $G_0$.  Tuning parameters for the random survival forest (random survival forest) were selected to minimize out-of-bag error rate, as measured by one minus Harrell's C-index (the default evaluation metric in the \texttt{rfsrc} software package). Table \ref{tab:SL algorithms} gives the algorithms included in the Super Learner library for global survival stacking and estimation of the residual oracle prediction function $f_{0,s}$ for landmark VIMs. Five-fold cross-validation was used to determine the optimal convex combination of these learners that minimized cross-validated squared-error loss, as described in \citet{Wolock2024}. Table \ref{tab:survSL algorithms} gives the algorithms included in the survival Super Learner library. Five-fold cross-validation was used to determine the optimal convex combination of these learners that minimized cross-validated oracle risk functions detailed in \citet{Westling2023}. The random survival forest algorithm was fitted twice, once to estimate $F_0$ and once to estimate $G_0$. For global survival stacking and survival Super Learner, estimates for both distributions are produced simultaneously.
	
	For landmark VIMs, the full oracle prediction function $f_0$ is a simple transformation of $S_0$ and was not estimated separately in the primary experiments. For the C-index, we implemented the boosting procedure details in Section \ref{sec:gradient boosted c} with five-fold cross-validation for tuning parameter selection. The unsmoothed C-index was used as the evaluation metric for cross-validation. The tuning parameters are detailed in Table \ref{tab:boosting tuning params}. 
	
	\begin{table}\centering \begin{tabular}{p{0.25\linewidth}  p{0.3\linewidth}  p{0.35\linewidth}} 
			\toprule
			Algorithm & \texttt{R} implementation & Tuning parameters\\ \midrule
			Random survival forest & \texttt{rfsrc} &\texttt{mtry} $\in \{1,\dots,\sqrt{p}^\dagger\}$ \\ 
			&\citep{rfsrc-package}& \texttt{nodesize} $\in \{5,15,25\}$\\
			&& \texttt{ntree} $\in \{500, 1000\}$\\
			Global survival stacking& \texttt{stackG}& \texttt{SL.library} (see Table \ref{tab:SL algorithms}) \\
			& \citep{Wolock2024}&\texttt{bin\_size = 0.04}\\
			Survival Super Learner & \texttt{survSuperLearner} & \texttt{SL.library} (see Table \ref{tab:survSL algorithms})\\
			&\citep{Westling2023}&\\ \bottomrule
		\end{tabular}
		\caption{Algorithms used for estimation of nuisance parameters. All options besides those listed here were set to default values. In particular, the random survival forests were grown using sampling without replacement and the log-rank splitting rule. The combination of \texttt{mtry}, \texttt{nodesize} and \texttt{ntree} minimizing out-of-bag error rate, as measured by one minus Harrell's C-index, was selected. For global survival stacking, \texttt{time\_basis} was set to \texttt{"continuous"} (time included as continuous predictor in the pooled binary regression), and \texttt{surv\_form} was set to \texttt{"PI"} (product-integral mapping from hazard to survival function). For both global stacking and survival Super Learner, five-fold cross-validation was used to determine the optimal convex combination of algorithms in \texttt{SL.library}. 
			\textsuperscript{\textdagger}: $p$ denotes the number of predictors.}
		\label{tab:nuisance estimators}
	\end{table} 
	
	\begin{table}\centering \begin{tabular}{l l l }
			\toprule
			Algorithm name & Algorithm description & Tuning parameters\\ \midrule
			\texttt{SL.mean} & Marginal mean&NA\\
			\texttt{SL.glm.int} & Logistic regression with pairwise interactions&NA\\ 
			\texttt{SL.gam} & Generalized additive model&default\\
			\texttt{SL.earth} & Multivariate adaptive regression splines&default\\
			\texttt{SL.ranger} & Random forest&default\\
			\texttt{SL.xgboost} & Gradient-boosted trees&\texttt{ntrees} $\in \{250, 500, 1000\}$\\
			&&\texttt{max\_depth} $\in \{1, 2\}$ \\ \bottomrule
		\end{tabular}
		\caption{Algorithms included in the Super Learner for global survival stacking and for estimation of the residual oracle prediction function for landmark VIMs. All tuning parameters besides those for \texttt{SL.xgboost} were set to default values. In particular, \texttt{gam} was implemented with \texttt{degree = 2}; \texttt{earth} with \texttt{degree = 2, penalty = 3, nk =} number of predictors plus 1, \texttt{endspan = 0}, \texttt{minspan = 0}; and \texttt{ranger} with \texttt{num.trees = 500, mtry =} the square root of the number of predictors, \texttt{min.node.size = 1, sample.fraction = 1} with replacement. For \texttt{SL.xgboost}, \texttt{shrinkage} was set to 0.01, \texttt{minobspernode} was set to 10, and each combination of \texttt{ntrees} and \texttt{max\_depth} was included in the Super Learner library.}
		\label{tab:SL algorithms}
	\end{table} 
	
	\begin{table}\centering \begin{tabular}{p{0.25\linewidth}  p{0.65\linewidth}}
			\toprule Algorithm name & Algorithm description \\ \midrule
			\texttt{survSL.km} & Kaplan-Meier estimator\\	
			\texttt{survSL.expreg.int} & Survival regression assuming event and censoring times follow an exponential distribution conditional on covariates, includes two-way interactions\\ 	
			\texttt{survSL.weibreg.int} & Survival regression assuming event and censoring times follow a Weibull distribution conditional on covariates, includes two-way interactions\\	
			\texttt{survSL.loglogreg.int} & Survival regression assuming event and censoring times follow a log-logistic distribution conditional on covariates, includes two-way interactions\\	
			\texttt{survSL.AFTreg.int} & Survival regression assuming event and censoring times follow a log-normal distribution conditional on covariates, includes two-way interactions\\	
			\texttt{survSL.coxph} & Main-terms Cox proportional hazards estimator with Breslow baseline cumulative hazard\\	
			\texttt{survSL.rfsrc} & Random survival forest as implemented in the \texttt{randomForestSRC} package\\ \bottomrule
		\end{tabular}
		\caption{Algorithms included in the survival Super Learner. All tuning parameters were set to default values. In particular, \texttt{gam} was implemented with \texttt{degree = 1}; and \texttt{rfsrc} with \texttt{ntree = 500}, \texttt{mtry =} the square root of the number of predictors, \texttt{nodesize = 15}, \texttt{splitrule = "logrank"}, \texttt{sampsize = 1} with replacement. The screening algorithm \texttt{survscreen.marg}, which uses a univariate Cox regression $p$-value, was applied to all algorithms.}
		\label{tab:survSL algorithms}
	\end{table}

	\begin{table}\centering \begin{tabular}{p{0.12\linewidth}  p{0.5\linewidth}  p{0.25\linewidth}} 
			\toprule Parameter & Description & Possible values \\ \midrule
			\texttt{mstop} &Number of boosting iterations&$\{100, 200, 300, 400, 500\}$\\
			\texttt{nu} &Learning rate&$0.1$ \\
			\texttt{zeta}&Smoothing parameter for sigmoid function&$\{0.01, 0.05\}$ \\ \bottomrule
		\end{tabular}\caption{Tuning parameters for the C-index boosting procedure.}
		\label{tab:boosting tuning params}
	\end{table} 
	
	\subsection{Additional results in Scenario I}
	
	In this subsection, we present additional results in Scenario I, where we set $p = 25$ and $\Sigma$ to be a $25 \times 25$ matrix with 1 on the diagonal. Off-diagonal elements were set to 0, except for $\Sigma_{1,6} = \Sigma_{6,1} = 0.7$ and $\Sigma_{2,3} = \Sigma_{3,2} = -0.3$. We generated 500 random datasets of size $n \in \{500, 750, \dots, 1500\}$ and assessed the importance of $X_1$ and $X_6$ individually, as well as the joint importance of $(X_1, X_6)$, using Algorithm \ref{alg:alternative, ss xfit}, which involves sample splitting and provides valid inference under the null hypothesis of zero importance. We estimated nuisance parameters as in other settings. We present results for empirical bias scaled by $n^{1/2}$, empirical variance scaled by $n$, empirical confidence interval coverage, average confidence interval width (for $X_1$ and $(X_1, X_6)$), and empirical type I error (for $X_6$).
	
	Results for Scenario I are displayed in Figures \ref{fig:X1_scenario4}--\ref{fig:X14_scenario4}. Generally, the proposed procedure performs similarly as in other simulation settings --- the presence of correlated features seems to have little impact on operating characteristics. Without cross-fitting, the random survival forests and global stacking estimators suffer from non-negligible bias, confidence interval coverage below the nominal level, and inflated type I error. With cross-fitting, all estimators perform well as the sample size increases. Furthermore, the procedure performs equally well for groups of features as for individual features. When features are expected to be correlated, considering groups of features may improve the interpretability of VIM analyses. 
	
	\begin{figure}
		\includegraphics[width=\linewidth]{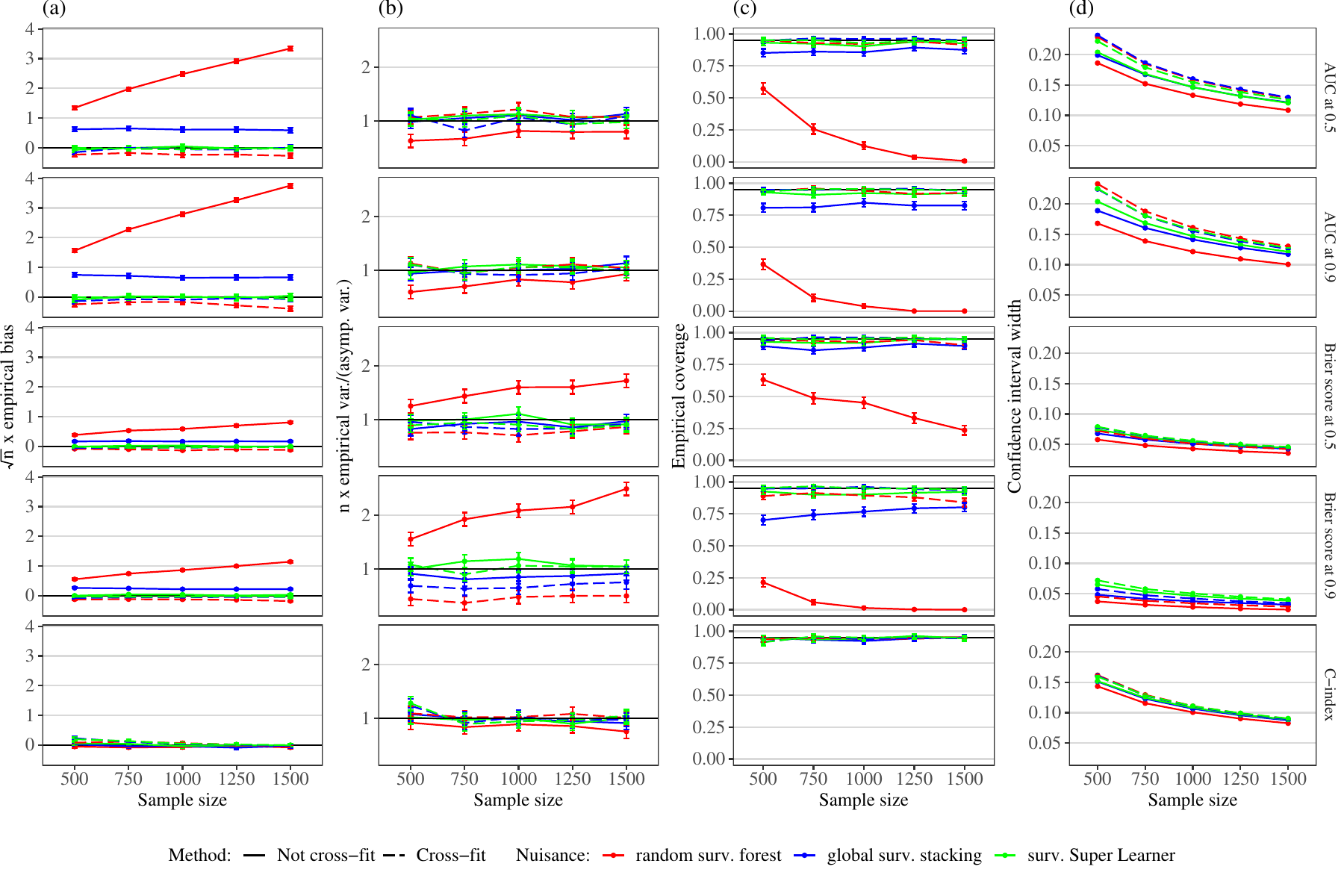}
		\caption{Performance of the one-step VIM estimator for the importance of $X_1$ in Scenario I. The  VIMs shown are AUC\ and Brier score at times 0.5 and 0.9 and the C-index restricted to time 0.9.  (a) empirical bias scaled by $\sqrt{n}$; (b) empirical variance scaled by $n / \sigma^2$, where $\sigma^2$ is the theoretical asymptotic variance of the estimator; (c) empirical coverage of nominal 95\% confidence intervals; (d) average confidence interval width. The colors denote different nuisance estimators, which were used to estimate both event and censoring distributions. Solid and dashed lines denote non-cross-fitted and cross-fitted estimators, respectively. Vertical bars represent 95\% confidence intervals taking into account Monte Carlo error.}
		\label{fig:X1_scenario4}
	\end{figure}
	
	\begin{figure}
		\includegraphics[width=\linewidth]{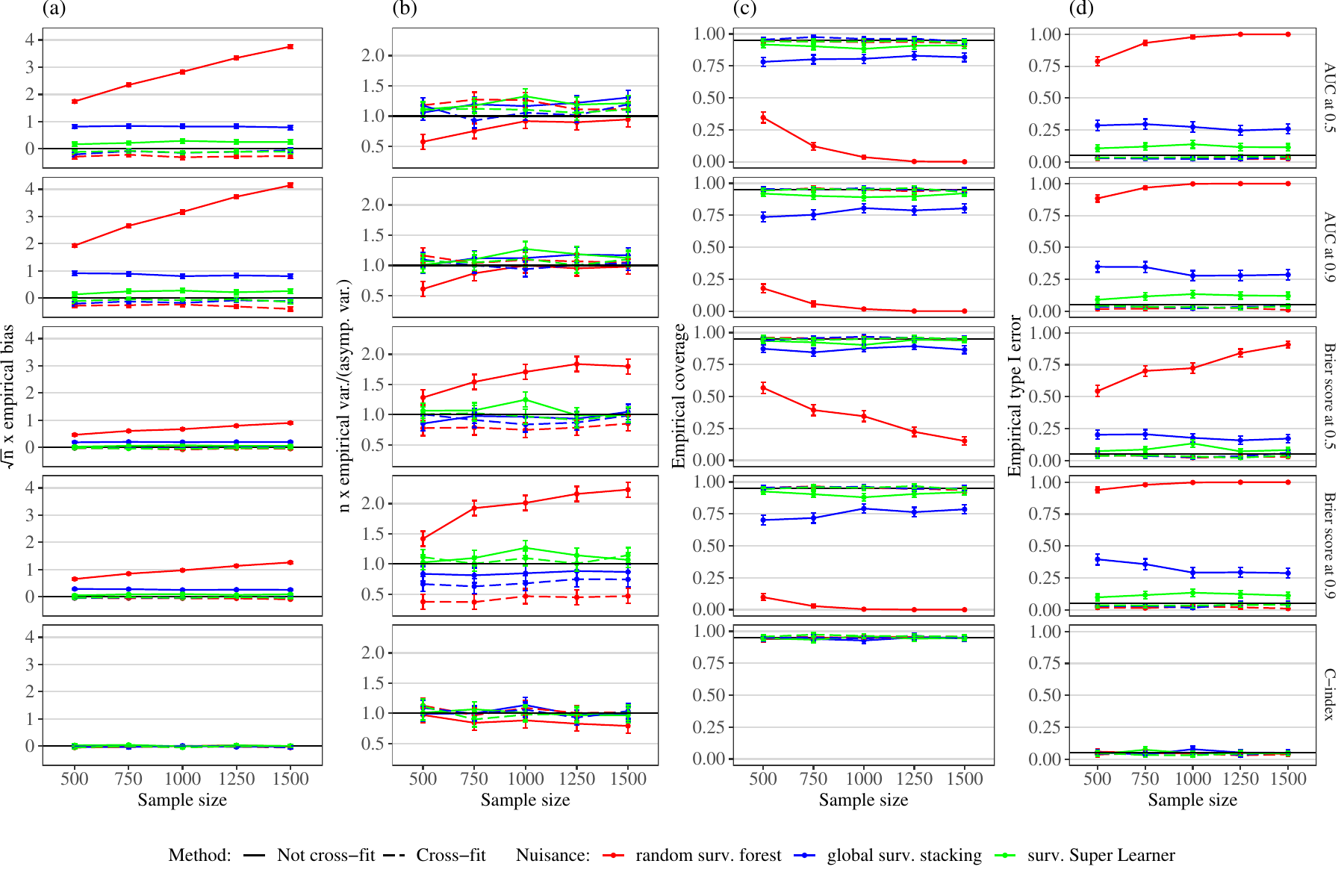}
		\caption{  Performance of the one-step VIM estimator for the (zero) importance of $X_6$ in Scenario I. The VIMs shown are AUC\ and Brier score at times 0.5 and 0.9 and the C-index restricted to time 0.9. (a) empirical bias scaled by $\sqrt{n}$; (b) empirical variance scaled by $n / \sigma^2$, where $\sigma^2$ is the theoretical asymptotic variance of the estimator; (c) empirical coverage of nominal 95\% confidence intervals; (d) empirical type I error. The colors denote different nuisance estimators, which were used to estimate both event and censoring distributions. Solid and dashed lines denote non-cross-fitted and cross-fitted estimators, respectively. Vertical bars represent 95\% confidence intervals taking into account Monte Carlo error.}
		\label{fig:X4_scenario4}
	\end{figure}
	
	\begin{figure}
		\includegraphics[width=\linewidth]{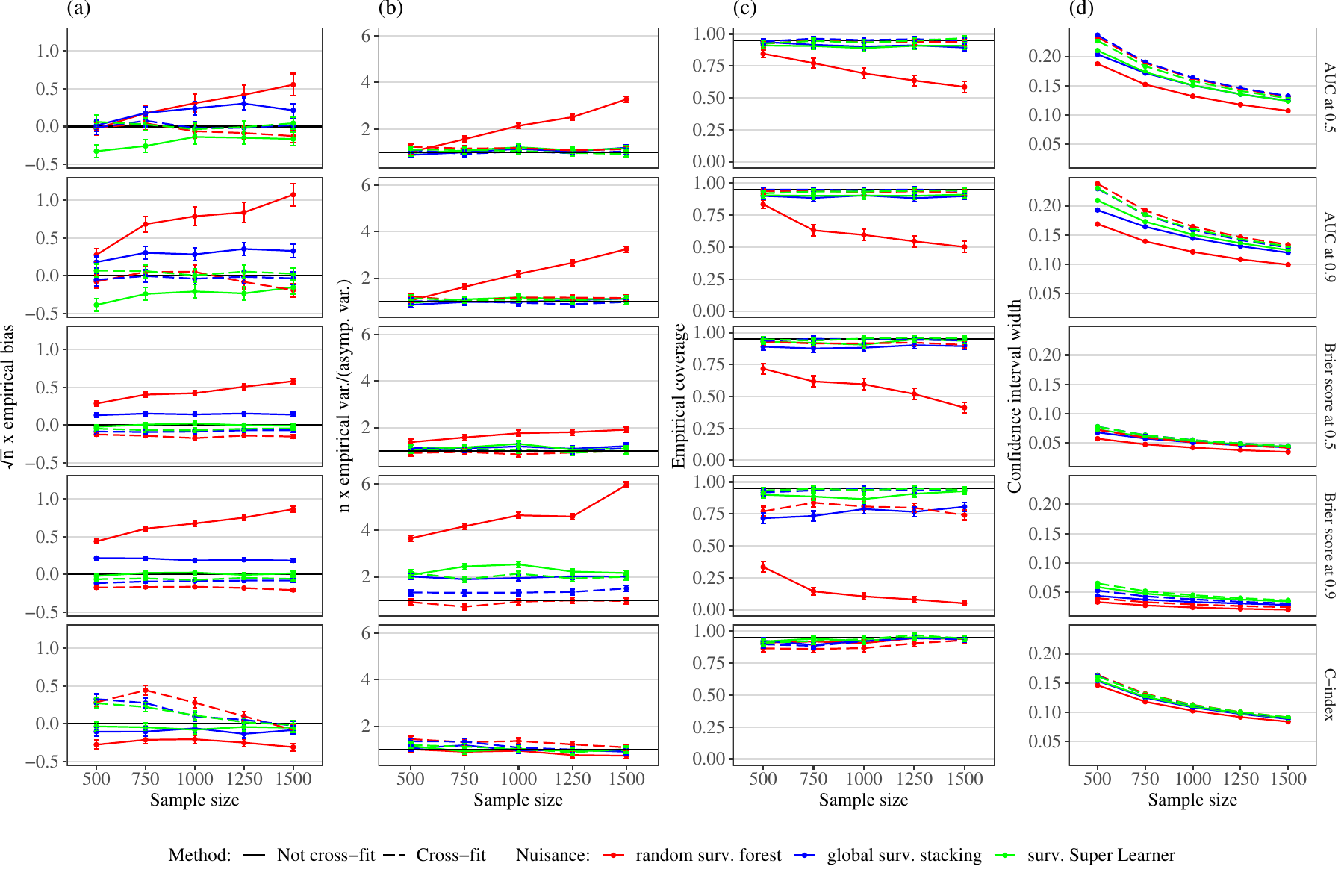}
		\caption{Performance of the one-step VIM estimator for the joint importance of $(X_1,X_6)$ in Scenario I. The VIMs shown are AUC\ and Brier score at times 0.5 and 0.9 and the C-index restricted to time 0.9. (a) empirical bias scaled by $\sqrt{n}$; (b) empirical variance scaled by $n / \sigma^2$, where $\sigma^2$ is the theoretical asymptotic variance of the estimator; (c) empirical coverage of nominal 95\% confidence intervals; (d) average confidence interval width. The colors denote different nuisance estimators, which were used to estimate both event and censoring distributions. Solid and dashed lines denote non-cross-fitted and cross-fitted estimators, respectively. Vertical bars represent 95\% confidence intervals taking into account Monte Carlo error. }
		\label{fig:X14_scenario4}
	\end{figure}
	
	\subsection{Simulation results in Scenario II}
	
	In this section, we provide results for Scenario II, in which $p = 25$, the features are independent,  and all features besides $X_1, \ldots, X_5$ have zero importance. We generated 500 random datasets of size $n \in \{500, 750, \dots, 1500\}$ and used Algorithm \ref{alg:alternative, ss xfit}. Figure \ref{fig:X1_scenario2} displays the results for $X_1$, and Fig. \ref{fig:X4_scenario2} displays the results for $X_6$, which has zero importance. The non-cross-fitted estimators demonstrate increased bias, decreased confidence interval coverage, and increased type I error compared to their cross-fitted counterparts. The cross-fitted estimator using random survival forest performs well for assessing the Brier score VIM of $X_6$, but the associated confidence intervals are moderately anti-conservative in the case of $X_1$.
	
	\begin{figure}
		\centering
		\includegraphics[width=\linewidth]{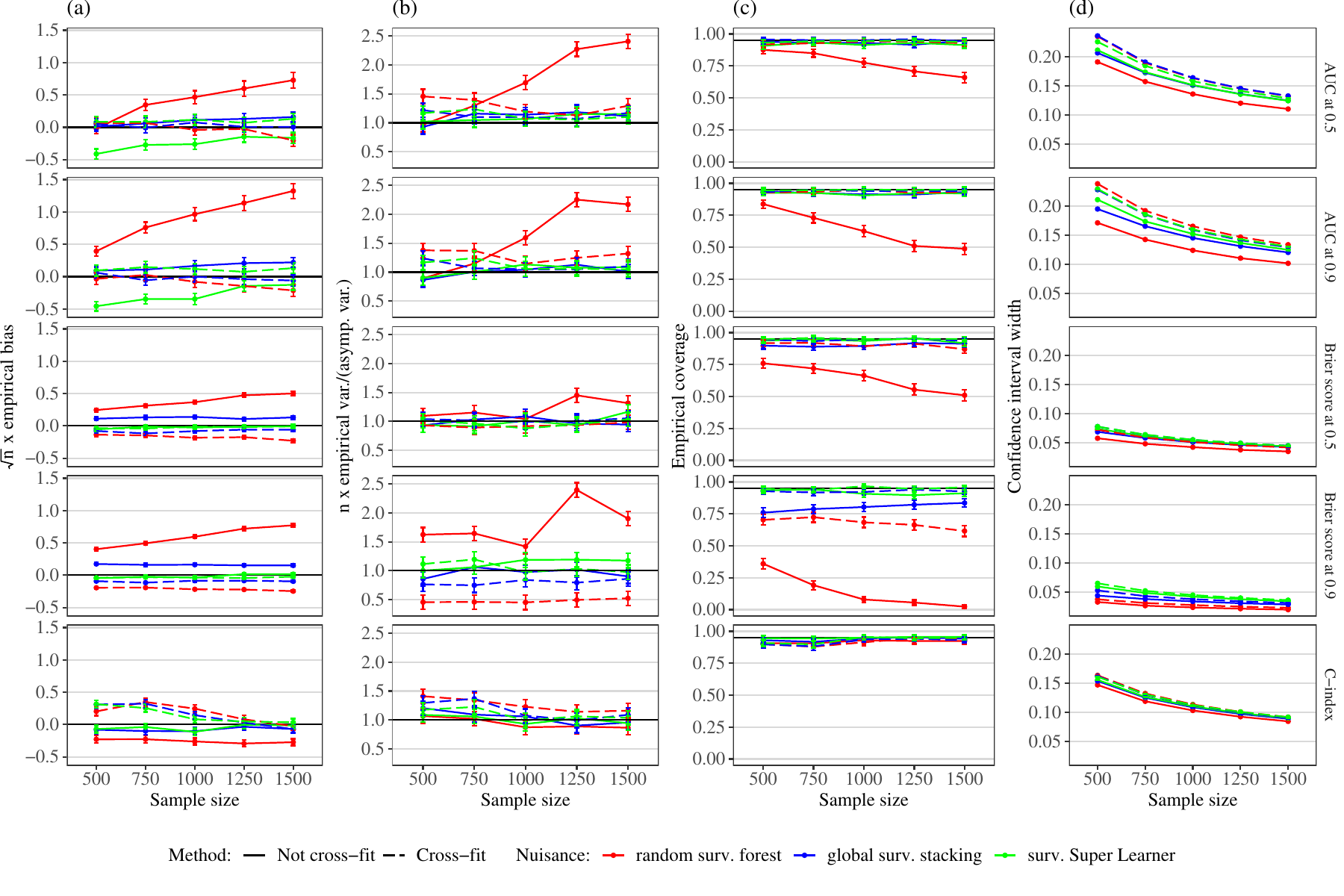}
		\caption{Performance of the one-step VIM estimator for the importance of $X_1$ in Scenario II. The VIMs shown are AUC\ and Brier score at times 0.5 and 0.9 and the C-index restricted to time 0.9. (a) empirical bias scaled by $\sqrt{n}$; (b) empirical variance scaled by $n / \sigma^2$, where $\sigma^2$ is the theoretical asymptotic variance of the estimator; (c) empirical coverage of nominal 95\% confidence intervals; (d) average confidence interval width. The colors denote different nuisance estimators, which were used to estimate both event and censoring distributions. Solid and dashed lines denote non-cross-fitted and cross-fitted estimators, respectively. Vertical bars represent 95\% confidence intervals taking into account Monte Carlo error.}
		\label{fig:X1_scenario2}
	\end{figure}
	
	\begin{figure}
		\centering
		\includegraphics[width=\linewidth]{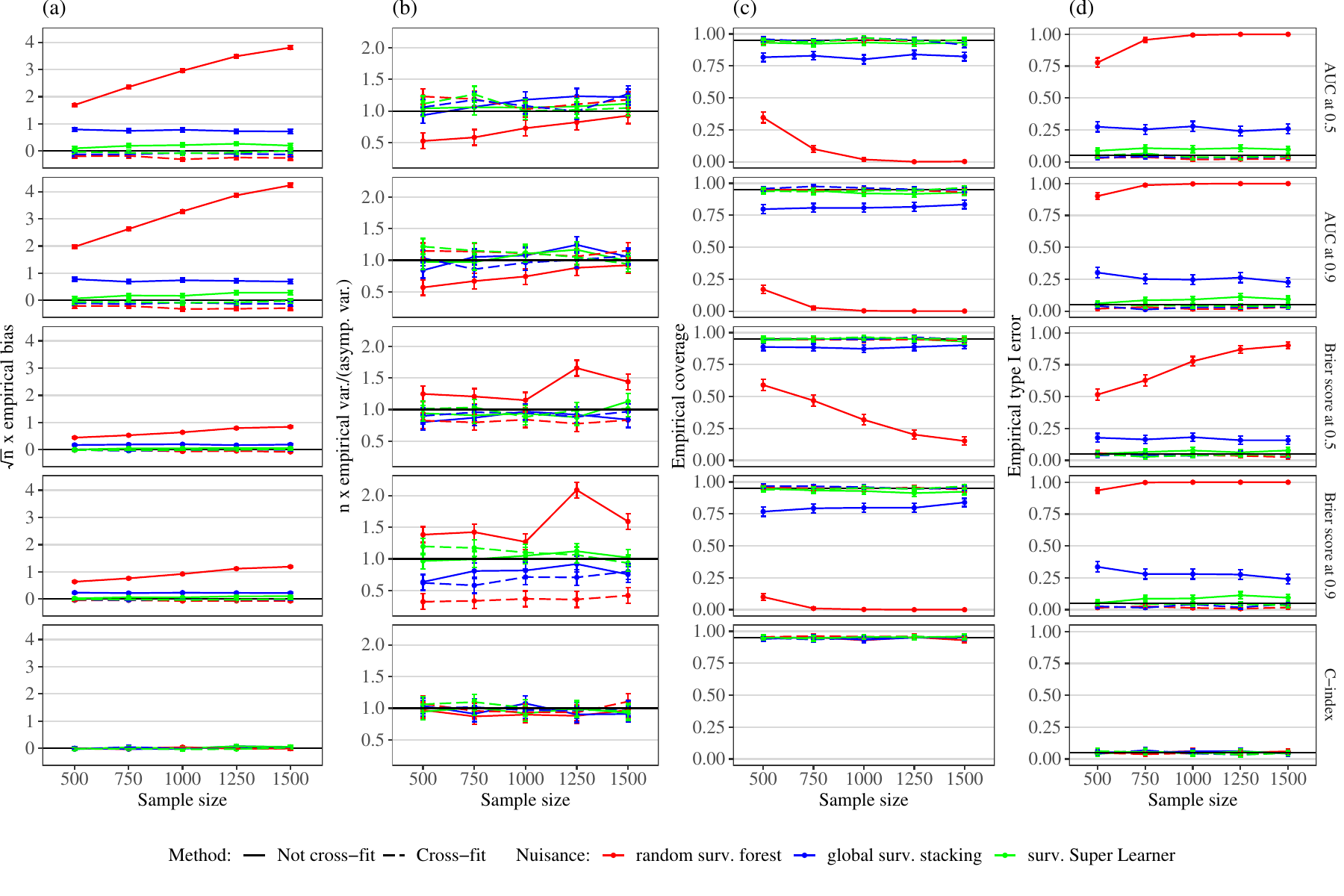}
		\caption{Performance of the one-step VIM estimator for the (zero) importance of $X_6$ in Scenario II. The VIMs shown are AUC\ and Brier score at times 0.5 and 0.9 and the C-index restricted to time 0.9. (a) empirical bias scaled by $\sqrt{n}$; (b) empirical variance scaled by $n / \sigma^2$, where $\sigma^2$ is the theoretical asymptotic variance of the estimator; (c) empirical coverage of nominal 95\% confidence intervals; (d) empirical type I error. The colors denote different nuisance estimators, which were used to estimate both event and censoring distributions. Solid and dashed lines denote non-cross-fitted and cross-fitted estimators, respectively. Vertical bars represent 95\% confidence intervals taking into account Monte Carlo error.}
		\label{fig:X4_scenario2}
	\end{figure}
	
	\subsection{Simulation results in Scenario III}
	In Scenario III, we set $p = 5$, such that all features had non-zero importance, and set $\Sigma$ to be the identity matrix. We generated 500 random datasets of size $n \in \{500, 750, \dots,1500\}$ and used Algorithm \ref{alg:alternative, xfit}, which is valid when the importance is \textit{a priori} known to be non-zero. We assess performance in the same manner as described in the main text. 
	
	Figure \ref{fig:X1_scenario1} displays the results for $X_1$, and Fig. \ref{fig:X2_scenario1} displays the results for $X_2$. We observe that for AUC\ and Brier score, cross-fitting is necessary for good performance. The cross-fitted global stacking and survival Super Learner estimators achieves low bias and coverage within Monte Carlo error of the nominal level, while the random survival forest implementation has larger bias and is somewhat anti-conservative for AUC\ and Brier score VIMs. For the C-index, surprisingly, cross-fitting appears less important for good performance. 
	
	\begin{figure}
		\centering
		\includegraphics[width=\linewidth]{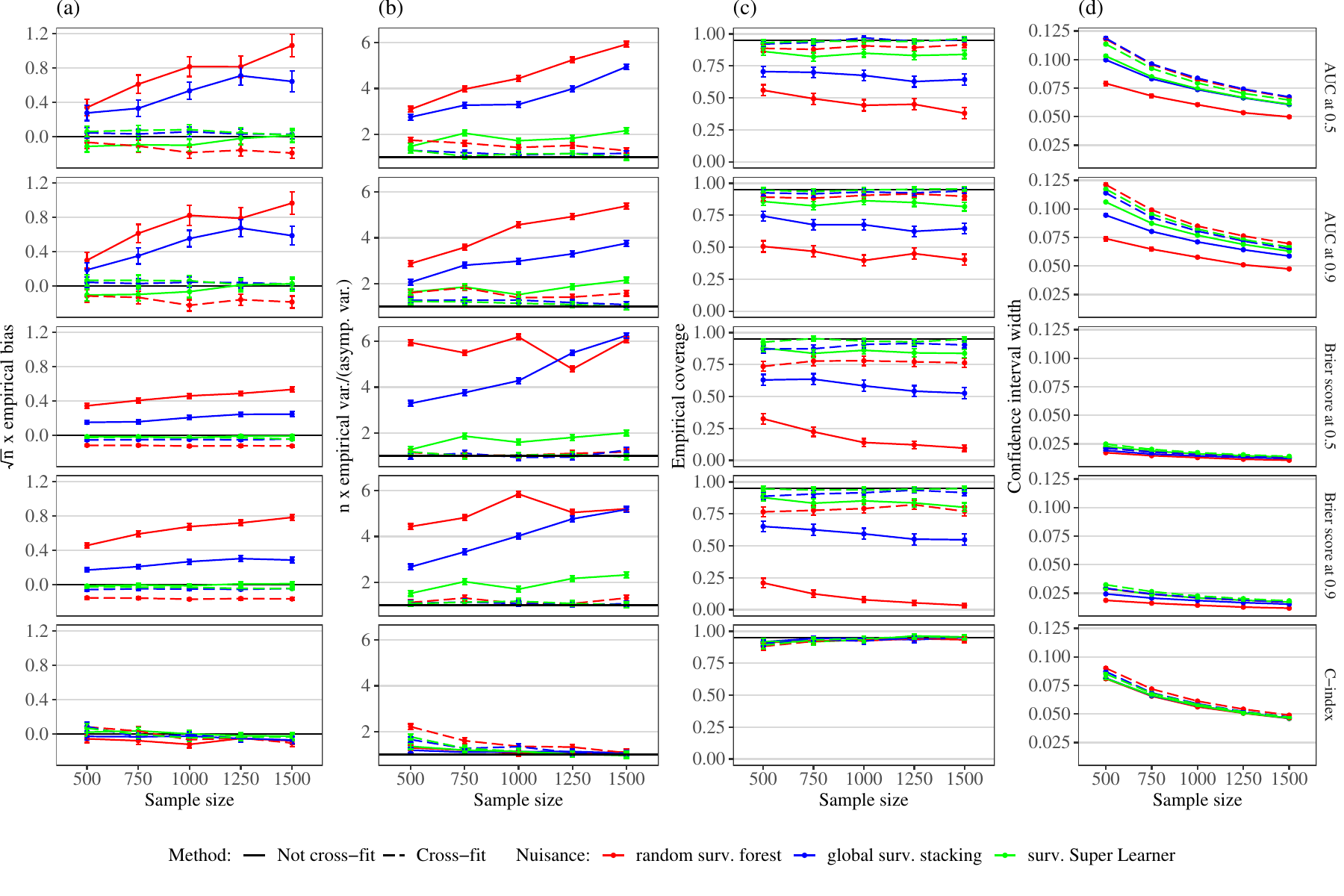}
		\caption{Performance of the one-step VIM estimator for the importance of $X_1$ in Scenario III. The VIMs shown are AUC\ and Brier score at times 0.5 and 0.9 and the C-index restricted to time 0.9. (a) empirical bias scaled by $\sqrt{n}$; (b) empirical variance scaled by $n / \sigma^2$, where $\sigma^2$ is the theoretical asymptotic variance of the estimator; (c) empirical coverage of nominal 95\% confidence intervals; (d) average confidence interval width. The colors denote different nuisance estimators, which were used to estimate both event and censoring distributions. Solid and dashed lines denote non-cross-fitted and cross-fitted estimators, respectively. Vertical bars represent 95\% confidence intervals taking into account Monte Carlo error.}
		\label{fig:X1_scenario1}
	\end{figure}
	
	\begin{figure}
		\centering
		\includegraphics[width=\linewidth]{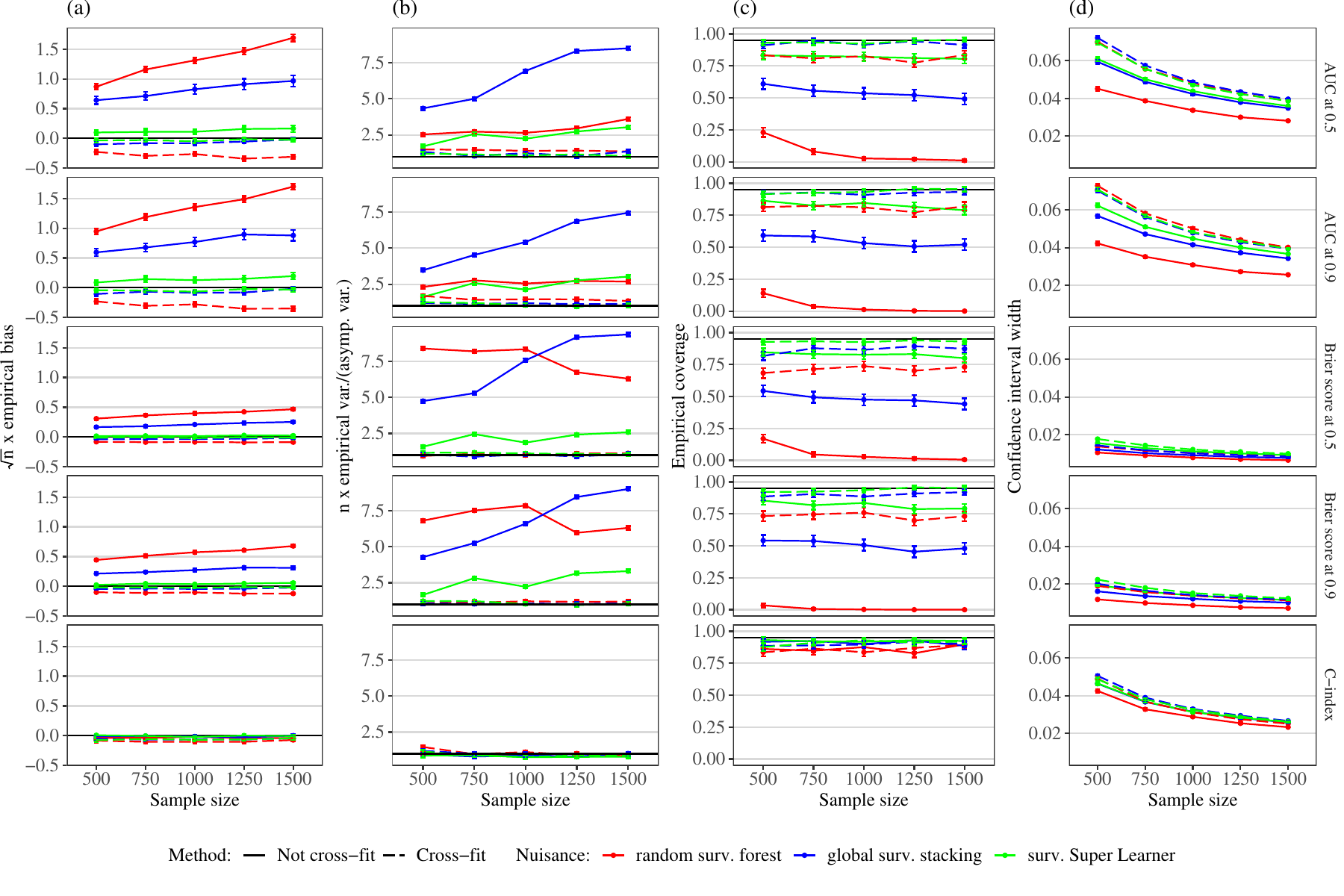}
		\caption{Performance of the one-step VIM estimator for the importance of $X_2$ in Scenario III. The VIMs shown are AUC\ and Brier score at times 0.5 and 0.9 and the C-index restricted to time 0.9. (a) empirical bias scaled by $\sqrt{n}$; (b) empirical variance scaled by $n / \sigma^2$, where $\sigma^2$ is the theoretical asymptotic variance of the estimator; (c) empirical coverage of nominal 95\% confidence intervals; (d) average confidence interval width. The colors denote different nuisance estimators, which were used to estimate both event and censoring distributions. Solid and dashed lines denote non-cross-fitted and cross-fitted estimators, respectively. Vertical bars represent 95\% confidence intervals taking into account Monte Carlo error.}
		\label{fig:X2_scenario1}
	\end{figure}
	
	\subsection{The effect of the censoring rate (Scenario IV)}
	
	In order to study the effect of censoring on our procedure, we performed a simulation study in Scenario IV, in which the censoring rate varied between 30\% and 70\%. The event times and covariates were generated as in Scenario II, while $\beta_{0,C}$ was selected to achieve overall censoring rates in $\{30\%, 40\%,\dots,70\%\}$. For this scenario, we generated 500 random datasets of size 1000. We considered the importance of $X_1$ and $X_6$ using AUC\ and Brier score at landmark times $\tau \in \{0.5, 0.9\}$ and C-index restricted to $\tau = 0.9$. We used Algorithm \ref{alg:alternative, ss xfit} to compute point and standard error estimates, from which we computed nominal 95\% Wald-type confidence intervals. We evaluated performance using the empirical bias, the empirical variance, the empirical confidence interval coverage, and the average confidence interval width (for $X_1$), and empirical type I error rate (for $X_6$). 
	
	In Fig. \ref{fig:cens rate 1} and Fig. \ref{fig:cens rate 6} we display the results of this experiment. The bias of the cross-fitted global stacking and survival Super Learner estimators is largely unaffected by the censoring rate. The bias of the random survival forest estimator is modestly affected by increased censoring, even with cross-fitting. The variance of all estimators tends to increase with increasing censoring, most dramatically at the later landmark time $\tau = 0.9$. The cross-fitted estimators demonstrate nominal coverage. As expected, confidence interval width increases with increased censoring. Overall, we see that the operating characteristics of the procedure are largely consistent across censoring levels. Unsurprisingly, the impact of censoring on estimator variance and confidence interval width is larger at the later landmark time.  
	
	\begin{figure}
		\centering
		\includegraphics[width=\linewidth]{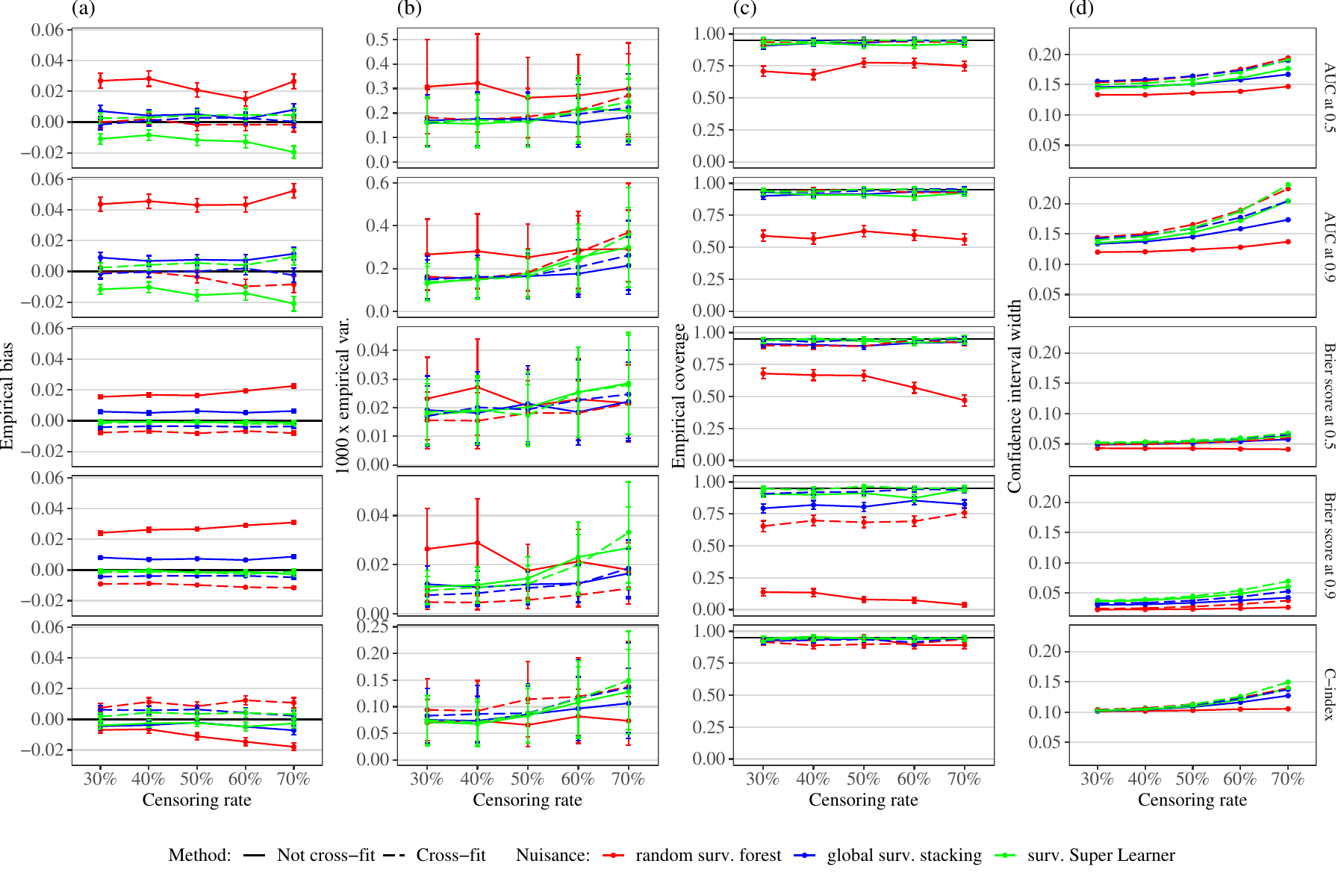}
		\caption{Performance of the one-step VIM estimator for the importance of $X_1$ in Scenario IV. The  VIMs shown are AUC\ and Brier score at times 0.5 and 0.9 and the C-index restricted to time 0.9. (a) empirical bias; (b) empirical variance; (c) empirical coverage of nominal 95\% confidence intervals; (d) average confidence interval width. The colors denote different nuisance estimators, which were used to estimate both event and censoring distributions. Solid and dashed lines denote non-cross-fitted and cross-fitted estimators, respectively. Vertical bars represent 95\% confidence intervals taking into account Monte Carlo error.}
		\label{fig:cens rate 1}
	\end{figure}
	
	\begin{figure}
		\centering
		\includegraphics[width=\linewidth]{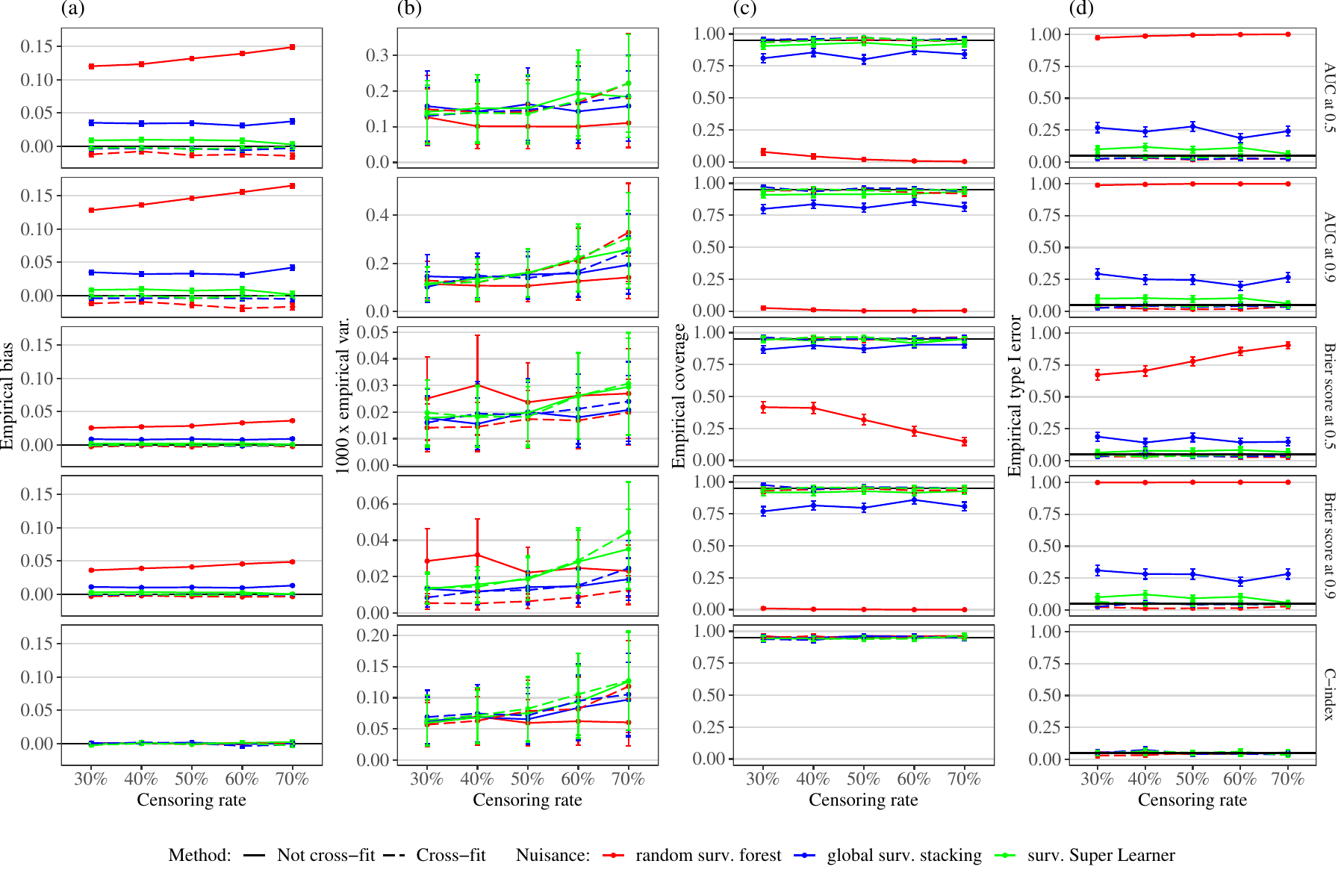}
		\caption{Performance of the one-step VIM estimator for the importance of $X_6$ in Scenario IV. The  VIMs shown are AUC\ and Brier score at times 0.5 and 0.9 and the C-index restricted to time 0.9. (a) empirical bias; (b) empirical variance; (c) empirical coverage of nominal 95\% confidence intervals; (d) average confidence interval width. The colors denote different nuisance estimators, which were used to estimate both event and censoring distributions. Solid and dashed lines denote non-cross-fitted and cross-fitted estimators, respectively. Vertical bars represent 95\% confidence intervals taking into account Monte Carlo error.}
		\label{fig:cens rate 6}
	\end{figure}
	
	\subsection{Subsampling for C-index boosting}
	As noted in Section \ref{sec:gradient boosted c}, the proposed C-index boosting procedure involves repeated calculation of double sums appearing in (\ref{eq:objective}) and (\ref{eq:gradient}), which can be computationally expensive. This computational expense can be reduced using subsampling in the boosting procedure. We conducted a simulation study to assess the effect of subsampling on the performance of the boosting procedure.
	
	We generated data under the settings of Scenario I, simulating 500 random datasets of size $n \in \{500, 750, \ldots, 1500\}$. We considered the importance of $X_1$, with the C-index restricted to $\tau = 0.9$. We used Algorithm \ref{alg:alternative, ss xfit} to compute point and standard error estimates, from which we constructed 95\% Wald-type confidence intervals. We evaluated performance using empirical bias, empirical variance, empirical confidence interval coverage, and average confidence interval width.
	
	Fig. \ref{fig:cindex subsampling} shows the results of this experiment. While the bias of the estimators using smaller subsampling proportions is somewhat larger in sample sizes smaller than 1000, subsampling appears to have little effect as the sample size increases. The variance, confidence interval coverage, and confidence interval width are all similar among the four subsampling proportions used.
	
	\begin{figure}
		\centering
		\includegraphics[width=\linewidth]{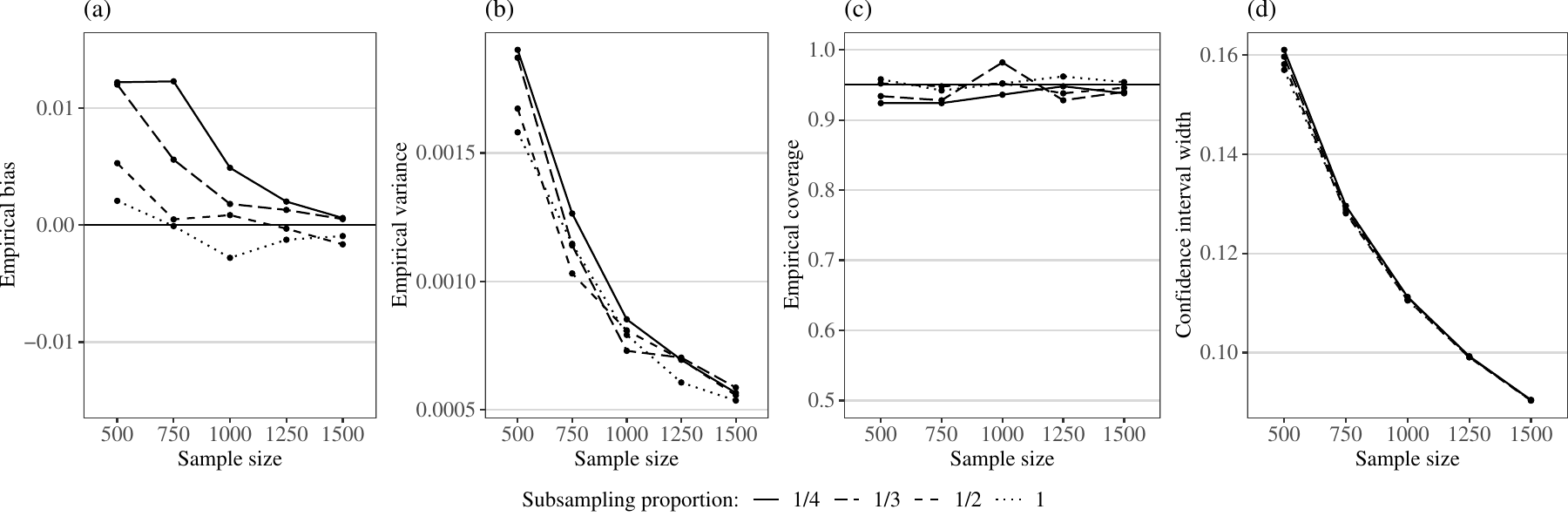}
		\caption{Performance of the one-step VIM estimator for the importance of $X_1$ using different subsampling proportions for C-index boosting. (a) empirical bias; (b) empirical variance; (c) empirical coverage of nominal 95\% confidence intervals; (d) average confidence interval width. Different line types denote different subsampling proportions.}
		\label{fig:cindex subsampling}
	\end{figure}

	\section{Additional data analysis}\label{sec:addtl data analysis}
	In this section, we give the results of the HVTN 702 VIM analysis relative to the full covariate vector. Here, the full feature vector included all available covariates and the reduced feature vector excluded the feature or feature group of interest. This quantifies the loss in predictiveness from exclusion of the feature(s) of interest. As in the analysis presented in the main text, we adjusted for treatment arm assignment. 
	
	In Figure \ref{fig:702 conditional}, we display the results of the this analysis. For the combined cohort, by C-index and by AUC\ across all time horizons, sex assigned at birth is estimated to be the most important feature for predicting the probability of HIV acquisition diagnosis. In this cohort, the ranking of feature groups across the three AUC\ VIM time horizons is similar, with behavioral features ranking behind sex assigned at birth, followed by the remaining five feature groups. The magnitude of the estimated VIMs is similar at the later landmark times, which suggests that the importance of sex assigned at birth and sexual behavior features, while not particularly large, are relatively stable over the time horizons of interest in this context. For C-index, sexual behavior features again rank second in importance. At the 24 and 30 month time horizons, sex assigned at birth achieves statistical significance for AUC\ importance ($p = 0.022$ and $p = 0.007$, respectively). Among participants assigned female sex at birth, the feature ranking is different than that observed in the combined cohort. In particular, sexual behavior features appear to have decreased AUC\ importance among females compared to the combined cohort. As in the combined cohort, the estimated VIMs and associated confidence intervals are similar among the three time horizons for AUC, while for C-index, sexual health features have somewhat larger importance. Among participants assigned male sex at birth, sexual behavior features have the highest estimated AUC\ importance at the 18-month time horizon, with a point estimate of around 0.1, although the estimated magnitude is lower at later time horizons. For C-index importance, sexual behavior features are associated with a change in oracle predictiveness of around 0.15 ($p = 0.010$).
	
	\begin{figure}
		\centering
		\includegraphics[width=\linewidth]{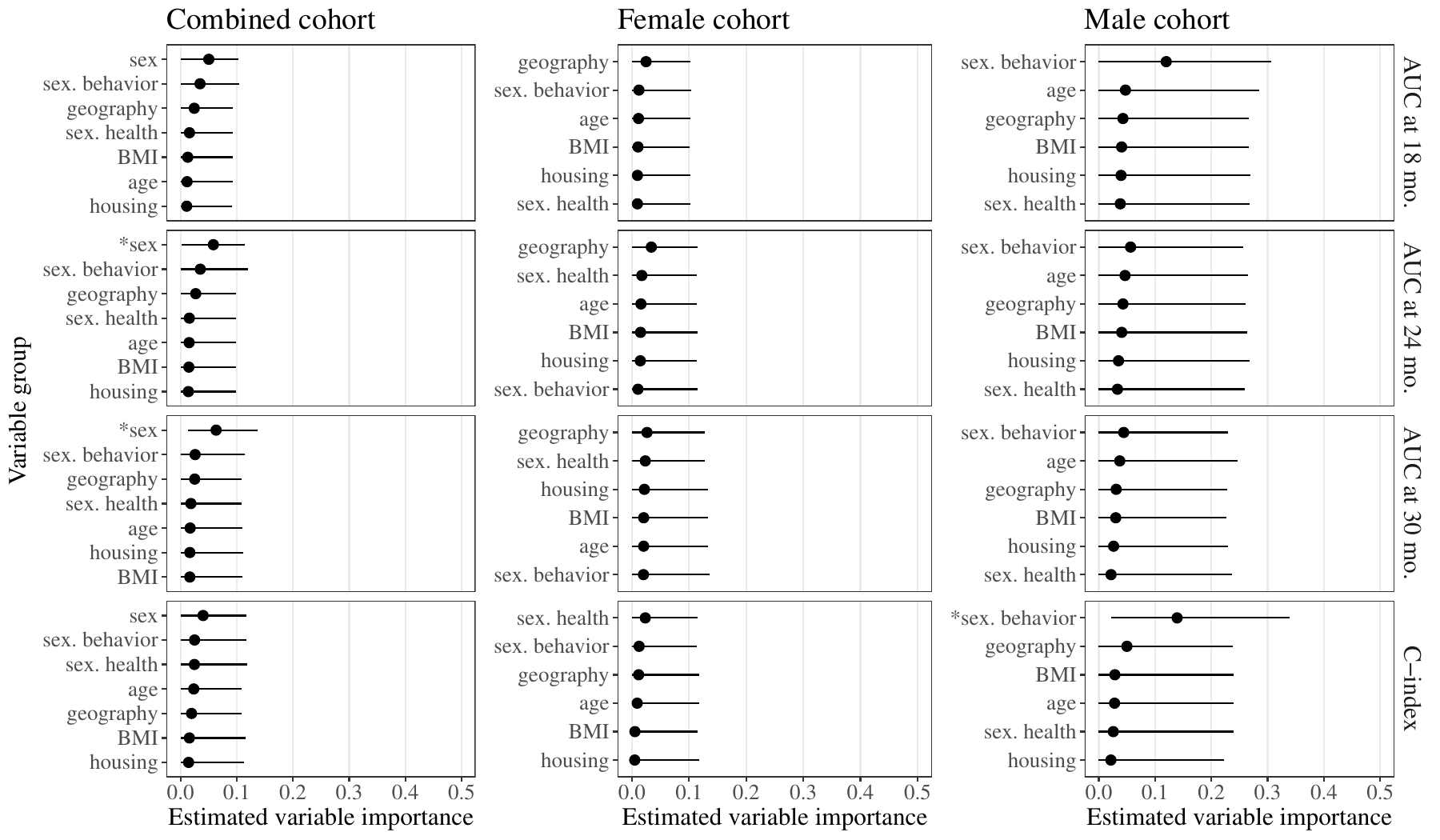}
		\caption{Variable importance in HVTN 702 evaluated relative to the full model that uses all features. Rows correspond, from top to bottom, to AUC\ variable importance evaluated at 18, 24, and 30 months of follow-up, and C-index variable importance. Columns correspond, from left to right, to the combined male and female cohort, female cohort, and male cohort. Feature groups include geographic features, sex assigned at birth (sex), age, body mass index (BMI), sexual health features, sexual behavior features, and housing features. Feature groups for which a one-sided test of zero importance reached nominal significance at a 0.05 level are marked with an asterisk.}
		\label{fig:702 conditional}
	\end{figure}
	
	\section{Robustness}\label{sec:robustness}
	
	\subsection{Comparison of debiasing approaches}\label{subsec:debiasing robustness comp}
	
	Here, we use a simple example to illustrate the lack of robustness of the direct one-step estimator of $v_0$ based on the efficient influence function of $P\mapsto V(f_P, P)$. For ease of exposition, we suppose that $m = 2$ and $V_2(H_0) = 1$, so that $V(f_0, H_0) = V_1(f_0, H_0)$. We also suppose that $f_0$ is fixed and known. As shown in the proof of Theorem \ref{thm:efficient influence function}, the parameter $P \mapsto V(f_0, H_P)$ has efficient influence function at $P_0$ relative to $\mathcal{M}_{\text{obs}}$ equal to 
	\begin{align*}
		\phi_0: z \mapsto 2\left[\iiint \omega((f_0(x),t),(f_0(x_2),t_2))F_0(dt_2\midd x_2)Q_0(dx_2)\left\{F_0(dt \midd x) - \chi_0(z, dt)\right\} - v_0\right].
	\end{align*}
	Given estimators $F_n$ and $G_n$ of $F_0$ and $G_0$, respectively, and setting $Q_n$ to be the empirical distribution of $X$ in the sample, we consider the one-step estimator 
	\begin{align*}
		v_n^{**} := V(f_0, H_n) + \frac{1}{n}\sum_{i=1}^{n}\phi_n(Z_i)
	\end{align*}
	with $\phi_n$ obtained by replacing $F_0$, $G_0$ and $Q_0$ by $F_n$, $G_n$ and $Q_n$, respectively, in the form of $\phi_0$. Here, we do not consider cross-fitting, since it only serves to weaken the regularity conditions under which the one-step estimator is asymptotically linear. We let $F_\infty$ and $G_\infty$ denote the in-probability limits of $F_n$ and $G_n$, respectively, and decompose
	\begin{align*}
		v_n^{**} - v_0 = (\mathbb{P}_n - P_0)\phi_\infty + \left(\mathbb{P}_n - P_0\right)\left(\phi_n - \phi_\infty\right) + V(f_0, P_n) - V(f_0, P_0) + P_0 \phi_n\ .
	\end{align*}
	The leading term on the right-hand side above is linear and mean-zero. The second is an empirical process term that is expected to tend to zero in probability regardless of the values of $F_\infty$ and $G_\infty$. Thus, we are concerned with the linearization remainder term $V(f_0, P_n) - V(f_0, P_0) + P_0 \phi_n$. For notational simplicity, we define $\tilde{\omega}_0:(x_1,t_1,x_2,t_2)\mapsto\omega((f_0(x_1), t_1), (f_0(x_2), t_2))$. In addition, for any conditional event distribution functions $F_1$ and $F_2$ and covariate distribution functions $Q_1$ and $Q_2$, we define 
	\begin{align*}
		\tilde{V}_0(F_1,F_2, Q_1,Q_2) := 2\iiiint\tilde{\omega}_0(x_1,t_1,x_2,t_2)F_1(dt_1\midd x_1)Q_1(dx_1)F_2(dt_2 \midd x_2)Q_2(dx_2)\ .
	\end{align*}
	The remainder term can then be written as 
	\begin{align}
		&V(f_0, P_n) - V(f_0, P_0) + P_0 \phi_n\nonumber\\
		&=-\tilde{V}_0(F_n,F_n,Q_n,Q_n) - \tilde{V}_0(F_0,F_0,Q_0,Q_0) + 2\tilde{V}_0(F_n,F_n,Q_n,Q_0) \nonumber\\
		&\hspace{0.6cm}- 2 \iiiint\tilde{\omega}_0(x,t,x_2,t_2)F_n(dt_2\midd x_2)Q_n(dx_2)\chi_n(z, dt)P_0(dz)\nonumber\\
		&= -\tilde{V}_0(F_n,F_n,Q_n,Q_n) - \tilde{V}_0(F_0,F_0,Q_0,Q_0) + 2\tilde{V}_0(F_n, F_0,Q_n,Q_0) \nonumber\\
		&\hspace{0.6cm}- 2\tilde{V}_0(F_n,F_0,Q_n,Q_0)+ 2\tilde{V}_0(F_n,F_n,Q_n,Q_0) \nonumber\\
		&\hspace{0.6cm}- 2 \iiiint\tilde{\omega}_0(x,t,x_2,t_2)F_n(dt_2\midd x_2)Q_n(dx_2)\chi_n(z, dt)P_0(dz)\nonumber\\
		&= -\tilde{V}_0(F_n,F_n,Q_n,Q_n) - \tilde{V}_0(F_0,F_0,Q_0,Q_0) + 2\tilde{V}_0(F_n, F_0,Q_n,Q_0) \nonumber \\
		&\hspace{0.6cm}- 2 \iiiint\tilde{\omega}_0(x,t,x_2,t_2)F_n(dt_2\midd x_2)Q_n(dx_2)\left\{F_n(dt \midd x) - F_0(dt \midd x) + \chi_n(z, dt)\right\}P_0(dz)\ . \label{eq:non robust}
	\end{align}
	Next, we define the functions
	\begin{align*}
		B^G_{n,0}:(x,t) &\mapsto \frac{G_0(t \midd x)}{G_n(t \midd x)} - 1\ ;  \\
		B^L_{n,0}:(x,t)&\mapsto L_n(t \midd x) - L_0(t \midd x) \ .
	\end{align*}
	Using the Duhamel equation (Theorem 6 of \citealp{Gill1990}), the last summand of (\ref{eq:non robust}) can be shown to equal
	\begin{align*}
		&2\iiiint \Big\{\tilde{\omega}_0(x_1,t_1,x_2,t_2)\int_0^{t_1}\frac{S_0(u^{-} \midd x_1)}{S_n(u \midd x_1)}B^G_{n,0}(x_1,u)B^L_{n, 0}(x_1, du)\\
		&\hspace{6cm} \times  F_n(dt_1 \midd x_1)F_n(dt_2 \midd x_2)Q_0(dx_1)Q_n(dx_2)\Big\}\\
		&\hspace{1cm}+ 2\iiiint\tilde{\omega}_0(x_1,t_1,x_2,t_2)S_0(t_1^-\midd x_1)B^G_{n, 0}(x_1,t_1)B^L_{n,0}(x_1, dt_1)F_n(dt_2 \midd x_2)Q_0(dx_1)Q_n(dx_2)\ .
	\end{align*}
	Roughly speaking, if either $B_{n,0}^G$ or $B_{n,0}^L$ tends to zero in probability, the above display will also tend zero in probability. In contrast, the first three terms of (\ref{eq:non robust}) do not involve the conditional censoring distribution $G_0$. Therefore, while we can expect them to tend to zero in probability provided $F_n$ tends to $F_0$ --- given that $Q_n$ is the empirical distribution, it certainly tends to $Q_0$ --- we cannot expect robustness against inconsistent estimation of $F_0$. 
	
	\subsection{Doubly-robust estimation of the oracle prediction function}
	Here, we present a strategy for doubly-robust estimation of the conditional survival function $F_0(\tau \midd x)$ using pseudo-outcomes. We define $q_1: (y, \delta)\mapsto  \delta + \I(y \geq \tau)$ and $q_2:y\mapsto \min\{y, \tau\}$. The pseudo-outcome depends on the function
	\begin{align*}
		m_{S,G}: (t, x) \mapsto -\frac{1}{S(t \midd x)}\int_t^\infty \I(u \geq \tau)S(du \midd x) = -\frac{\int_{\max\{t,\tau\}}^\infty S(du \midd x)}{S(t \midd x)} =\frac{S(t \vee \tau \midd x)}{S(t \midd x)}\ .
	\end{align*}
	The doubly-robust pseudo-outcome mapping of \citet{Rubin2007} is then defined as
	\begin{align}
		& p_{S,G}: z \mapsto \frac{q_1(y,\delta)\I(y \geq \tau)}{G(q_2(y)\midd x)} + \frac{(1 - q_1(y,\delta))m_{S,G}(q_2(y), x)}{G(q_2(y)\midd x)} - \int_0^{q_2(y)}\frac{m_{S,G}(u, x)\Lambda_G(du \midd x)}{G(u \midd x)}\ . \label{eq:DR pseudo outcome}
	\end{align}
	Under conditional independence of $T$ and $C$ given the measured covariate vector $X$, \citet{Rubin2007} showed that $E_0[p_{S,G}(Z)\midd X = x] = F_0(\tau \midd x)$ $P_0$--almost surely if either $S = S_0$ or $G = G_0$. Then, given estimators $S_n$ and $G_n$ of $S_0$ and $G_0$, the estimated pseudo-outcomes are given by $\{p_{S_n, G_n}(Z_1),  \ldots, p_{S_n, G_n}(Z_n)\}$. We regress these pseudo-outcomes on the feature vectors $\{X_1,  \ldots, X_n\}$ to estimate $F_0(\tau \midd x)$. 
	
	\section{Comparison to permutation variable importance}\label{sec:permutation}
	Permutation-based approaches are widely used to assess the importance of a feature in a fixed prediction algorithm \citep{Breiman2001, fisher2019}. As in the main text, for an index set $s$, we let $\mathbbmsl{P}_{0,s}$ denote the distribution of $(X_{s}^{(2)}, X_{-s}^{(1)}, T^{(1)}, C^{(1)})$, where $(X^{(1)}, T^{(1)}, C^{(1)})$ and $(X^{(2)}, T^{(2)}, C^{(2)})$ are drawn independently from $\mathbbmsl{P}_0$. We note that $\mathbbmsl{P}_{0, s}$ defines a valid joint distribution for $(X, T, C)$ but is not, in general, equal to $\mathbbmsl{P}_0$. For a prediction function $f$, the permutation variable importance of $X_s$ relative to $X$ is given by
	\begin{align}\label{eq:permutation importance}
		\mathbbmsl{V}(f, \mathbbmsl{P}_0) - \mathbbmsl{V}(f, \mathbbmsl{P}_{0, s})\ .
	\end{align}
	This quantity measures the decrease in predictiveness of $f$ when the feature group of interest $X_s$ is sampled independently of the other features and the outcome. In practice, permutation variable importance is typically estimated by (i) assessing the predictiveness of $f$ using the unperturbed data; (ii) permuting the feature or features in $X_s$ across observations; and (iii) assessing the predictiveness of $f$ using the permuted data. The permutation process is often repeated several times, and the results are averaged across permutations; see, for example, \citet{spytek2023survex}.
	
	As apparent from its definition, the permutation variable importance measures a different quantity than the exclusion importance we consider in this work. First, the prediction function $f$ is generally considered to be fixed and given; and second, the decrease in predictiveness corresponds to a comparison between the target population and a population in which the data unit has the same marginal distributions for $X_s$ and $(X_{-s}, T)$ but in which $X_s$ is independent of $(X_{-s}, T)$. To make this second point concrete, we consider the example of the mean squared error predictiveness $\mathbbmsl{V}(f, \mathbbmsl{P}) := E_{\mathbbmsl{P}}[\{f(X) - T\}^2]$. For exclusion importance, the oracle and residual oracle prediction functions are given by the conditional means $\mathbbmsl{f}_0: x \mapsto E_{\mathbbmsl{P}_0}(T \midd X = x)$ and $\mathbbmsl{f}_{0,s}: x \mapsto E_{\mathbbmsl{P}_0}(T \midd X_{-s} = x_{-s})$. The exclusion variable importance is
	\begin{align*}
		\psi_{0,s}^{\text{exclusion}} := E_{\mathbbmsl{P}_0}\left[\left\{\mathbbmsl{f}_0(X) - \mathbbmsl{f}_{0,s}(X)\right\}^2\right] = \int \left\{\mathbbmsl{f}_0(x) - \mathbbmsl{f}_{0,s}(x)\right\}^2\mathbbmsl{Q}_{0}(dx)
	\end{align*}
	Now, we suppose that the prediction function $f$ in (\ref{eq:permutation importance}) is taken to be $\mathbbmsl{f}_0$; this is often implicitly true, as many machine learning algorithms are aimed at estimating the conditional mean of the outcome given features. \citet{hooker2021permutation} show that the permutation importance can be written as
	\begin{align*}
		\psi_{0,s}^{\text{permutation}} = E_{\mathbbmsl{P}_0}\left[\left\{\mathbbmsl{f}_0(X) - T\right\}^2 - \left\{\mathbbmsl{f}_{0,s}^*(X) - T\right\}^2\right] = \int \left\{\mathbbmsl{f}_0(x) - \mathbbmsl{f}_{0,s}^*(x)\right\}^2\mathbbmsl{Q}_{0}(dx), 
	\end{align*}
	where $\mathbbmsl{f}_{0,s}^*: x \mapsto \int \mathbbmsl{f}_0(x)\mathbbmsl{Q}_{0, X_s}(dx_s)$ and we denote by $\mathbbmsl{Q}_{0, X_s}$ the marginal distribution of $X_s$. We note that $\mathbbmsl{f}_{0,s}(x) = \int \mathbbmsl{f}_0(x)\mathbbmsl{Q}_{0, X_s \midd X_{-s}}(dx_s)$, where we denote by $\mathbbmsl{Q}_{0, X_s \midd X_{-s}}$ the conditional distribution of $X_s$ given $X_{-s}$. Therefore, for mean squared error predictiveness, the permutation importance of $X_s$ relative to $X$ for the conditional mean prediction function is equal to the exclusion importance if $X_{s}$ and $X_{-s}$ are independent.
	
	Estimation issues aside, this example suggests that when the features are independent, we may in some cases expect exclusion and permutation variable importance to produce qualitatively similar results. When the features are correlated, the targets of estimation differ, and which is preferred for a given scientific context may be debatable. However, as noted by \citet{hooker2021permutation}, feature correlation can also have a substantial impact on estimation performance for permutation importance, which requires the evaluation of $\mathbbmsl{f}_0$ at values $(x_s, x_{-s})$ that have low probability (in extreme cases, zero) mass under the joint distribution $\mathbbmsl{Q}_{0}$. \citet{hooker2021permutation} study this phenomenon in the context of random forests and neural networks and refer to it as extrapolation bias.
	
	We conducted a simulation study to compare our proposed approach to an existing implementation of the permutation approach for right-censored data, available in the \texttt{survex} software package \citep{spytek2023survex}. The general simulation setup was as described in the experiments reported in the main text. We set $p$, the dimension of $X$, to 4, and generated $T$ and $C$ from the log-normal accelerated failure time models
	\begin{align*}
		\log T = \tfrac{10}{20}x_1 - \tfrac{6}{20}x_2 - \tfrac{3}{20}x_3+ \varepsilon_T\ ,\hspace{0.2in} \log C = \varepsilon_C\ ,
	\end{align*}
	where $\varepsilon_T$ and $\varepsilon_C$ were independent standard normal random variables. We set $\Sigma$, the covariance matrix of $X$, to 
	\begin{align*}
		\Sigma = \begin{pmatrix}
			1 & 0 & 0 & \Sigma_{14}\\
			0 & 1 & 0 & 0 \\
			0 & 0 & 1 & 0\\
			\Sigma_{14} & 0 & 0 & 1
		\end{pmatrix}.
	\end{align*}
	We varied $\Sigma_{14}$ over the values $\{0, 0.3, 0.6, 0.9\}$, representing scenarios in which $X_1$ and $X_4$ varied from uncorrelated to strongly correlated. We note that in this simulation setting, $T$ and $C$ were independent; the permutation importance estimation procedure implemented in \texttt{survex} uses Kaplan-Meier IPCW weights to estimate model performance and, unlike our proposed approach, was not designed to handle informative censoring. 
	
	We considered estimation of the importance of $X_4$ relative to the full covariate vector $X$. For $X_4$, both $\psi_{0,s}^{\text{exclusion}}$ and $\psi_{0,s}^{\text{permutation}}$ are zero regardless of the value of $\Sigma_{14}$, since the conditional survival function of $T$ given $X$ does not depend on $X_4$. We considered importance based on Brier score predictiveness evaluated at landmark time $\tau = 0.5$. For our proposed method, we used random survival forests for nuisance estimation, with five-fold cross-fitting. We used Super Learner \citep{vanderlaan2007} to estimate the residual oracle prediction function. For \texttt{survex}, we used the \texttt{model\_parts()} function with 10 permutations. For both approaches, five-fold cross-validation was used to select the tuning parameters for the random forests algorithm, as described in Section \ref{sec:nuisance estimation details}.
	
	We generated 500 datasets of size $n \in \{500, 1000, \ldots, 3000\}$. For each simulation setting, we evaluated the performance of each method by calculating the proportion of simulation replicates in which $X_4$ was correctly ranked as the least important of the four features.
	
	In Fig. \ref{fig:survex_comparison}, we show the results of this simulation study. We note that when $\Sigma_{14} = 0$, both the permutation and exclusion approaches tend to correctly rank $X_4$ as the least important feature as the sample size increases. The exclusion approach demonstrates similar behavior in each of the other correlation settings, while the permutation approach appears to be negatively affected by moderate to strong correlation between $X_4$ and $X_1$. When $\Sigma_{14} = 0.9$, the correct ranking proportion of the permutation approach does not appear to increase with sample size. The results of this experiment suggest that the empirical performance of our proposed approach is insensitive to the presence of correlation between features, while existing permutation importance methods can be negatively affected by correlated features.
	
	\begin{figure}
		\centering
		\includegraphics[width=\linewidth]{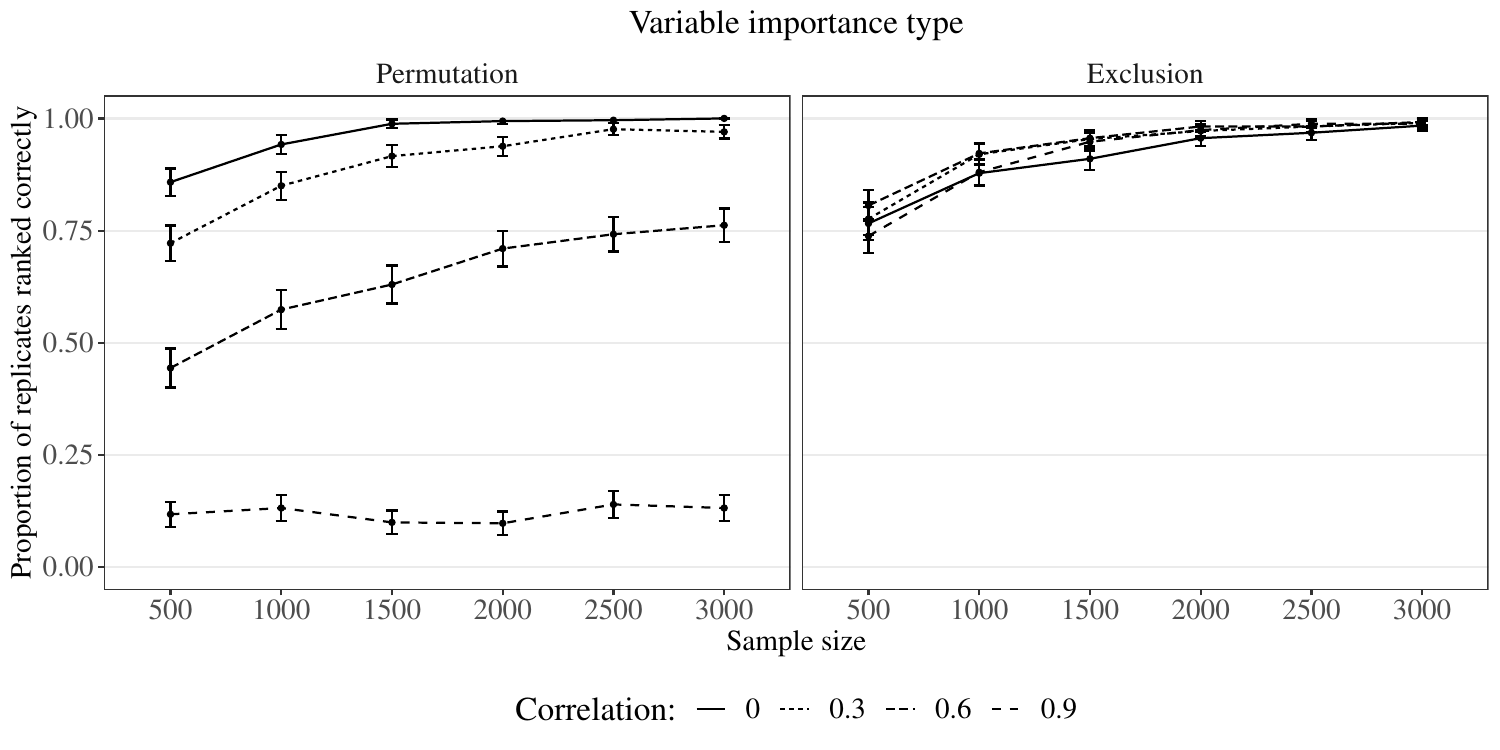}
		\caption{Comparison of permutation and exclusion variable importance methods for correctly ranking $X_4$, which had zero importance under both paradigms. For both methods, random survival forests were used to estimate the outcome conditional survival function. Different line types indicate varying levels of correlation between $X_4$ and $X_1$. Vertical bars represent 95\% confidence intervals taking into account Monte Carlo error.}
		\label{fig:survex_comparison}
	\end{figure}

\end{document}